\documentclass[english,a4paper,cleveref,thm-restate,hidecoq]{llncs}

\usepackage{amsmath} 
\usepackage{relsize} 
\usepackage{stmaryrd} 
\usepackage{amssymb} 
\usepackage{mathtools} 
\usepackage{prooftree} 

\usepackage[dvipsnames,svgnames]{xcolor}

\usepackage[frozencache,cachedir=.]{minted}
\usepackage{listings}

\usepackage{ifthen}

\newif\ifhidecoq

\usepackage{enumerate} 
\usepackage{graphics} 
\usepackage{hyperref}

\usepackage{hyperref}
\usepackage{academicons}
\usepackage{xcolor}

\newcommand{\orcid}[1]{\href{https://orcid.org/#1}{\textcolor[HTML]{A6CE39}{\aiOrcid}}}

\RequirePackage{tikz}
\RequirePackage{color}

\usetikzlibrary{calc,arrows,shapes,decorations.pathmorphing,backgrounds,positioning,fit}

\tikzstyle{distr} = [circle,fill=black!20,draw=black,thick,minimum size=2mm,
                     node distance=10mm and 12mm,inner sep=0pt]
\tikzstyle{dot} = [circle,fill=black!20,draw=black,thick,minimum size=1mm,
                     node distance=10mm and 12mm,inner sep=1pt]

\tikzstyle{state} = [rectangle,rounded corners,draw=black,thick,
                     minimum size=5mm, node distance=10mm and 12mm,
                     inner sep=3pt]
\tikzstyle{legent} = [node distance=10mm and 12mm, inner sep=2pt]
\tikzstyle{action} = [auto]
\tikzstyle{from} = [<->, shorten <=1pt, >=stealth',semithick]
\tikzstyle{timeto} = [->>, shorten >=1pt, >=stealth',semithick]
\tikzstyle{to} = [->, shorten >=1pt, >=stealth',semithick]
\tikzstyle{todistr} = [-, shorten >=1pt, >=stealth',semithick]
\tikzstyle{distrto} = [->, decorate, decoration={snake,pre length=1mm,post length=1mm}, shorten >=1pt, >=stealth',semithick]
\tikzstyle{tosqunder} = [to,rounded corners,swap,
                         to path={ (\tikztostart.south)
                                   -- ++(0,-.5)
                                   -- ($(\tikztotarget.south) + (0,-.5)$)
                                   \tikztonodes
                                   -| (\tikztotarget.south) }]
\tikzstyle{loop wnw} = [loop,looseness=7,out=185,in=175]
\tikzstyle{loop nee} = [bend right,looseness=5,out=135,in=100]
\tikzstyle{loop ene} = [bend left,looseness=5,out=-80,in=-45]
\tikzstyle{loop ese} = [loop,looseness=7,out=5,in=-5]
\tikzstyle{loop sse} = [loop,looseness=12,out=260,in=320]

\newcommand{\TODO}[1]{{\color{red} TODO: #1}}

\newcommand{\Names}{\ensuremath{\mathcal{N}}\xspace}
\renewcommand{\aa}{a}
\newcommand{\ab}{b}

\newcommand{\co}[1]{\ensuremath{\overline{#1}}}

\newcommand\Act{\ensuremath{\mathsf{Act}}\xspace}
\newcommand\Actfin{\ensuremath{\mathsf{Act}^\star}\xspace}

\newcommand\Acttau{\ensuremath{\mathsf{Act}_\tau}\xspace}

\newcommand{\BNFsep}{\;\;|\;\;}

\newcommand{\Nil}{\mathop{\textsf{0}}}
\newcommand{\Unit}{\mathop{\textsf{1}}}
\newcommand{\mailbox}[1]{\ensuremath{#1}}
\newcommand{\out}[1]{\mailbox{\ensuremath{\overline{#1}}}}

\newcommand{\extc}{\mathrel{+}}

\newcommand{\Par}{\mathrel{\parallel}}

\newcommand{\rec}[2][x]{\mathsf{rec} #1. #2}

\newcommand\CCS{\textsc{CCS}\xspace}
\newcommand\ACCS{\texttt{ACCS}\xspace}
\newcommand\TACCS{\texttt{TACCS}\xspace}

\newcommand{\csys}[2]{#1 \mathrel{\llceil} #2}
\newcommand\pl[1]{{\color{violet}{pl: #1}\xspace}}

\newcommand{\rfig}[1]{Figure~\ref{fig:#1}}
\newcommand{\rpt}[1]{part~(\ref{pt:#1})}

\newcommand{\rtab}[1]{Table~\ref{tab:#1}}
\newcommand{\rsec}[1]{Section~\ref{sec:#1}}
\newcommand{\rapp}[1]{Appendix~\ref{sec:#1}}

\newcommand{\rexa}[1]{Example~\ref{ex:#1}}

\newcommand{\rdef}[1]{Definition~\ref{def:#1}}

\newcommand{\rprop}[1]{Proposition~\ref{prop:#1}}
\newcommand{\rthm}[1]{Theorem~\ref{thm:#1}}
\newcommand{\rcor}[1]{Corollary~\ref{cor:#1}}
\newcommand{\rlem}[1]{Lemma~\ref{lem:#1}}
\newcommand{\rptlem}[2]{Lemma~\ref{lem:#1}(\ref{pt:#2})}
\newcommand{\req}[1]{Equation~(\ref{eq:#1})}

\newcommand{\rulename}[1]{{\footnotesize{\textsc{[#1]}}}}
\newcommand{\rname}[1]{\rulename{#1}}

\newcommand{\wtrefl}{\rname{wt-refl}\xspace}
\newcommand{\wttau}{\rname{wt-tau}\xspace}
\newcommand{\wtmu}{\rname{wt-mu}\xspace}

\newcommand{\mnow}{\rname{axiom}\xspace}
\newcommand{\mstep}{\rname{ind-rule}\xspace}

\newcommand{\parL}{\rname{Par-L}\xspace}
\newcommand{\parR}{\rname{Par-R}\xspace}
\newcommand{\com}{\rname{Com}\xspace}

\newcommand{\mboxelim}{\rname{Mb-Out}\xspace}

\newcommand{\rinput}{\rname{Input}}
\newcommand{\rtau}{\rname{Tau}}
\newcommand{\unfold}{\rname{Unf}}
\newcommand{\extL}{\rname{Sum-L}}
\newcommand{\extR}{\rname{Sum-R}}

\newcommand{\scom}{\rname{S-com}}

\newcommand{\stauserver}{\rname{S-Srv}}
\newcommand{\stauclient}{\rname{S-Clt}}

\newcommand{\stproclift}{\rname{L-Proc}\xspace}
\newcommand{\stminplift}{\rname{L-Minp}\xspace}
\newcommand{\stmoutlift}{\rname{L-Mout}\xspace}
\newcommand{\stcommlift}{\rname{L-Comm}\xspace}

\newcommand{\cnvepsilon}{\rname{cnv-epsilon}}
\newcommand{\cnvmu}{\rname{cnv-mu}}

\newcommand{\axiom}[1]{\textsc{#1}\xspace}

\newcommand{\outputcommutativity}{\axiom{Output-commutativity}}
\newcommand{\outputconfluence}{\axiom{Output-confluence}}
\newcommand{\outputdeterminacy}{\axiom{Output-determinacy}}
\newcommand{\outputfeedback}{\axiom{Feedback}}
\newcommand{\outputtau}{\axiom{Output-tau}}
\newcommand{\outputdeterminacyinv}{\axiom{Backward-output-determinacy}}

\newcommand{\restrictedinputcommutativity}{\axiom{Input-commutativity}}

\newcommand{\fwdfeedback}{\axiom{Fwd-Feedback}}
\newcommand{\boom}{{\sc Input-Boomerang~}} 
\newcommand{\accSym}{\mathcal{A}}

\newcommand{\Fwd}{\mathcal{F}}

\newcommand{\fw}{\mathsf{fw}}

\newcommand{\sta}[1]{\ensuremath{\mathrel{\overset{#1} \longrightarrow_\fw}}}
\newcommand{\Nsta}[1]{\ensuremath{\mathrel{\overset{#1}{\longarrownot\longrightarrow_\fw}}}}

\newcommand{\wta}[1]{\ensuremath{\mathrel{\overset{#1} \Longrightarrow_\fw}}}

\newcommand\acnvalong{\Downarrow}

\newcommand\set[1]{\{ #1 \}}
\newcommand\setof[2]{\{ #1 \ |\  #2 \}}

\newcommand\parts[1]{\mathcal{P}(#1)}
\newcommand\pparts[1]{\mathcal{P}^{+}(#1)}

\newcommand{\mset}[1]{\{\!|#1|\!\}}

\newcommand\emptyMset{\varnothing}

\newcommand{\cardinality}[1]{\mid #1 \mid}

\newcommand{\modulo}[2]{{#1}_{#2}}

\newcommand{\coqbasepath}[1]{https://shouldnothappen.com/must/src/#1}
\newcommand{\coqpic}{\includegraphics[scale=0.1]{index}}
\newcommand{\coqlink}[1]{(\href{\coqbasepath{#1}}{\coqpic})}

\newcommand{\myshow}[1]{}

\newcommand\coqLTS[1]{\myshow{\coqlink{Must.LTS.html\##1}}}

\newcommand\coqMT[1]{\myshow{\coqlink{Must.MustT.html\##1}}}

\newcommand\coqCom[1]{\myshow{\coqlink{Must.Completeness.html\##1}}}

\DeclareMathOperator{\opMust}{\textsc{must}}
\DeclareMathOperator{\opMay}{\textsc{may}}

\newcommand{\testname}[2]{\ensuremath{\mathrel{{#1}^{\text{\smaller \ensuremath{#2}}}}}}

\newcommand{\May}[1][]{\testname{\opMay}{#1}}

\newcommand{\opMusti}{\opMust\ensuremath{_i}}
\newcommand{\musti}[2]{\ensuremath{#1 \opMusti #2}}
\newcommand{\Nmusti}[2]{\ensuremath{#1 \centernot{\opMusti} #2}}

\newcommand{\Must}[2]{\ensuremath{#1 \opMust #2}}

\newcommand{\opMustset}{\opMust_{\textsf{aux}}}

\newcommand{\mustset}[2]{\ensuremath{#1 \opMustset #2}}

\newcommand{\sqsubsetsim}{\vcenter{\offinterlineskip\hbox{$\sqsubset$}\vskip 0.2ex\hbox{$\sim$}}}
\newcommand{\Nsqsubsetsim}{\centernot{\vcenter{\offinterlineskip\hbox{$\sqsubset$}\vskip 0.2ex\hbox{$\sim$}}}}
\NewDocumentCommand{\testleq}{O{}O{}}{\ensuremath{\mathrel{\sqsubsetsim_{\text{\smaller #1}}^{\text{\smaller #2}}}}}
\newcommand{\testeq}[1][]{\ensuremath{\mathrel{\eqsim_{\text{#1}}}}}
\newcommand{\Ntestleq}[1][]{\mathrel{\Nsqsubsetsim_{\kern-3pt\text{#1}}}}

\newcommand{\testleqS}{\testleq[\ensuremath{\opMust}]}
\newcommand{\testleqSset}{\testleq[\ensuremath{\opMust}][Set]}

\newcommand{\testeqS}{\testeq[\ensuremath{\opMust}]}

\newcommand{\ok}{\ensuremath{\checkmark}}

\newcommand{\bhveq}[1][]{\mathrel{\eqsim_{\mathsf{bhv}}}}

\newcommand{\asynleq}[1][]{\mathrel{\preccurlyeq^{#1}_{\mathsf{asyn}}}}

\newcommand{\asyneq}[1][]{\mathrel{\eqsim_{\mathsf{asyn}}}}

\newcommand{\bhvleqone}{\mathrel{\preccurlyeq_{\mathsf{cnv}}}}
\newcommand{\bhvleqtwo}{\mathrel{\preccurlyeq_{\mathsf{acc}}}}

\newcommand{\asynleqone}{\asynleq[\textsf{cnv}]}

\newcommand{\asynleqtwo}{\asynleq[\textsf{acc}]}

\newcommand{\altleq}{\preccurlyeq_{\mathit{alt}}}

\newcommand\accfwp[3]{\mathcal{A}_{\fw}(#1,#2,#3)}

\newcommand\accht[2]{\mathcal{A}_{\fw}(#1,#2)}

\newcommand{\asleq}{\mathrel{\preccurlyeq_{\mathsf{AS}}}}

\newcommand{\asleqNF}{\mathrel{\preccurlyeq^{\mathsf{NF}}_{\mathsf{AS}}}}

\newcommand{\msleq}{\mathrel{\preccurlyeq_{\mathsf{MS}}}}

\newcommand{\coindleq}{\mathrel{\preccurlyeq_{\mathsf{co}}}}
\newcommand{\msleqNF}{\mathrel{\preccurlyeq^{\mathsf{NF}}_{\mathsf{MS}}}}

\renewcommand{\and}{\text{ and }}
\renewcommand{\lor}{\text{ or }}
\renewcommand{\implies}{\text{ implies }}
\newcommand{\imply}{\text{ imply }}

\newcommand{\forevery}{\text{for every }}
\newcommand{\Forevery}{\text{For every }}

\newcommand{\wehavethat}{. \;}
\newcommand{\suchthat}{\wehavethat}
\usepackage[utf8]{inputenc}
\usepackage{xspace}
\usepackage{tikz}
\usetikzlibrary{positioning} 
\usepackage{tikz-cd} 
\usepackage{centernot}
\usepackage{cancel}
\usepackage{dsfont}
\usepackage[normalem]{ulem}
\usepackage{mdframed}
\mdfsetup{skipabove=1em,skipbelow=1em}

\usepackage{yfonts} 
\usepackage[normalem]{ulem} 

\newcommand{\pierre}{\text{{\em Pierre}}\xspace}

\newcommand{\disjoint}[2]{#1 \mathop{\#} #2}

\newcommand{\trace}{s}
\newcommand{\traceA}{t}
\newcommand{\stateA}{s}
\newcommand{\stateB}{s'}
\newcommand{\stateC}{s''}

\newcommand{\N}[1][]{\mathbb{N}_{#1}}

\newtheorem{counterexample}[definition]{Counterexample}
\newtheorem{myaxiom}[definition]{Axiom}

\newcommand{\len}[1]{\mathsf{len}(#1)}

\newcounter{thm}
\newcommand{\thistheoremname}{}
\newtheorem{genericlem}[thm]{\thistheoremname}

\newcommand{\koenigslemma}{K\H{o}nig's lemma\xspace}

\newcommand{\accP}[3]{\accSym(#1,#2,#3)}
\newcommand{\acc}[2]{\accSym(#1,#2)}

\newcommand\gas[3]{\mathcal{GA}(#1,#2,#3)}

\newcommand{\leaveout}[1]{}

\newcommand{\eqdef}{\mathrel{\stackrel{\mathsf{def}}{=}}}

\ifdefined\conv
\renewcommand{\conv}{\ensuremath{\downarrow}}
\else
\newcommand{\conv}{\ensuremath{\downarrow}}
\fi
\newcommand{\Conv}{\ensuremath{\Downarrow}}

\newcommand{\cnvalong}{\ensuremath{\mathrel{\Conv}}}
\newcommand{\convi}{\ensuremath{\conv_i}}

\newcommand\defsrc[1]{}

\makeatletter
\newcommand{\subsetsim}{\mathrel{\mathpalette\subset@sim\relax}}
\newcommand{\subset@sim}[2]{%
  \vtop{\offinterlineskip\m@th
    \ialign{\hfil##\cr
      $#1\sqsubset$\cr\noalign{\kern0.5pt}\scalebox{0.9}{$#1\sim$}\cr
   }%
  }%
}
\makeatother

\newcommand{\st}[1]{\ensuremath{\mathrel{\overset{#1}\longrightarrow}}}
\newcommand{\stx}[1]{\xrightarrow{#1}}
\newcommand{\Nst}[1]{\ensuremath{\mathrel{\overset{#1}{\longarrownot\longrightarrow}}}}
\newcommand{\stable}{\Nst{\tau}}

\newcommand{\wt}[1]{\ensuremath{\mathrel{\overset{#1}
      \Longrightarrow}}}

\newcommand{\Nwt}[1]{\ensuremath{\mathrel{\overset{#1}{\longarrownot\Longrightarrow}}}}

\newcommand{\after}[2]{#1 \mathrel{\mathsf{after}} #2}

\newcommand\StatesA{A}
\newcommand\StatesB{B}
\newcommand\StatesC{C}
\newcommand\States{\StatesA}

\newcommand{\SysStates}{\ensuremath{S}}
\newcommand{\sysstate}{\ensuremath{s}}

\newcommand{\state}{p}
\renewcommand{\stateA}{p'}
\renewcommand{\stateB}{q}

\newcommand{\subst}[2]{[^{#1}/_{#2}]}

\newcommand\rr[1]{\ensuremath{\stackrel{#1}{\rightsquigarrow}}}

\newcommand{\gen}[2]{g(#1, #2)}

\newcommand{\testconvSym}{\mathit{tc}} 
\newcommand{\testconv}[1]{\testconvSym(#1)} 
\newcommand{\testaccSym}{\mathit{ta}} 
\newcommand{\testacc}[2]{\testaccSym(#1,#2)} 
\newcommand{\tacc}[2]{\testacc{#1}{#2}} 

\newcommand{\genlts}{\mathcal{L}}
\newcommand{\lts}[3]{\ensuremath{\langle #1, #2, #3 \rangle}}

\newcommand{\reducts}[3]{\setof{ \stateA \in #2 }{ #1 #3{\tau} \stateA } }

\newcommand{\outactions}[1]{\setof{ \co{\aa} \in \co{\Names }}{ #1 \st{ \co{\aa}} } }
\newcommand{\inactions}[1]{\setof{ \aa \in \Names }{ #1 \st{  \aa} } }

\newcommand{\I}{I}

\newcommand{\chopSym}{\mathsf{nf}}
\newcommand{\chop}[1]{\chopSym(#1)}

\newcommand{\LTSPar}{\Par}

\newcommand{\client}{r}
\newcommand{\server}{p}

\newcommand{\serverA}{p}
\newcommand{\serverB}{q}

\newcommand{\goodSym}{\textsc{good}}
\newcommand{\good}[1]{\goodSym(#1)}
\newcommand{\im}[2]{\mathit{IM}(#1, #2)}

\newcommand{\MI}{\mathit{MI}}
\newcommand{\MO}{\mathit{MO}}

\newcommand{\oba}{\textsc{OBA}}

\newcommand{\obaFB}{\textsc{Fdb}}
\newcommand{\obaFW}{\textsc{Fwd}}

\newcommand{\liftFWSym}{\textsc{FW}}
\newcommand{\liftFW}[1]{\liftFWSym(#1)}

\newcommand{\mustpreorder}{$\opMust$-preorder\xspace}
\newcommand{\mustequivalence}{$\opMust$-equivalence\xspace}

\newcommand{\svrclt}{client-server\xspace}

\newcommand{\nondeterminism}{nondeterminism\xspace}
\newcommand{\nondeterministic}{nondeterministic\xspace}

\newcommand{\MustSet}{$\opMust$-set\xspace}
\newcommand{\MustSets}{$\opMust$-sets\xspace}

\newcommand{\AcceptanceSets}{acceptance sets\xspace}

\newcommand{\barinduction}{bar induction\xspace}
\newcommand{\Barinduction}{Bar induction\xspace}

\newcommand{\intentional}{intensional\xspace}
\newcommand{\intentionally}{intensionally\xspace}
\newcommand{\Intentional}{Intensional\xspace}

\newcommand{\extensional}{extensional\xspace}

  \newcommand{\stripSym}{\mathsf{strip}}
\newcommand{\strip}[1]{\stripSym(#1)}

\newcommand{\ie}{{\em i.e.}\xspace}
\newcommand{\sts}[2]{\ensuremath{\langle #1, #2 \rangle}}
\newcommand{\myspace}{\phantom{\scalebox{.6}{$\ok$}}}
\newcommand{\ltsof}[1]{\ensuremath{\textsc{lts}(#1)}}

\renewcommand{\traceA}{s_1}
\newcommand{\traceB}{s_2}
\newcommand{\traceC}{s_3}

\DeclareUnicodeCharacter{2208}{$\in$}
\DeclareUnicodeCharacter{2203}{$\exists$}
\DeclareUnicodeCharacter{2200}{$\forall$}
\DeclareUnicodeCharacter{2113}{$\ell$}
\DeclareUnicodeCharacter{03B1}{$\alpha$}
\DeclareUnicodeCharacter{27F6}{$\longrightarrow$}
\DeclareUnicodeCharacter{2227}{$\wedge$}
\DeclareUnicodeCharacter{2228}{$\vee$}
\DeclareUnicodeCharacter{2261}{$\equiv$}
\DeclareUnicodeCharacter{3BC}{$\mu$}
\DeclareUnicodeCharacter{2260}{$\neq$}
\DeclareUnicodeCharacter{27F9}{$\Longrightarrow$} 
\DeclareUnicodeCharacter{03B7}{$\eta$} 
\DeclareUnicodeCharacter{03BB}{$\lambda$} 
\DeclareUnicodeCharacter{03C4}{$\tau$} 
\DeclareUnicodeCharacter{21D3}{$\Downarrow$}
\DeclareUnicodeCharacter{227C}{$\preccurlyeq$} 
\DeclareUnicodeCharacter{2081}{${}_1$} 
\DeclareUnicodeCharacter{2082}{${}_2$} 
\DeclareUnicodeCharacter{2091}{${}_{\textrm{e}}$} 
\DeclareUnicodeCharacter{219B}{$\nrightarrow$} 
\DeclareUnicodeCharacter{2286}{$\subseteq$} 
\DeclareUnicodeCharacter{2291}{$\sqsubseteq$} 
\DeclareUnicodeCharacter{2205}{$\emptyset$} 
\DeclareUnicodeCharacter{227E}{$\precsim$} 
\DeclareUnicodeCharacter{25B7}{$\triangleright$} 
\DeclareUnicodeCharacter{2194}{$\leftrightarrow$} 
\DeclareUnicodeCharacter{2250}{$\doteq$} 
\DeclareUnicodeCharacter{228E}{$\uplus$} 
\DeclareUnicodeCharacter{27FF}{$\rightsquigarrow$} 
\DeclareUnicodeCharacter{22CD}{$\simeq$} %
\DeclareUnicodeCharacter{2216}{$\setminus$} %
\DeclareUnicodeCharacter{2093}{$_x$} %
\DeclareUnicodeCharacter{2913}{$\downarrow_i$} %
\DeclareUnicodeCharacter{22D6}{$\lessdot$} %
\DeclareUnicodeCharacter{2AB7}{$\precapprox$} %
\DeclareUnicodeCharacter{2A7D}{$\leq$} %

\newcommand{\outof}[1]{O(#1)}

\newcommand{\outputmultisetSym}{\mathsf{mbox}}
\newcommand{\outputmultiset}[1]{\outputmultisetSym(#1)}

\newcommand{\WD}{\mathit{WD}}

\newcommand{\msetnow}{\rname{Mset-now}}
\newcommand{\msetstep}{\rname{Mset-step}}

\newcommand{\off}[1]{}

\renewcommand{\blacksquare}{~}

\title{Constructive characterisations of the
  \texorpdfstring{\mustpreorder}{MUST-preorder} for asynchrony}

\institute{%
  Université Paris Cité, CNRS, IRIF, F-75013, Paris, France 
  \and INRIA, Université Côte d'Azur, France 
  \and Nomadic Labs, Paris, France 
  \and University of Cambridge, United Kingdom 
}

\author{%
  Giovanni Bernardi\inst{1}\orcidID{0009-0008-3653-3040}\and
  Ilaria Castellani\inst{2}\orcidID{0000-0001-9820-0892}\and
  Paul Laforgue\inst{1,3}\orcidID{0009-0009-6688-4850}\and
  Léo Stefanesco\inst{4}\orcidID{0000-0002-4719-2922}\thanks{Work done at
  MPI-SWS, Germany.}
}

\hypersetup{
  bookmarks=true,     
  bookmarksopen=true, 
}

\usepackage{fancyhdr}
\begin{document}
\maketitle

\begin{abstract}
  De Nicola and Hennessy's \mustpreorder is a liveness preserving
  refinement which states that a server~$\serverB$ refines a
  server~$\serverA$ if all clients satisfied by~$\serverA$ are also
  satisfied by~$\serverB$. Owing to the universal quantification over
  clients, this definition does not yield a practical proof method,
  and alternative characterisations are necessary to reason over
  it.
Finding these characterisations for asynchronous semantics, i.e. where
outputs are non-blocking, has thus far proven to be a challenge,
usually tackled via ad-hoc definitions.

We show that the standard characterisations of the \mustpreorder carry over as
they stand to asynchronous communication, if servers are enhanced to act as
forwarders, \ie they can input any message as long as they store it back into
the shared buffer.
Our development is constructive, is
completely mechanised in Coq, and is
independent of any calculus: our results pertain to Selinger output-buffered
agents with feedback. This is a class of Labelled Transition Systems that
captures programs that communicate via a shared unordered buffer, as in
asynchronous \CCS or the asynchronous $\pi$-calculus.
We show that the standard coinductive characterisation lets us prove in Coq that
concrete programs are related by the \mustpreorder.
Finally, our proofs show that Brouwer's \barinduction\ principle is a useful
technique to reason on liveness preserving program transformations.
\end{abstract}

\section{Introduction}
\label{sec:intro}
Code refactoring is a routine task to develop or update software, and
it requires methods to ensure that a program~$\serverA$ can be safely
replaced by a program~$\serverB$.  One way to address this issue is
via refinement relations, \ie preorders.  For programming languages,
the most well-known one is Morris \emph{extensional} preorder
\cite[pag.~$50$]{morris}, defined by letting~$p \leq q$ if for all
contexts~$C$, whenever~$C[p]$ reduces to a normal form~$N$,
then~$C[q]$ also reduces to~$N$.

\paragraph{Comparing servers}
This paper studies a version of Morris preorder for
\nondeterministic asynchronous \svrclt systems.
In this setting it is natural to reformulate the preorder by replacing
reduction to normal forms (\ie termination) with a suitable
\emph{liveness} property.
Let~$\csys{ \server }{ \client }$ denote a {\em \svrclt\ system},
that is a parallel composition in which the identities of the
server~$\server$ and the client~$\client$ are distinguished, and
whose computations have the form
$
\csys{\server~}{~\client} =
\csys{ \server_0 }{ \client_0 } \st{ }
\csys{ \server_1 }{ \client_1 } \st{ }
\csys{ \server_2 }{ \client_2 } \st{ } \ldots,
$
where each step represents either an internal computation
of one of the two components, or an interaction between them.
Interactions correspond to handshakes, where
two components ready to perform matching input/output actions
advance together.
We express liveness by saying that  $\server \text{ \emph{must pass} }
\client$, denoted $\Must{ \server }{ \client }$, if in every maximal
computation of $\csys{ \server }{ \client
}$ there exists a state $\csys{ \server_i}{ \client_i}$ such that
$\good{\client_i}$, where~$\goodSym$ is a decidable predicate
 indicating that the client has reached a successful state.
Servers are then compared according to their capacity to
satisfy clients, \ie via contexts of the form~$\csys{[-]}{\client}$
and the predicate~$\opMust$.
Morris preorder 
then becomes the \mustpreorder\
by De Nicola and Hennessy~\cite{DBLP:journals/tcs/NicolaH84} :
$
  \serverA \testleqS \serverB \text{ when } \forall \client \wehavethat
  \Must{\serverA}{\client} \implies
  \Must{\serverB}{\client}.
  $

\paragraph{Advantages}
  The \mustpreorder is by definition liveness preserving,
  because $\Must{ \server }{ \client }$ literally means that
  ``in every execution something good must happen (on the client
  side)''.  Results on~$\testleqS$ thus shed light on
  liveness-preserving program transformations.

  The~\mustpreorder is independent of any particular calculus,
  as its definition requires simply
  (1) a reduction semantics for the parallel composition
  $\csys{ \server }{ \client }$, and (2)
  a predicate $\goodSym$ over programs.
  Hence~$\testleqS$
  may relate servers written in different languages. For instance, servers written in
  \textsc{OCaml} may be compared to servers written in \textsc{Java}
  according to clients written in \textsc{Python}, because all
 these languages use the same basic protocols for communication.

\paragraph{Drawback}
  The definition of the \mustpreorder is {\em contextual}: proving
  $\serverA~\testleqS~\serverB$ requires analysing an {\em
  infinite} amount of clients, and so the definition of
the preorder does not entail an effective proof method.
A solution to this problem is to define an {\em alternative (semantic)
  characterisation} of the preorder~$\testleqS$, \ie a
preorder~$\altleq$ that coincides with~$\testleqS$
and does away with the universal quantification over clients (\ie contexts).
In {\em synchronous} settings, i.e. when both input and output
actions are blocking, such alternative characterisations have been thoroughly
investigated, typically via a behavioural approach based on labelled transition
systems.

\begin{figure}[t]
  \hrulefill
  \begin{center}
    \begin{minipage}{4cm}
        \centering
      \begin{tikzpicture}
        \node[state,scale=0.8] (s1) at (6,0) {$\server_0$};
        \node[state,scale=0.8,below of=s1,left of=s1] (s2) {$\server_1$};
        \node[state,scale=0.8,below of =s1,right  of=s1] (s3) {$\server_2$};
        \node[state,scale=0.8,below of=s2] (s4) {$\server_3$};
        \node[state,scale=0.8,below of=s3] (s5) {$\server_4$};

        \node[scale=0.8, below of = s5] (dummy) {$$};

        \path[->]
        (s1) edge node [above left,scale=0.8] {$\texttt{str}$} (s2)
        (s1) edge node [above right,scale=0.8] {$\texttt{float}$} (s3)
        (s2) edge node [left,scale=0.8] {$\co{\texttt{int}}$} (s4)
        (s3) edge node [right,scale=0.8] {$\co{\texttt{long}}$} (s5);
      \end{tikzpicture}
      \end{minipage}%
      \begin{minipage}{4cm}
        \centering\vskip-0.95em
      \begin{tikzpicture}
        \node[state,scale=0.8] (s0) at (0,0) {$\client_0$};
        \node[state,dashed,scale=0.8,below right = +15pt and +15pt of s0] (s1) {$\client_2$};
        \node[state,scale=0.8,below left = +15pt and +15pt of s0] (s2) {$\client_1$};
        \node[state,dashed,scale=0.8,below right = +15pt and +15pt of s2] (s3) {$\client_3$};

        \path[->]
        (s0) edge node [above left,scale=0.8] {$\co{\texttt{str}}$} (s2)
        (s0) edge node [above right,scale=0.8] {$\texttt{int}$} (s1)
        (s2) edge node [below left,scale=0.8] {$\texttt{int}$} (s3)
        (s1) edge node [below right,scale=0.8] {$\co{\texttt{str}}$} (s3);
    \end{tikzpicture}
  \end{minipage}
  \end{center}
  \vskip-1.5em
  \caption{The behaviours of a server $\server_0$ and of a client $\client_0$.}
  \label{fig:first-example}
  \hrulefill
\end{figure}
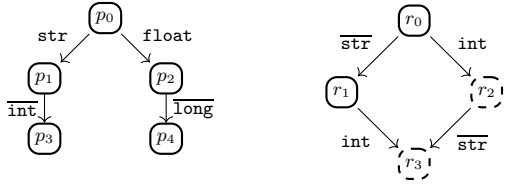

\paragraph{Labelled transition systems}
A program~$\serverA$ is associated with
a labelled transition system (LTS) representing its behaviour,
 which we denote by $\ltsof{\serverA}$.
\rfig{first-example} presents two instances of LTSs, where
transitions are labelled by input actions such as $\texttt{str}$, 
output actions such as $\co{\texttt{str}}$, or the internal action $\tau$ (not
featured in~Figure~\ref{fig:first-example}), while dotted nodes represent
successful states, \ie those satisfying the predicate~$\goodSym$. There, the
server~$\server_0$ is ready to input either a string or a float.
It is the environment that, by offering an output of
either type, will make~$\server$ move to either~$\server_1$
or~$\server_2$.
The client~$\client_0$, on the other hand, is ready to either output a
string, or input an integer. The input ${\tt int}$ makes the client
move to the 
successful state~$\client_2$, while the output $\co{\tt{str}}$ makes
the client move to the state $\client_1$, where it can still perform
the input ${\tt int}$ to reach the 
successful state~$\client_3$. In an asynchronous setting, output
transitions enjoy a 
commutativity property on which we will return later.
Programs $\serverA$ are usually associated with their behaviours
$\ltsof{\serverA}$ via inference rules that we omit in the main body of the
paper, as they are standard.

\paragraph{Alternative preorders for synchrony} 

Program behaviours, \ie~LTSs, are used to define the
alternative preorders for~$\testleqS$
following one of two different
approaches: \MustSets or \AcceptanceSets.

Both approaches were
originally proposed for Milner's
Calculus of Communicating Systems ($\CCS$)~\cite{DBLP:books/daglib/0098267},
where communication is \emph{synchronous}.
The first alternative preorder, which we denote by~$\msleq$, was put
forth by De Nicola~\cite{DBLP:journals/tcs/NicolaH84}, and it
compares server behaviours according to their \MustSets, \ie~the sets of
actions that they
may perform after doing a given sequence of actions.
%
The second alternative preorder, which we denote by~$\asleq$, was put
forth by Hennessy~\cite{DBLP:books/daglib/0066919}, and it compares the
\AcceptanceSets of servers, \ie~how servers can be moved out of their
potentially deadlocked states, namely, states
from which the servers cannot evolve autonomously.
Both these preorders characterise~$\testleqS$ in the
following sense:
\begin{align}
      \label{eq:bhv-mustset-characterisation}
  \forall \serverA , \serverB \in \CCS \wehavethat\; & \serverA \testleqS \serverB
  \text{ iff } \ltsof{\serverA} \msleq \ltsof{\serverB}
  \\
    \label{eq:bhv-accset-characterisation}
  \forall \serverA , \serverB \in \CCS \wehavethat\; & \serverA \testleqS \serverB
  \text{ iff } \ltsof{\serverA} \asleq \ltsof{\serverB}
\end{align}
While these
  alternative preorders do away with the universal quantification
over clients, they are not practical to use directly, as they still universally
quantify over (finite) traces of actions.
A more practical approach \cite{DBLP:journals/jacm/AcetoH92} is to use a
coinductively defined preorder~$\coindleq$ based on~$\asleq$
\cite{DBLP:journals/jacm/AcetoH92,DBLP:conf/concur/LaneveP07,DBLP:journals/mscs/BernardiH16}.
This preorder has two advantages: first, its definition
  quantifies universally only on single actions; second, it allows the
  user to use standard coinductive methods, as found in the literature
  on bisimulation.
%
In the case where the LTS is image-finite, such as for CCS and most process
calculi, the coinductive preorder is sound and complete:
\begin{equation}
      \label{eq:bhv-coind-characterisation}
  \forall \serverA , \serverB \in \CCS \wehavethat\; \serverA \testleqS \serverB
  \text{ iff } \ltsof{\serverA} \coindleq \ltsof{\serverB}
\end{equation}

\paragraph{Asynchrony} %
In distributed systems, communication is inherently
asynchronous. For instance, the standard TCP transmission on the
Internet is asynchronous.  Actor languages like \textsc{Elixir} and
\textsc{Erlang} implement asynchrony via mailboxes, and both
\textsc{Python} and \textsc{JavaScript} offer developers the
constructs \textsc{async/wait}, to return promises (of results) or
wait for them.  In this paper we model asynchrony via
\emph{output-buffered agents with feedback}, as introduced by
Selinger~\cite{DBLP:conf/concur/Selinger97}.  These are LTSs obeying
the axioms in \rfig{axioms}, where~$\aa$ denotes an input
action,~$\co{\aa}$ denotes an output action,~$\tau$ denotes the
internal action, and~$\alpha$ ranges over all these actions.  For
instance, the \outputcommutativity axiom states that an output
$\co{a}$ 
can always be postponed: if $\co{a}$ is followed by any
action~$\alpha$, it can commute with it.  In other words, outputs are
non-blocking, as illustrated by the LTS for~$\client_0$ in
\rfig{first-example}.
We defer a more detailed discussion of these axioms
to \rsec{preliminaries}.

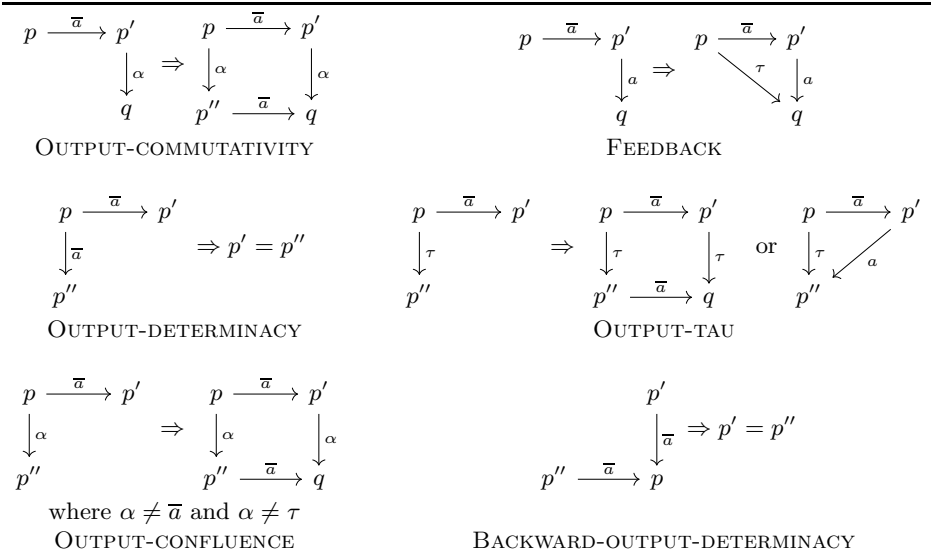
\begin{figure*}[t]
  \hrulefill\\
  \begin{tabular}{c@{\hskip 20pt}c@{\hskip 20pt}c@{\hskip 20pt}c}
    \begin{tabular}{l@{\hskip 4pt}c@{\hskip 0pt}l}
    \begin{tikzcd}
      p \arrow[r, "\co{\aa}"]
      &
      p' \arrow[d, "\alpha"]
      \\
      &
      q
    \end{tikzcd}
    &$\Rightarrow$&
    \begin{tikzcd}
      p
      \arrow[r, "\co{\aa}"]
      \arrow[d, "\alpha"]
      &
      p' \arrow[d, "\alpha"]
      \\
      p'' \arrow[r, "\co{\aa}"]
      &
      q
    \end{tikzcd}
    \end{tabular}
&
\begin{tikzcd}
p \arrow[r, "\co{\aa}"]
&
p' \arrow[d, "\aa"]
\\
&
q
\end{tikzcd}
$\Rightarrow$
\begin{tikzcd}
p
\arrow[r, "\co{\aa}"]
\arrow[rd, "\tau"]
&
p' \arrow[d, "\aa"]
\\
&
q
\end{tikzcd}
\\
\outputcommutativity  &  \outputfeedback
\\[10pt]
\begin{tikzcd}
p \arrow[r, "\co{\aa}"]
\arrow[d, "\co{\aa}"]
&
p'
\\
p''
\end{tikzcd}
$\Rightarrow p' = p''$
&
\begin{tikzcd}
p \arrow[r, "\co{\aa}"]
\arrow[d, "\tau"]
&
p'
\\
p''
&
\end{tikzcd}
$\Rightarrow$
\begin{tikzcd}
p
\arrow[r, "\co{\aa}"]
\arrow[d, "\tau"]
&
p'
\arrow[d, "\tau"]
\\
p''
\arrow[r, "\co{\aa}"]
&
q
\end{tikzcd}
\,\, or
\begin{tikzcd}
p
\arrow[r, "\co{\aa}"]
\arrow[d, "\tau"]
&
p'
\arrow[ld, "\aa"]
\\
p''
\end{tikzcd}
&

\\[5pt]
\outputdeterminacy & \outputtau
\\[10pt]
\begin{tikzcd}
p \arrow[r, "\co{\aa}"]
\arrow[d, "\alpha"]
&
p'
\\
p''
&
\end{tikzcd}
$\Rightarrow$
\begin{tikzcd}
p
\arrow[r, "\co{\aa}"]
\arrow[d, "\alpha"]
&
p' \arrow[d, "\alpha"]
\\
p'' \arrow[r, "\co{\aa}"]
&
q
\end{tikzcd}
&
\begin{tikzcd}
&
p'
\arrow[d, "\co{\aa}"]
\\
p''\arrow[r, "\co{\aa}"]
&
p
\end{tikzcd}
$\Rightarrow p' = p''$
\\
where $\alpha \neq \co{\aa}$ and $\alpha \neq \tau$
\\
\outputconfluence & \outputdeterminacyinv
\end{tabular}
\caption{
First-order axioms for output-buffered agents with feedback as given by Selinger  \cite
  {DBLP:conf/concur/Selinger97}, extended with the \outputdeterminacyinv axiom.}
\label{fig:axioms}
\hrulefill
\end{figure*}

\paragraph{Technical difficulties} 
The practical importance of asynchrony motivates a 
specific study of~$\testleqS$. Efforts in this direction have
been made, all of which focussed on process calculi
\cite{DBLP:conf/fsttcs/CastellaniH98,DBLP:journals/iandc/BorealeNP02,DBLP:phd/us/Thati03,DBLP:journals/jlp/Hennessy05},
while the axioms in \rfig{axioms} apply to LTSs.
Note that these axioms impose conditions
  only over outputs, and this asymmetric treatment of inputs and outputs
substantially complicates the proofs of completeness and soundness of
the alternative characterisations of~$\testleqS$.
To underline the subtleties due to asynchrony, we note that the completeness
result for asynchronous
  \CCS given by Castellani and Hennessy
  in~\cite{DBLP:conf/fsttcs/CastellaniH98},
  and subsequently extended to the $\pi$-calculus by 
  Hennessy~\cite{DBLP:journals/jlp/Hennessy05},
  is false (see \rapp{counterexample}).

\subsection*{Contributions and paper structure.} 
Our main contributions may be summarised as follows (where for each of
them, we 
indicate where it is presented in the paper):
\begin{itemize}
\item 
The first behavioural
characterisations of the \mustpreorder (\rthm{testleqS-equals-accleq},
\rthm{testleqS-equals-mustsetleq}) that are {calculus independent}, in
that both our definitions and our proofs work directly on LTSs.
Contrary to all the previous works on the topic, we show that the {\em
  standard} alternative preorders characterise the~\mustpreorder also
in Selinger asynchronous setting.  To this end, it suffices to enrich
the server semantics with {\em forwarding}, i.e. ensure that servers
are ready to
input any message, as long as they
store it back in a global shared buffer. This idea,
although we use it here in a slightly different form, was pioneered by
Honda et al.~\cite{DBLP:conf/ecoop/HondaT91}.  In this paper we
propose a construction that works on any LTS (\rlem{liftFW-works}) and
we show the following counterparts of
Equations~\eqref{eq:bhv-mustset-characterisation},
\eqref{eq:bhv-accset-characterisation}, and
\eqref{eq:bhv-coind-characterisation} where $\obaFB$ denotes the LTSs
of output-buffered agents with feedback, and~$\liftFWSym$ is the
function that enhances them with forwarding:
\begin{align}
  \forall \serverA , \serverB \in \obaFB
    \wehavethat\;&
    \serverA \testleqS \serverB
    \text{ iff }
    \liftFW{\serverA} \msleq \liftFW{\serverB}
    \tag{a}
    \\
\forall \serverA , \serverB \in \obaFB
    \wehavethat\;&
    \serverA \testleqS \serverB
    \text{ iff }
    \liftFW{\serverA} \asleq \liftFW{\serverB}
    \tag{b}
    \\
\forall \serverA , \serverB \in \obaFB
    \wehavethat\;&
    \serverA \testleqS \serverB
    \text{ iff }
    \liftFW{\serverA} \coindleq \liftFW{\serverB}
    \tag{c}
\end{align}

  
  \noindent
  Quite surprisingly, the alternative preorders~$\asleq$, $\msleq$
  and~$\coindleq$ need not be changed. We present these results in
  \rsec{bhv-preorder}. 
%
  We use the coinductive preorder $\coindleq$ to prove the correctness of a form
  of code hoisting \eqref{eq:mailbox-hoisting}.

  \item%
    The first characterisations of the \mustpreorder that fully exploit
  asynchrony, \ie disregard irrelevant (that is, non-causal) orders of visible
  actions in traces (\rcor{asynleq-equals-bhvleq}).

  \item %
  The first constructive account of the \mustpreorder.
  We show that if the~$\opMust$ and termination
  predicates are defined {\em \intentionally} (in the sense of
  Brede and Herbelin~\cite{DBLP:conf/lics/BredeH21}),
  then~$\testleqS$ can be characterised constructively.
  The original definitions of~$\opMust$ and termination given
  by De Nicola~\cite{DBLP:journals/tcs/NicolaH84}, though, are {\em \extensional}.
  We show how to use Brouwer \barinduction principle to
  prove that
  the two approaches are logically equivalent (\rcor{inductive-char-must}).
  Since Rahli et al.~\cite{DBLP:journals/jacm/RahliBCC19} have shown \barinduction
  to be compatible  with constructive type theory, we argue 
  that our development is entirely constructive.


\item The first mechanisation of the theory of \mustpreorder in a
  fully nondeterministic setting, which consists of around 8000 lines
  of Coq.  In \rapp{coq} we gather the Coq versions of all the
  definitions and the results presented in the main body of the paper.
\end{itemize}

In~\rsec{conclusion}, we discuss the impact of the above
contributions, as well as related and future work.
In \rsec{preliminaries}, we recall the necessary background
  definitions and illustrate them with a few examples.

\begin{figure}[t]
  \hrulefill
  \begin{minted}{coq}
    Class Sts (A: Type) := MkSts {
      sts_step: A → A → Prop;
      sts_stable: A → Prop; }.

    Inductive ExtAct (A: Type) :=     Inductive Act (A: Type) :=
    | ActIn (a: A) | ActOut (a: A).   | ActExt (ext: ExtAct A) | τ.

    Class Label (L: Type) :=          Class Lts (A L : Type) `{Label L} := 
    MkLabel {                         MkLts {
     label_eqdec: EqDecision L;        lts_step: A → Act L → A → Prop;
     label_countable: Countable L; }.  lts_outputs: A → finite_set L;
                                       lts_performs: A → (Act L) → Prop; }.
  \end{minted}
  \caption{Highlights of our Sts and Lts typeclasses.}
  \label{fig:mechanisation-lts}
  \label{fig:sketch-mechanisation-sts}
  \hrulefill
\end{figure}


\section{Preliminaries}
\label{sec:preliminaries}
We model individual programs such as
servers~$\server$ and clients~$\client$ 
as LTSs obeying Selinger axioms, while client-server systems
$\csys{\server}{\client}$ are modelled as state transition systems
with a reduction semantics. We now formally define this two-level
semantics.



\paragraph{Labelled transition systems}
A \emph{labelled transition system} (LTS) is a triple
$\genlts = \lts{ \States }{ L }{ \st{} }$ where~$\States$ is the set
of states, $L$~is the set of labels
  and ${\st{}}
  \subseteq \States \times L \times \States$ is the transition
  relation.
  When modelling programs as LTSs, we use transition labels to
  represent program actions. The set of labels in Selinger LTSs has
  the same structure as the set of actions in Milner's calculus \CCS:
  one assumes a set of names~$\Names$, denoting input actions and
  ranged over by~$\aa, \ab, c$, a complementary set of
  conames~$\overline{\Names}$, denoting output actions and ranged over
  by $\co{\aa}, \co{\ab}, \co{c}$, and an \emph{invisible}
  action~$\tau$, representing internal computation.
The set of all actions, ranged over by $\alpha, \beta, \gamma$,
is given by
$\Acttau \;\;\eqdef\;\; \Names \uplus \overline{\Names}
   \uplus \{ \tau \} $.
We use $\mu, \mu'$ to range over the set of visible actions $\Names
\uplus \overline{\Names}$, and we extend the complementation function
$\co\cdot$ to this set by letting ${\co{\co \aa}} \eqdef \aa$.
In the following, we will always assume $L = \Acttau$.
Once the LTS is fixed, we write 
$\state \st{\alpha} \stateA$ to mean that
$(\state,\alpha,\stateA)~\in~{\st{}}$ and $ \state \st{ \alpha }$
to mean $\exists \stateA \suchthat \state \st{\alpha} \stateA$.
%

We use $\genlts$ to range over LTSs. To reason
simultaneously on different LTSs, we will use the
symbols~$\genlts_A$ and~$\genlts_B$ to denote respectively
the LTSs~$\lts{\StatesA}{L}{\st{}_A}$ and~$\lts{\StatesB}{L}{\st{}_B}$.

In our mechanisation LTSs are borne out by the typeclass
\mintinline{coq}{Lts} in \rfig{mechanisation-lts}. The states of the LTS
have type $\States$, labels have type $L$, and \mintinline{coq}{lts_step} is
the characteristic function of the transition relation, which we
assume to be decidable.  We let $O(\state) = \outactions{ \state }$
and $I(\state) = \inactions{ \state}$ be respectively the set of
outputs and the set of inputs of state~$\state$.  
We assume that the set $O(\state)$
  is finite for any $\state$.
  In our mechanisation, the set
$O(\state)$ is rendered by the function \mintinline{coq}{lts_outputs},
and we shall also use a function \mintinline{coq}{lts_performs} that
lets us decide whether a state can
perform a transition labelled by a given action.

%
%




\begin{figure}[t]
\hrulefill
$$
\begin{array}{lll}
  \stauserver
  &
  \stauclient
  &
  \scom\\
\begin{prooftree}
{ \server \st{\tau} \server' }
  \justifies
      {\csys{\server}{\client} \st{} \csys{\server'}{\client}}
\end{prooftree}
  \hspace{3em}
&
\begin{prooftree}
  \client \st{\tau} \client'
  \justifies
      {\csys{\server}{\client} \st{} \csys{\server}{\client'}}
\end{prooftree}
  \hspace{3em}
&
\begin{prooftree}
  \server\st{ \mu } \server' \quad \client \st{\co{ \mu }} \client'
  \justifies
  \csys{\server}{\client} \st{} \csys{\server'}{\client'}
\end{prooftree}
\end{array}
$$
\caption{The STS of server-client systems.}
  \label{fig:rules-STS}
\hrulefill
\end{figure}

\paragraph{Client-server systems} %
%
%
%
%
A {\em \svrclt} system (or {\em system}, for short) is a pair
$\csys{\server}{\client}$ in which~$\server$ is deemed to be the server of
client~$\client$.
%
%
In general, every system~$\csys{\server~}{~\client}$ is the root of a
{\em state transition system} (STS), $\sts{ \SysStates }{ \st{ }} $,
where~$\SysStates$ is the set of states and~$\st{}$ is the reduction
relation.  For the sake of simplicity\footnote{In general the
  reduction semantics and the LTS of a calculus are defined
  independently, and connected via the Harmony lemma
  (\cite{sangiorgi}, Lemma 1.4.15 page 51).  
  We have a mechanised proof of it.\label{harmony}}%
  we derive the reduction relation from the LTS
semantics of servers and clients as specified by the rules in
\rfig{rules-STS}.
In our mechanisation (\rfig{sketch-mechanisation-sts}), \mintinline{coq}{sts_step}
is the
characteristic function of the reduction relation~$\st{}$, and
\mintinline{coq}{sts_stable}
is the function that states whether a state can
reduce or not. Both functions are assumed decidable.


\begin{definition}[Computation]
  Given an STS $\sts{ \SysStates }{ \st{} }$ and a state
  $\sysstate_0 \in \SysStates$, a 
  \emph{computation} of $\sysstate_0$ is a finite or infinite
  reduction sequence starting from $\sysstate_0$.
  A computation is
{\em maximal} if either it cannot be extended or it is infinite.
  \hfill $\blacksquare$
  \end{definition}



To formally define the \mustpreorder, we assume a decidable
predicate~$\goodSym$ over clients.  A
computation~$ \csys{\server~}{~\client} = \csys{ \server_0 }{
  \client_0 } \st{ } \csys{ \server_1 }{ \client_1 } \st{ } \csys{
  \server_2 }{ \client_2 } \st{ } \ldots $ is {\em successful} if
there exists a state $\csys{ \server_i}{ \client_i}$ such that
$\good{\client_i}$.
%
We assume the predicate $\goodSym$
to be \emph{invariant under outputs}:
\begin{equation}
\label{eq:good-invariance}
  \text{If} ~~\client \st{ \co{a}} \client' ~~ \text{then}~~
  \good{ \client} \iff \good{ \client'}
\end{equation}
All the previous works on
asynchronous calculi implicitly make this
assumption, since they rely on ad-hoc actions such as $\omega$ or
$\checkmark$ to signal success and they treat them as outputs.
In \rapp{accs} we show that this
assumption holds for the language~\ACCS~(the asynchronous variant of
\CCS) extended with the
process~$\Unit$, which is used as a syntactic means
  to denote GOOD states.
%
Moreover, when considering an equivalence on programs $\simeq$ that is
compatible with transitions, in the sense of Figure~\ref{fig:Axiom-LtsEq}, we
assume the predicate $\goodSym$ to be preserved also by this equivalence.
These assumptions are met by the frameworks in
  \cite{DBLP:conf/fsttcs/CastellaniH98,DBLP:journals/iandc/BorealeNP02,DBLP:journals/jlp/Hennessy05}.

\begin{definition}[Client satisfaction]
  \label{def:must-extensional}
  We write $\Must{\server}{\client} $ if every maximal
  computation of $\csys{\server}{\client}$ is successful.\hfill$\blacksquare$
\end{definition}


\begin{definition}[\mustpreorder]
  \label{def:testleq}
  \label{def:testleqS}
We let $ \server \testleqS \serverB$ whenever for every
client $r$ we have that
$\Must{\server}{\client}$ implies $\Must{\serverB}{\client}$.\hfill$\blacksquare$
\end{definition}


%

\begin{example}
\label{ex:max-comp}
 Consider the system~$\csys{ \server_0 }{ \client_0 }$, where~$\server_0$
  and~$\client_0$ are the server and client given in
  \rfig{first-example}. The unique maximal computation of this system is
  $\csys{ \server_0 }{ \client_0 } \st{} \csys{ \server_1 }{ \client_1
  } \st{} \csys{ \server_3 }{ \client_3 }$.
This computation is successful since it leads the client to the $\goodSym$ state
$\client_3$.
Hence, client  $\client_0$ is
satisfied by server $\server_0$.
%
Since \outputcommutativity implies an absence of causality
between the output $\co{{\tt str}}$ and the input ${\tt int}$ in the
client, it is the order between the input ${\tt str}$ and the output
$\co{\tt int}$ in the server that guides the order of client-server
interactions.\hfill$\qed$
\end{example}

\paragraph{A closer look at Selinger axioms}
Let us now discuss the axioms in \rfig{axioms}.
The \outputcommutativity axiom expresses the non-blocking behaviour of
outputs: 
an output cannot be a cause of any subsequent transition, since it can
also be executed after it, leading to the same resulting state. Hence,
outputs are concurrent with any subsequent transition.  The
\outputfeedback axiom says that an output followed by a complementary
input can also synchronise with it to produce a $\tau$-transition.
These first
two axioms specify properties of outputs that are followed by another
transition. Instead, the following three axioms, \outputconfluence,
\outputdeterminacy and \outputtau, specify properties of outputs that
are co-initial with another transition\footnote{Two transitions are
  co-initial if they stem from the same state.}. The
\outputdeterminacy and \outputtau axioms apply to the case where the
co-initial transition is an identical output or a $\tau$-transition
respectively, while the \outputconfluence axiom applies to the other
cases.  When taken in conjunction, these three axioms state that outputs
cannot be in conflict 
with any co-initial transition, except when this is a
$\tau$-transition: in this case, the \outputtau axiom allows for a
confluent nondeterminism between the $\tau$-transition on one side and
the output followed by the complementary input on the other side.


We now explain the novel \outputdeterminacyinv axiom.  It is the dual of \outputdeterminacy, as it states that also backward transitions with identical outputs lead to the same state. The intuition is that if two programs arrive at the same state by removing the same message from the mailbox, then they must coincide. 
This axiom need not be assumed in \cite{DBLP:conf/concur/Selinger97} because it can be derived from Selinger axioms when modelling a calculus like~\ACCS equipped with a parallel composition operator~$\parallel$ (see Lemma~\ref{lem:output-shape} in Appendix~\ref{sec:accs}).  We use the \outputdeterminacyinv axiom only to prove a technical property of clients (\rlem{inversion-gen-mu}) that is used to prove our completeness result.

\paragraph{Calculi}
A number of asynchronous calculi
\cite{DBLP:conf/ecoop/HondaT91,boudol:inria-00076939,DBLP:conf/fsttcs/CastellaniH98,DBLP:journals/toplas/HennessyR02,palamidessi_2003,DBLP:conf/birthday/Sangiorgi19}
have an LTS that enjoys the axioms in \rfig{axioms}, at least up to
some structural equivalence~$\equiv$. The reason is that these calculi
syntactically enforce outputs to have no continuation, \ie outputs can
only be composed in parallel with other
processes.\footnote{In the calculus \TACCS~(a variant of \ACCS
    tailored for testing semantics)
    of~\cite{DBLP:conf/fsttcs/CastellaniH98} there is a construct of
    asynchronous output prefix, but its behaviour is to spawn the
    corresponding atom in parallel with the continuation, so it
      does not act as a prefix.}.  For
example, Selinger~\cite{DBLP:conf/concur/Selinger97} shows that the
  axioms of \rfig{axioms} hold for the LTS of the calculus~\ACCS (the
asynchronous variant of CCS\footnote{The syntax of~\ACCS, which
    is closely inspired by that of the asynchronous $\pi$-calculus
    with input- and $\tau$-guarded choice~\cite{ACS96,ACS98}, is given
    in~\req{syntax-processes} and discussed later.})  modulo
bisimulation, and in \rapp{accs}
(\rlem{ACCSmodulos-equiv-is-out-buffered-with-feedback})
we prove that they hold also for the LTS {of \ACCS modulo $\equiv$:
\begin{lemma}
  \label{lem:ACCS-obaFB}
  We have that $\lts{\modulo{\ACCS}{\equiv}}{L}{\modulo{\st{}}{\equiv}} \in \obaFB$.
\end{lemma}
}
To streamline
reasoning modulo some (structural) equivalence we introduce the
typeclass \mintinline{coq}{LtsEq}, 
whose instances are LTSs
equipped with an equivalence~$\simeq$
that satisfies the property in \rfig{Axiom-LtsEq}.
Defining output-buffered agents with feedback using \mintinline{coq}{LtsEq}
does not entail any loss of generality, because the equivalence~$\simeq$
can be instantiated using the identity over the states~$\States$.
Further details can be found in \rapp{structural-congruence}.


When convenient we denote LTSs using the following minimal syntax for \ACCS:
\begin{equation}
  \label{eq:syntax-processes}
    p,q,r ::= 
\out{\aa}
\BNFsep g
\BNFsep p \Par p
 \BNFsep \rec{p}
 \BNFsep x
\qquad\quad
 g ::= 
 \Nil
 \BNFsep a.p
 \BNFsep \tau.p
 \BNFsep g \extc g
\end{equation}
as well as its standard LTS\footnote{Where the recursion rule is replaced by the one
    usually adopted for testing semantics, which introduces a
    $\tau$-transition before each unfolding.}  whose properties we
discuss in detail in \rapp{accs}.  This is exactly the syntax used
in~\cite{DBLP:conf/concur/Selinger97,DBLP:journals/iandc/BorealeNP02},
without the operators of restriction and relabelling.
Here the syntactic
category $g$ defines \emph{guards}, \ie the terms that may be used as
arguments for the $+$ operator.
As in most process calculi,
  0 denotes the terminated process that cannot do any action.
Note that, apart from $\Nil$, only
input-prefixed and $\tau$-prefixed terms are allowed as guards, and
that the output prefix operator is replaced by \emph{atoms} $\out{a}$.
In fact, this syntax is
completely justified by Selinger axioms, which, as we argued above,
specify that outputs cannot cause any other action, nor be in
conflict with it.


\begin{figure}
  \hrulefill
  \begin{center}
    \begin{tikzcd}
      p
      \arrow[d, dash, dotted, "{\simeq}" description]
      &
      \phantom{p'}
      \\
      q
      \arrow[r, "\alpha"]
      &
      q'
    \end{tikzcd}
    $\Rightarrow$
    \begin{tikzcd}
      p
      \arrow[r, "\alpha"]
      \arrow[d, dash, dotted, "{\simeq}" description]
      &
      p' \arrow[d, dash, dotted, "{\simeq}" description]
      \\
      q \arrow[r, "\alpha"]
      &
      q'
    \end{tikzcd}
  \end{center}
    \vspace{-1em}
  \caption{Axiom stating that equivalence $\simeq$ is compatible with
    a transition relation.}
\label{fig:Axiom-LtsEq}
\hrulefill
\end{figure}

 \begin{definition}[Transition sequence] 
   \label{def:inf-transition-sequence}
   Given an LTS  $\lts{ \States }{ L }{ \st{} }$ and a state $\state_0 \in
  \States$, a \emph{transition sequence} of $\state_0$ is a finite or infinite
  sequence of the form $\state_0 \alpha_1 \state_1 \alpha_2
  \state_2 \cdots$ with $\state_i \in \States$ and $\alpha_i \in L$, and
  such that, for every $n \geq 1$ such that $p_n$ is in the sequence we have
  $\state_{n-1} \stx{\alpha_n} \state_{n}$.
~\hfill$\blacksquare$
 \end{definition}
 \noindent
 If a transition sequence 
 is made only of  $\tau$-transitions, 
 it is called a {\em computation} by abuse of notation, the idea being that
 usually $\tau$-steps are related to reduction steps via the Harmony
 lemma (see footnote on page~\pageref{harmony}).



We give now an example that illustrates the use of the testing
machinery in our asynchronous setting. This is also a counter-example
to the completeness of the alternative preorder
proposed in~\cite{DBLP:conf/fsttcs/CastellaniH98}, as discussed
in detail in~\rapp{counterexample}.


\begin{example}
\label{ex:pierre}
Let $\Omega = \rec{\tau. x}$ and $\pierre = \ab.(\tau.\Omega \extc c.\co{d}) $.
The LTS of \pierre is as follows:
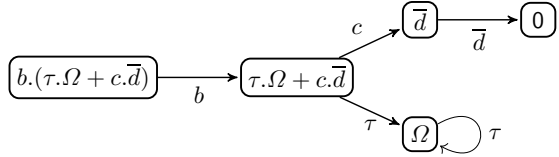
\begin{center}
  \scalebox{.9}{%
    \begin{tikzpicture}
      \node[state] (p) {$b.(\tau.\Omega + c.\co{d})$};
      \node[state][right =of p](p3){$\tau.\Omega + c.\co{d}$};
      \node[state][below right=+8pt and +20pt of p3](p2){$\Omega$};
      \node[state][above right=+8pt and +20pt of p3](p1){$\mailbox{\co{d}}$};
      \node[state][right =of p1](p4){$\Nil$};
      \path (p) edge[to] node[action,swap] {$b$} (p3)
      (p3) edge[to] node[action,below] {$\tau$} (p2)
      (p3) edge[to] node[action] {$c$} (p1)
      (p1) edge[to] node[action,swap] {$\co{d}$} (p4)
      (p2) edge [out=30, in=-30, loop] node[right] {$\tau$} (p2);
    \end{tikzpicture}
    }
  \end{center}
\pierre\ models a 
citizen confronted with an unpopular pension reform. To begin with, \pierre\ can only do
the input~$\ab$, which models 
his getting aware of the brute-force imposition of the
reform by the government.
After performing the input, \pierre reaches the state
$ \tau.\Omega \extc c.\co{d} $, where he behaves in a \nondeterministic
manner. He can internally choose not to trust the government for
  any positive change, in which case he will diverge, refusing any
  further interaction.  But this need not happen: in case the
   government offers
  the action $\co{c}$, which models a positive change in political
  decision, \pierre can decide to accept this change, and then
  he expresses his agreement with the output $\co{d}$,
  which stands for ``done''.
 \hfill$\qed$
\end{example}

\begin{example}
  \label{ex:p-testleqS-Nil}
  We prove now the inequality $\pierre \testleqS \Nil$ by leveraging %
  the possibility of divergence of~$\pierre$ after the
  input~$\ab$. Fix an~$\client$ such that $ \Must{\pierre}{\client} $.
  Note that, since 0 is the terminated process, the condition for
    the server 0 to satisfy r is that r reaches by itself a successful
    state in each of its maximal computations.
We
distinguish two cases, according to whether $\client \st{ \co{b}}$
or $\client \Nst{\co{b}}$.

$i)$ %
Let $\client \st{ \co{b}} \client'$ for some~$\client'$.
Consider the maximal computation
$ \csys{ \pierre }{\client} \st{} \csys{\tau.\Omega \extc
  c.\co{d}}{\client'} \st{} \csys{\Omega}{\client'} \st{}
\ldots $ in which~$\pierre$ diverges and $\client$ does not move after
the first output.  Since $\Must{\pierre}{\client}$,
either~$\good{\client}$ or $\good{\client'}$. In
case~$\good{\client'}$, by~\req{good-invariance} we get also
$\good{\client}$.
Hence $ \Must{\Nil}{\client}$.

$ ii)$ Let $\client \Nst{\co{b}}$.
Suppose $\client = \client_0 \st{\tau} \client_1
  \st{\tau} \client_2 \st{\tau} \ldots$ is a maximal computation of
  ~$\client$.
  Then 
  $\csys{ \pierre }{\client}$ has a maximal computation $
\csys{\pierre}{\client_0} \st{}
\csys{\pierre}{\client_1} \st{}
\csys{\pierre}{\client_2} \st{}\ldots$.
As $\Must{\pierre}{\client}$,
there must exist an $i \in \N$ such that $\good{\client_i}$. Hence
$ \Must{\Nil}{\client}$.
\hfill$\qed$
\end{example}

The argument in \rexa{p-testleqS-Nil} can directly use \rdef{testleq}
because it is very simple to reason on the process $\Nil$.
The issues brought about by the contextuality of \rdef{testleq},
though, hinder showing general properties of~$\testleqS$.
Consider the following form of code hoisting:
  \begin{equation}
    \label{eq:mailbox-hoisting}
     \tau.(\co{\aa} \Par \co{b}) \extc
     \tau.(\co{\aa} \Par \co{c}) \testleqS \co{\aa} \Par (\tau.\co{b} \extc
     \tau.\co{c})
  \end{equation}
  If we see the above nondeterministic sums as representing the two
  branches of a conditional statement, this refinement corresponds to
  hoisting the shared action $\co{a}$ before the conditional
  statement, a common compiler optimisation.
  Proving \req{mailbox-hoisting} via the contextual definition of~$\testleqS$
  is cumbersome. This motivates the study of alternative
  characterisations for~$\testleqS$, and in the rest of the paper we
  present several preorders that fit the purpose, in particular the
  coinductive preorder~$\coindleq$, which we will use to establish
  \req{mailbox-hoisting} in~\rsec{coind-char}.

We conclude this section by recalling auxiliary and rather standard
notions: given an LTS $\lts{\States}{L}{\st{}}$,
the weak transition relation $\state \wt{ \trace } \stateA $,
where $ \trace \in\Actfin$,
is defined via the rules
\begin{description}
\item[\rname{wt-refl}]
  $\state \wt{\varepsilon} \state$
\item[\rname{wt-tau}]
  $\state \wt{ \trace } \stateB$ if $\state\st{\tau}\stateA$
  and $\stateA \wt{ \trace } \stateB$
\item[\rname{wt-mu}]
  $\state \wt{\mu.s} \stateB$ if $\state \st{\mu} \stateA$
  and $\stateA \wt{\trace} \stateB$
\end{description}
We write $ \state \wt{ \trace }  $ to mean  $\exists \stateA \suchthat \state
\wt{ \trace } \stateA$.



\renewcommand{\stateA}{p'}
We write $\state \conv$ and say that {\em $\state$ converges} if every
computation of~$\state$ is finite, and we lift the convergence predicate
to finite traces by letting the relation
${\cnvalong} \subseteq \States \times \Actfin$ be the least one that
satisfies the following rules
\begin{description}
\item[\cnvepsilon] $\state \cnvalong \varepsilon$ if $\state \conv$,
\item[\cnvmu] $ \state \cnvalong \mu.\trace $ if $p \conv$ and
  $\state \wt{\mu} \stateA \implies \stateA \cnvalong \trace$.
\end{description}

  To understand the next section, one should
  keep in mind that all the predicates defined
  above have an implicit parameter: the LTS of programs. 
  By changing this parameter, we may change the meaning of the
  predicates.
  For instance, letting~$\Omega$ be the~\ACCS process~$\rec{\tau. x}$, in the standard LTS
  $\lts{\ACCS}{\st{}}{\Acttau}$ we have~$ \Omega \st{ \tau } \Omega$
  and $\lnot (\Omega \conv)$, while in the LTS $\lts{\ACCS}{ \emptyset
  }{\Acttau}$ we have $\Omega \Nst{\tau}\,$ and thus $\Omega \conv$.
  In other words, the {\em same} predicates can be applied to
  different LTSs, and since the alternative characterisations
  of~$\testleqS$ are defined using such predicates, they
  can relate different LTSs.

\renewcommand{\state}{p}
\renewcommand{\stateA}{p_1}
\renewcommand{\stateB}{p_2}
\renewcommand{\stateC}{q_3}
\newcommand{\stdleq}{\mathrel{\preccurlyeq_{\mathsf{std}}}}
\renewcommand{\state}{p}
\renewcommand{\stateA}{p'}
\renewcommand{\trace}{s}

\newcommand{\asleqC}{\preccurlyeq_{\textsf{cnv}}}
\newcommand{\asleqA}{\preccurlyeq_{\textsf{acc}}}
\newcommand{\asleqAfw}{\preccurlyeq^{\fw}_{\textsf{acc}}}

\newcommand{\util}{\rightarrow\equiv}

\section{Preorders based on \AcceptanceSets}

\label{sec:bhv-preorder}

We first recall the definition of the standard alternative preorder~$\asleq$,
and show how to use it to characterise~$\testleqS$ in our asynchronous setting.
We also present a new characterisation that disregards the order of non-causally
related actions.
We then explain the tools we use to prove these characterisations, and in
particular their soundness.
This section ends with the coinductive version~$\coindleq$ of~$\asleq$,
which we use to prove the hoisting refinement~\eqref{eq:mailbox-hoisting}.

\subsection{Trace-based characterisations} 
The {\em ready set} of a program~$\server$ is defined as
\coqLTS{ready_set} $R( \server ) = I( \server ) \cup O( \server )$,
and it contains all the {\em visible} actions that~$\server$ can
immediately
perform.  If a program~$\server$ is
stable, \ie it cannot perform any
$\tau$-transition, we say that it is a {\em potential
  deadlock}.  In general, the ready set of a potential
deadlock~$\server$ shows how to make~$\server$ move to a different
state, possibly one that can perform further computation: if
$R( \server ) = \emptyset$ then there is no way to make~$\server$ move
on, while if~$R( \server )$ contains some action, then~$\server$ is
a state waiting for the environment to interact with it.
Indeed, potential deadlocks are called {\em waiting states}
in~\cite{DBLP:conf/ecoop/HondaT91}.
In particular, in an asynchronous setting the outputs of a
potential deadlock~$\server$ show how~it 
can unlock the inputs
of a client, which in turn may lead the client to a novel state that
can make~$\server$ move, possibly to a state that can perform further
computation.
A standard manner to capture all the ways out of the potential
deadlocks that a program~$\server$ encounters after executing a
trace~$\trace$ is its {\em acceptance set}:
$\accP{ \server }{ \trace }{ \st{} } = \setof{ R( \stateA ) }{ \state \wt{ \trace } \stateA \stable }$.

In our presentation we indicate explicitly the third parameter of
$\mathcal{A}$, \ie the transition relation of the LTS at hand, because
when necessary we will manipulate this parameter.
For any two LTSs $\genlts_\StatesA, \genlts_\StatesB$ and servers
  $ \serverA \in \StatesA, \serverB \in \StatesB$, we 
  write $\accP{ \serverA }{ \trace }{\st{}_\StatesA} \ll \accP{
    \serverB }{ \trace }{ \st{}_\StatesB}$ 
  if for every $R \in \accP{
    \serverB }{ \trace }{ \st{}_\StatesB}$
there exists $\widehat{R} \in \accP{ \serverA }{ \trace }{\st{}_\StatesA}$
such that $\widehat{R}
  \subseteq R$.
  We can now recall the definition of the behavioural preorder à la
  Hennessy, $\asleq$, which is based on acceptance
  sets~\cite{DBLP:books/daglib/0066919}.

  \newcommand{\cnvalongLTS}[2]{ \Downarrow_{#1} #2 }

   \begin{definition}\coqMT{bhv_pre}
  \label{def:accset-leq}%
  \label{def:standard-char}
  {We write}
  \begin{itemize}
  \item
    $
    \serverA \bhvleqone \serverB$ whenever $\forall \trace \in \Actfin \suchthat
    \serverA \cnvalongLTS{\StatesA}{ \trace } \implies \serverB \cnvalongLTS{\StatesB}{ \trace }
    $,
    
  \item
    $
    \serverA \bhvleqtwo \serverB$ whenever $\forall \trace \in \Actfin \suchthat
    \serverA \cnvalongLTS{\StatesA}{ \trace }  \implies 
    \accP{ \serverA }{ \trace }{ \st{}_{\StatesA} }
    \ll \accP{ \serverB }{ \trace }{ \st{}_{\StatesB} }
    $,
  \item
    $\serverA \asleq \serverB$ whenever $\serverA \bhvleqone
    \serverB  \text{ and }  \serverA \bhvleqtwo \serverB$.\hfill$\blacksquare$
  \end{itemize}  
   \end{definition}

In the synchronous setting, the behavioural preorder~$\asleq$ is closely related to the denotational semantics
based on Acceptance Trees  proposed by Hennessy in
\cite{DBLP:journals/jacm/Hennessy85,DBLP:books/daglib/0066919}.
There the predicates need not be annotated with the LTS that they are used on,
because those works treat a unique LTS.
Castellani and Hennessy~\cite{DBLP:conf/fsttcs/CastellaniH98} show in their Example~4 that
the condition on acceptance sets, \ie $\bhvleqtwo$, is too demanding
in an asynchronous setting.
Letting $\serverA = \aa.\Nil$ and $\serverB = \Nil$, they show that
$\serverA \testleqS \serverB $ but $ \serverA \not \asleq \serverB$,
because $\acc{ \serverA }{ \epsilon } =\set{\set{a}}$ and
$\acc{ \serverB }{ \epsilon } =\set{\emptyset}$, and corresponding to
the ready set $\emptyset \in \acc{ \serverB }{ \epsilon }$ there is no
ready set $\widehat{R} \in \acc{ \serverA }{ \trace}$ such that
$\widehat{R} \subseteq \emptyset$. %
%
%
Intuitively this is the case because acceptance sets treat inputs and
outputs similarly, while in an asynchronous setting only outputs can
be tested.

Nevertheless~$\asleq$ characterises~$\testleqS$, if servers
are enhanced as with forwarding. We now introduce this concept.



\paragraph{Forwarders}
We say that an LTS~$\genlts$ is of output-buffered agents {\em with forwarding},
for short is \obaFW, if it satisfies all the axioms in
\rfig{axioms} except \outputfeedback, and also the two following axioms:

\begin{equation}
  \label{eq:axioms-forwarders}
  \begin{tabular}{cc}
    \begin{tikzpicture}
      \node (p) {$p$};
      \node[right=+40pt of p] (p') {$p'$};

      \path[->]
      (p) edge [bend left] node[above] {$a$}  (p')
      (p') edge [bend left] node[above] {$\co{a}$}  (p);
    \end{tikzpicture}
    &
    \begin{tabular}{l@{\hskip 4pt}c@{\hskip 0pt}l}
      \begin{tikzcd}
      p \arrow[r, "\co{\aa}"]
      &
      p' \arrow[d, "\aa"]
      \\
      &
      q
    \end{tikzcd}
    &$\Rightarrow$&\,\,
    $p \st{ \tau } q$ or $p = q$

    \end{tabular}
    \\[5pt]
    \boom
&

    \fwdfeedback
  \end{tabular}
  \end{equation}

  The~\boom axiom states a kind of input-enabledness property, which
  is however more specific as it stipulates that the target state of
  the input should loop back to the source state via a complementary
  output. This is the essence of the behaviour of a forwarder, whose
  role is simply to pass on a message and then get back to its
  original state.  The~\fwdfeedback axiom is a weak form of
  Selinger's~\outputfeedback axiom, which is better understood in
  conjunction with the~\boom axiom: if the transition
    sequence 
  $p\st{\out{a}} p'\st{a}q$ in the~\fwdfeedback axiom is taken to be
  the transition sequence 
  $p'\st{\out{a}} p\st{a}p'$ in the~\boom axiom, then we see that it
  must be $q=p$ in the~\fwdfeedback axiom.  Moreover, no~$\tau$ action
  is issued when moving from $p$ to $q$, since no synchronisation
  occurs in this case: the message is just passed on.

To prove that~$\asleq$ is sound and complete with respect to~$\testleqS$:
\begin{enumerate} \item we define a function
  $\liftFWSym : \obaFB \longrightarrow \obaFW$ that lifts any
  LTS~$\genlts \in \obaFB$ into a suitable
  LTS~$\liftFW{\genlts}\in \obaFW$, %
  and
\item we check the
  predicates~$\cnvalong$ and~$\accP{-}{-}{-}$ over the LTS~$\liftFW{\genlts}$.
\end{enumerate}

Let~$\MO$ denote the set of all finite multisets of output actions, 
for instance we have $\varnothing, \mset{ \co{ a } }, \mset{ \co{ a },   \co{
    a }  }, \mset{ \co{ a },  \co{ b },  \co{ a },  \co{ b }} \in
\MO$.
We let 
$M, N, \ldots$ range over~$\MO$. The symbol~$M$ stands for {\em mailbox}.
We denote with~$\uplus$ the multiset union.

\begin{definition}
  \label{def:sta}
  \label{def:liftFW}\coqLTS{lts_a}
  Let $\liftFW{\genlts} = \lts{\States \times \MO}{L}{\sta{}}$ for every $\genlts = \lts{\States}{L}{\st{}}$,
where the states in $\liftFW{\genlts}$ are pairs denoted $p
\triangleright \mailbox{M}$, such that $p \in \States$ and $M \in \MO$,
and the transition relation~$\sta{}$ is defined via the rules in
\rfig{rules-liftFW}.\hfill$\blacksquare$
\end{definition}

\begin{figure}
\hrulefill
  $$
  \begin{array}{llllll}
    \stproclift &
    \begin{prooftree}
      \server \st{\alpha} \server'
      \justifies
      \server \triangleright M \sta{\alpha} \server' \triangleright M
    \end{prooftree}
    &
    \stcommlift &
  \begin{prooftree}
    \server \st{a} \server'
    \justifies
    \server \triangleright (\mset{\co{a}} \uplus M) \sta{\tau} \server' \triangleright M
  \end{prooftree}
  \\[2em]
  \stmoutlift &
  \begin{prooftree}
    \justifies
    \server \triangleright (\mset{\co{a}} \uplus M) \sta{\co{a}} \server \triangleright M
  \end{prooftree}
  &
  \stminplift &
  \begin{prooftree}
    \justifies
    \server \triangleright M \sta{a} \server  \triangleright (\mset{\co{a}} \uplus M)
  \end{prooftree}
  &&
  \end{array}
  $$
  \caption{Lifting of an LTS to an LTS with forwarding.}
  \label{fig:rules-liftFW}
  \hrulefill
\end{figure}

Let us briefly comment on~\rdef{liftFW} and the rules
in~\rfig{rules-liftFW}.  The pair $\server \triangleright M$ is a
kind of asymmetric parallel composition between a process $\server$
and a mailbox $M$.
Rule $\stproclift$ says that the process can evolve
independently of the mailbox.
Rule $\stcommlift$ says that an input in the process can
  synchronise with a complementary output in the mailbox.
  Rules $\stminplift$ and $\stmoutlift$ express the essence of the
  forwarding behaviour: the pair
  $\server \triangleright M$ may input
  any message from the environment
  and store it into the mailbox $M$ (Rule $\stminplift$); dually, the
  pair $\server \triangleright M$ may output any message in the
  mailbox $M$ towards the environment (Rule $\stmoutlift$).  Note that
  in both
  cases, the interaction occurs between the mailbox $M$ and the
  environment, without any participation of the process $\server$.


\begin{example}
  \label{ex:forwarders-in-TACCS}
  If a calculus is fixed, then the function~$\liftFWSym$ may have a
  simpler definition.  For instance Castellani and
    Hennessy~\cite{DBLP:conf/fsttcs/CastellaniH98} define it in their
  calculus~\textsf{TACCS} by letting $\sta{\alpha}$ be the least
  relation over~\textsf{TACCS} such that
  (1) for every $\alpha\in\Acttau \wehavethat \st{\alpha} {} \subseteq {} \sta{\alpha}$, and 
  (2) for every $\aa \in \Names \wehavethat p \sta{\aa} p \Par \out{\aa}$.\hfill$\qed$
\end{example}

The transition relation $\sta{}$ is reminiscent of the one introduced
in Definition 8 by Honda and
  Tokoro in~\cite{DBLP:conf/ecoop/HondaT91}.  The construction given in
our \rdef{liftFW}, though, does not yield the LTS of Honda and Tokoro,
as $\sta{}$ adds the forwarding capabilities to the states only at the
top-level, instead of descending structurally into terms. As a
consequence, in the LTS of~\cite{DBLP:conf/ecoop/HondaT91}
$\aa.\Nil + \Nil \st{ \ab } \co{ \ab }$, while
$\aa.\Nil + \Nil \triangleright M \Nsta{ \ab } M \uplus \mset{ \co{ \ab } }$.

\begin{example}
  As the set~$\Names$ is countable, every process~$\server$ that belongs to
  the LTS $\lts{\ACCS \times \MO}{\Acttau}{\sta{}}$ is
  infinitely-branching: 
  for every mailbox $M$ we have
  $\server \triangleright M \sta{\aa_i} \server \triangleright
(\mset{\co{\aa_i}} \uplus M)$
    for every $a_i \in \Names$. This is illustrated by the following
    picture, where for simplicity we omit the subscript $\mathsf{fw}$ under the arrows.
%
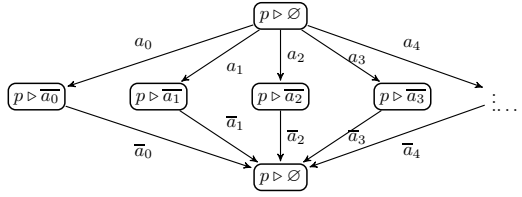
\begin{center}
  \scalebox{.7}{%
  \begin{tikzpicture}
  \node[state](t){$\server \triangleright \varnothing$};

  \node[state][below=  of t](p20){$\server \triangleright \mailbox{\co{\aa_2}}$};
  \node[state][below=  of p20] (nil){$\server \triangleright \varnothing$};

  \node[state][left=of p20] (p10) {$\server \triangleright \mailbox{\co{\aa_1}}$};

  \node[state][left =of p10] (p00) {$\server \triangleright \mailbox{\co{\aa_0}}$};

  \node[state][right=of p20](p30){$\server \triangleright \mailbox{\co{\aa_3}}$};

  \node[][right=of p30](p5){\vdots\ldots\vdots};

\path (t) edge[to] node[action,swap] {$\aa_0$} (p00)
      (t) edge[to] node[action] {$\aa_1$} (p10)
      (t) edge[to] node[action] {$\aa_2$} (p20)
      (t) edge[to] node[action,yshift=-8pt] {$\aa_3$} (p30)
      (t) edge[to] node[action] {$\aa_4$} (p5);

\path
(p00)  edge[to] node[action, swap] {$\co{\aa}_0$} (nil)
(p10)  edge[to] node[action] {$\co{\aa}_1$} (nil)
(p20)  edge[to] node[action] {$\co{\aa}_2$} (nil)
(p30)  edge[to] node[action, right] {$\co{\aa}_3$} (nil)
(p5)  edge[to] node[action] {$\co{\aa}_4$} (nil);
  \end{tikzpicture}
  }
\end{center}
 \hfill$\qed$
    \vspace{-12pt}
\end{example}

The intuition behind \rdef{sta} 
is that, when a client interacts with a server asynchronously, the
client can send any message it likes,
regardless of the inputs that
the server can actually perform. In fact, asynchronous clients behave
as if the server was saturated with \emph{forwarders}, namely
processes of the form $\aa. \co{\aa}$, for any $\aa\in\Names$.

The function~$\liftFWSym$ enjoys two crucial properties: 
it lifts any LTS of output-buffered agents with feedback
to an LTS with forwarding, and the lifting preserves the~$\opMust$ predicate. We can thus
reason on~$\testleqS$ using LTSs in $\obaFW$.

\begin{lemma}
  \label{lem:liftFW-works}
  For every LTS~$\genlts \in \obaFB$, $\liftFW{\genlts} \in \obaFW$.
\end{lemma}
\begin{proof}
  See \rapp{appendix-forxarders}.
  \qed
\end{proof}

\begin{lemma}
  \label{lem:musti-obafb-iff-musti-obafw}
  For every $\genlts_A, \genlts_B, \genlts_C \in \obaFB, \serverA \in A, \serverB \in B, \client \in C$,
  \begin{enumerate}
    \item
      $ \Must{\server}{\client}$ if and only if $\Must{\liftFW{\server}}{\client}$,
    \item
      $\serverA \testleqS \serverB$ if and only if $\liftFW{\serverA} \testleqS \liftFW{\serverB}$.
  \end{enumerate}
\end{lemma}






We now simplify the definition of acceptance sets to reason on
LTSs that are in \obaFW:
for any LTS $\genlts = \lts{A}{\Act}{ \st{} } \in \obaFW$ and program 
$ \serverA \in \StatesA$ we let
$\accfwp{ \state }{ \trace }{ \st{} } =  \setof{ O(\state') }{ \state
  \wt{ \trace } \state' \stable }$.
This definition suffices to characterise~$\testleqS$ because in each LTS that is \obaFW\ every state performs
every input, thus comparing inputs has no impact on the
preorder~$\bhvleqtwo$ of \rdef{standard-char}. More formally,
for every $\genlts_\StatesA, \genlts_\StatesB \in \obaFW$ and every $\server \in \StatesA$ and $\serverB \in \StatesB$,
we let 
$$
\serverA \asleqAfw \serverB \text{ iff } \forall \trace \in \Actfin \wehavethat \serverA \cnvalong \trace \implies \accfwp{ \serverA }{ \trace }{ \st{}_\StatesA } \ll \accfwp{ \serverB }{ \trace }{ \st{}_\StatesB }
$$
We have the following logical equivalence.
\begin{lemma}
  \label{lem:conditions-on-accsets-logically-equivalent}
  Let $\genlts_A, \genlts_B \in \obaFW$.
  For every $\serverA \in \StatesA, \serverB \in \StatesB, \serverA \bhvleqtwo \serverB$
  if and only if $\serverA \asleqAfw \serverB$.
\end{lemma}
\begin{proof}
  The {\em only if} implication is trivial, so we discuss the {\em if}
  one. Suppose that $ \serverA \asleqAfw \serverB  $ and that for some
  $\trace$ we have that $R \in \accP{\serverB}{s}{\st{}_B}$. Let~$X$ be
  the possibly empty subset of~$R$ that contains only output actions.
  Since~$\genlts_B$ is~\obaFW\ we know by definition that $R =
  {X} \cup {\Names}$.
  By definition $X \in \accfwp{\serverB}{s}{\st{}_B}$, and thus by
  hypothesis there exists a set of output actions $Y \in
  \accfwp{\serverA}{s}{\st{}_A}$ such that $Y \subseteq X$.
  It follows that the set ${{Y} \cup {\Names}} \in
  \accP{\serverA}{s}{\st{}_A}$, and trivially $Y \cup \Names \subseteq
  {X} \cup {\Names} = R$.
  \qed
\end{proof}



\renewcommand{\serverA}{p}
\renewcommand{\serverB}{q}


In view of the second point of \rlem{musti-obafb-iff-musti-obafw},
to prove completeness it suffices to show that~$\asleq$
includes~$\testleqS$ over LTSs with
forwarding. This is indeed true:
\begin{lemma}
  \label{lem:completenessA}
  For every $\genlts_A, \genlts_B \in \obaFW$ and
  $\serverA \in \StatesA, \serverB \in \StatesB $,
  if ${ \serverA } \testleqS { \serverB }$
  then we have~${ \serverA } \asleq { \serverB }$.
\end{lemma}

By a slight abuse of notation,
given an LTS $\genlts = \lts{\States}{L}{\st{}}$ and a state
$\server \in \States$,
we denote with $\liftFW{
  \server }$ the LTS rooted at $\server \triangleright \mailbox{ \emptyMset }$ in $\liftFW{\genlts}$.

\begin{theorem}
  \label{thm:testleqS-equals-bhvleq}
  \label{thm:testleqS-equals-accleq}
  \label{thm:testleqS-equals-asleq}
  For every $\genlts_A, \genlts_B \in \obaFB$
  and $\serverA \in \States, \serverB \in \StatesB,$
  \begin{equation*}
    \serverA \testleqS \serverB \;\;\;\text{iff}\;\;\;
  \liftFW{\serverA} \asleq \liftFW{\serverB}.
  \end{equation*}
\end{theorem}

\noindent
This theorem is the linchpin of this paper, as all other results presented in
this paper are corollaries.
To begin with, instantiating this theorem to a calculus which can be given an
LTS that satisfies the axioms for output-buffered agents with feedback
such as \ACCS (\rlem{ACCS-obaFB}), the core join-calculus~\cite{join-calculus}
or KLAIM~\cite{klaim}, we get a characterisation of the \mustpreorder essentially
for free. In our Coq development, we instantiate it with $\ACCS$:
\begin{corollary}
  \label{cor:characterisation-for-aCCS}
For every $\serverA, \serverB \in \modulo{\ACCS}{\equiv}, \serverA \testleqS \serverB$ iff $\liftFW{\serverA} \asleq
\liftFW{\serverB}.$
\end{corollary}

\noindent
Another application of \rthm{testleqS-equals-bhvleq} is a novel
behavioural characterisation of the \mustpreorder, which fully exploits
asynchrony, \ie disregards irrelevant non-causal orders of visible
actions in traces. %
For space reasons, we defer the discussion of this result to \rapp{normal-forms}.

  So far, we have seen the more direct applications of
  \rthm{testleqS-equals-bhvleq}.  Before we explain two other
  applications, namely the coinductive characterisation of the
  \mustpreorder and the relation of~$\testleqS$ with the failure
  refinement, we outline the proof of \rthm{testleqS-equals-bhvleq},
  and, in particular, the technical tools we used.
  

\subsection{Proof of \rthm{testleqS-equals-bhvleq}}

The full proof of \rthm{testleqS-equals-bhvleq}, which is given in the Appendix,
as well as in the Coq development, comprises two parts:
\rapp{proof-completeness} deals with completeness, where the main aim is to show
\rlem{completenessA}, and \rapp{proof-soundness} deals with soundness.
Here we outline the main tools we use to prove the soundness of the
$\asleq$ preorder: \emph{Bar-induction}, which allows us to relate the standard
definition of $\opMust$ with an inductive one that is more practical to use,
especially in a constructive setting, and the \emph{LTS of sets}, which is
derived from the LTS of the processes.

\subsubsection{\Barinduction: from \extensional to \intentional definitions}
\label{sec:barinduction-main-body}
\label{sec:bar-induction-technicalities}
We present the inductive characterisations of $\conv$ and $\opMust$ in any state
transition system (STS) \sts{S}{\to} that is countably branching.
In practice, this condition is satisfied by most concrete LTS of programming
languages, which usually contain countably many terms; this is 
the case for \ACCS and for the asynchronous $\pi$-calculus.

Following the terminology of \cite{DBLP:conf/lics/BredeH21} we introduce
extensional and \intentional\ predicates associated to any decidable
predicate~$Q: S \to \mathbb{B}$ over an STS~$\sts{S}{\to}$, where
$\mathbb{B}$ denotes the set of booleans.

\begin{definition}
  \label{def:def-bar}
  The \emph{\extensional\ predicate} $\mathsf{ext}_Q(s)$ is defined, for~$s \in
  S$, as
  \[
    \forall \eta \text{ maximal execution of~$S$} \wehavethat
    \eta_0 = s \implies 
    \exists n \in \N,\;
    Q(\eta_n)
  \]
  The \emph{\intentional\ predicate} $\mathsf{int}_Q$ is the inductive predicate
  (least fixpoint) defined by the following rules:
  $$
  \begin{array}{l@{\hskip 3pt}l@{\hskip 20pt}l@{\hskip 3pt}l}
    \rname{axiom}
&    \begin{prooftree}
      Q(s)
      \justifies
      \mathsf{int}_Q(s)
    \end{prooftree}
    &
    \rname{ind-rule}
    &
    \begin{prooftree}
      s \to
      \qquad
      \forall s' \wehavethat  s \to s' \implies \mathsf{int}_Q(s')
      \justifies
      \mathsf{int}_Q(s)
    \end{prooftree}
  \end{array}
  $$
  \hfill$\blacksquare$
\end{definition}
\noindent
For instance, by letting
\begin{equation*}
  Q_1(\state) \;\iff\; \state \Nst{\phantom{\tau}} \qquad\qquad
  Q_2(\state, \client) \;\iff\; \good{\client}
\end{equation*}
\noindent
we have by definition that
\begin{equation}
  \tag{ext-preds}
  \label{eq:def-extensional-predicates}
  \state \conv \;\iff\; \mathsf{ext}_{Q_1}( \state ) \qquad\qquad
  \Must{ \server }{\client } \;\iff\;
  \mathsf{ext}_{Q_2}(\state, \client )
\end{equation}
that is the standard definitions of $\conv$ and $\opMust$ are \extensional.
Our aim now is to prove that they coincide with their \intentional\
counterparts.
The reader not familiar with this terminology may find in \rapp{bar-induction}
an informal and hopefully intuitive explanation.
Since we will use the \intentional\ predicates in the rest of
the paper a little syntactic sugar is in order, let
\begin{equation}
  \tag{int-preds}
  \label{eq:def-intentional-predicates}
  \state \convi \;\iff\; \mathsf{int}_{Q_1}( \state ) \qquad\qquad
  \musti{ \server }{\client } \;\iff\;
  \mathsf{int}_{Q_2}(\state, \client )
\end{equation}

The proofs of soundness, \ie that the inductively defined predicates imply the extensional ones,
are by rule induction:
\begin{lemma}
  \label{lem:intensional-implies-extensional}
  For $\state \in S$,
  \begin{enumerate}[(a)]
  \item $\state \convi $ implies $\state \conv$,
  \item for every $\client \wehavethat \musti{\state}{\client}$ implies $\Must{ \state }{\client}$.
  \end{enumerate}
\end{lemma}

The proofs of completeness are more delicate. To the best of our knowledge, the ones about \CCS \cite{TCD-CS-2010-13,phdbernardi}
proceed by induction on the greatest number of steps necessary to arrive at
termination or at a successful state. Since the STS of \sts{\CCS}{\st{\tau}} is
finite branching, \koenigslemma\ guarantees that such a bound exists. This
technique does not work on infinite-branching STSs, for example the one of \CCS
with infinite sums \cite{DBLP:journals/corr/BernardiH15}.
If we reason in classical logic, we can prove completeness without \koenigslemma
and also over infinite-branching STSs via a proof {\em ad absurdum}: suppose $p \conv$. If
$\lnot (p \convi)$ no finite derivation tree exists to prove $p \convi$, and
then we construct an infinite sequence of $\tau$ moves starting with $p$, thus $\lnot (p \conv)$.
Since we strive to be constructive we replace reasoning {\em ad absurdum} with a
constructive axiom: (decidable) \emph{\barinduction}. In the rest of this section we discuss this axiom, and adapt it to our \svrclt setting. This requires a little terminology.

\newcommand{\utree}{T_{\N}}

\paragraph{\Barinduction}
The axiom we want to use is traditionally stated using natural numbers.
We use the standard notations $\N^\star$ for finite sequences of natural numbers,
$\N^\omega$ for infinite sequences, and~$\N^\infty = \N^\star \cup \N^\omega$ for
finite or infinite sequences.
Remark that, in constructive logics, given $u \in \N^\infty$, we cannot do a
case analysis on whether $u$~is finite or infinite.
The set $\N^\infty$ equipped with the prefix order can be seen as a \emph{tree},
denoted $\utree$, in the sense of set theory: a tree is an ordered set~$(A,
\leq)$ such that, for each~$a \in A$, the set~$\{ b \mid b < a \}$ is
well-ordered by~$<$. A \emph{path} in a tree A is a maximal element in~$A$.
In the tree~$\N^\infty$, each node has~$\omega$ children, and the paths are
exactly the infinite sequences~$\N^\omega$.

A predicate~$P \subseteq \N^\star$ over finite words is a {\em bar} if every
infinite sequence of natural numbers has a finite prefix in~$P$. Note that a bar
defines a subtree of $\utree$ \emph{extensionally}, because it defines each path
of the tree, as a path~$u \in \N^\omega$ is in the tree if and only if there exists a finite
prefix which is in the bar~$P$.

A predicate~$Q \subseteq \N^\star$ is {\em hereditary} if
\[
  \forall w \in \N^\star, \;\;
  \text{if}\;\;
  \forall n \in \N, w \cdot n \in Q
  \;\;\text{then}\;\;
  w \in Q.
\]
\Barinduction states that the extensional predicate associated to a bar implies its \emph{\intentional} counterpart:
a predicate~$P_{\mathit{int}} \subseteq \N^\star$ which contains~$Q$ and which is
hereditary.

\renewcommand{\bar}{Q}
\newcommand{\pint}{P_{\mathit{int}}}

\begin{myaxiom}[Decidable \barinduction over $\N$] 
  Given two predicates~$\pint, \bar$ over~$\N^\star$, such that:
  \begin{enumerate}
    \item for all $\pi \in \N^\omega$, there exists~$n \in
      \N$ such that $(\pi_1, \dots, \pi_n) \in \bar$;
    \item for all $w \in \N^\ast$, it is decidable whether $\bar(w)$ or~$\neg \bar(w)$;
    \item for all $w \in \N^\ast$, $\bar(w) \Rightarrow \pint(w)$;
    \item $\pint$ is hereditary;
  \end{enumerate}
  then $\pint$ holds over the empty word: $\pint(\varepsilon)$.
\end{myaxiom}

\Barinduction is a generalisation of the fan theorem, i.e. the
constructive version
of \koenigslemma \cite[pag. 56]{dummett2000elements},
and states that any extensionally well-founded tree~$T$ can be turned into an
inductively-defined tree~$t$ that realises~$T$ \cite{DBLP:conf/lics/BredeH21,kleene1965foundations}.

Our mechanisation of \barinduction principle is formulated as a Proposition that
is proved using classical reasoning, since it is not provable directly in the
type theory of Coq.
Unfortunately, while \barinduction
  is a constructive principle, mainstream proof
assistants do not support it yet, which is why on the one hand we had
to postulate it as a proof principle while on the other hand we proved
it in classical logic using the Excluded Middle axiom. 
%
This principle though has a computational content, Spector bar recursion\footnote{\url{https://en.wikipedia.org/wiki/Bar_recursion}}, which,
currently, cannot be used in mainstream proof assistants such as Coq. 
Developing a type theory with a principle of \barinduction is 
recent and ongoing work~\cite{DBLP:conf/types/Fridlender98,DBLP:journals/jacm/RahliBCC19}.





\paragraph{Encoding states}
The version of \barinduction we just outlined is not directly suitable for our
purposes, as we need to reason about sequences of reductions rather than sequences
of natural numbers. The solution is to encode STS states by natural 
numbers.
%
This leads to the following issue: the nodes of the tree~$\utree$ have a fixed
arity, namely~$\N$, while processes have variably many reducts, including
zero if they are stable.
To deal with this glitch, it suffices to
assume that there exists the following family of surjections:
%
\renewcommand{\stateB}{q}
\begin{equation}
\label{eq:surjection}
  F( \state ) : \N \to \setof{ \stateB }{ \state \to \stateB }
\end{equation}
where a surjection is defined as follows.
%
\begin{definition}
  A map~$f: A \to B$ is a surjection if it has a section~$g: B \to A$, that is,
  $f \circ g = \mathrm{Id}_B$.
\end{definition}
\noindent
This definition implies the usual one which states the existence of an
antecedent~$x \in A$ for any~$y \in B$, and it is equivalent to it if we assume
the Axiom of Choice.

Using this map~$F$ as a decoding function, any sequence of natural
numbers corresponds to a path in the STS. Its subjectivity means that all paths
of the LTS can be represented as such a sequence. This correspondence allows us
to transport
\barinduction from sequences of natural numbers to executions of
processes.

Note that such a family of surjections $F$ exists for \ACCS\ processes, and
generally to most programming languages, because the set $\Acttau$ is countable,
and so are processes.
This leads to the following version of \barinduction where words and
sequences are replaced by finite and infinite executions.
\begin{proposition}[Decidable \barinduction over an STS]
 Let~$\sts{S}{\to}$ be an STS such that a  surjection as in~(\ref{eq:surjection})
 exists. Given two predicates~$\bar, \pint$ over finite executions, if
  \begin{enumerate}
    \item for all infinite execution~$\eta$, there exists~$n \in
      \N$ such that $(\eta_1, \dots, \eta_n) \in \bar$;
    \item for all finite execution $\zeta$,
    $\bar(\zeta)$ or~$\neg \bar(\zeta)$ is decidable;
    \item for all finite execution $\zeta$, $\bar(\zeta) \Rightarrow \pint(\zeta)$;
    \item $\pint$ is hereditary, as defined above except that $\zeta \cdot q$ is a
      partial operation defined when~$\zeta$ is empty or its last state is~$p$
      and $p \to q$;
  \end{enumerate}
then $\pint$ holds over the empty execution: that is $\pint(\varepsilon)$ holds.
\end{proposition}
The last gap towards a useful principle is the requirement
that every state in our STS has an outgoing transition.
This condition is necessary to ensure the existence of the
surjection in \req{surjection}.
To ensure this requirement given any countably-branching STS, we enrich it by
adding a \emph{sink} state, which (a) is only reachable from stable states of
the original STS, and (b) loops. This is a typical technique, see for instance
\cite[pag. 17]{DBLP:books/aw/Lamport2002}.
\newcommand{\sinkto}{\mathrel{\to^{\kern-0.3pt\text{\smaller[.5]{\tiny\ensuremath{\top}}}}}}
\begin{definition}
  Define~$\mathit{Sink}(S, \to) \coloneqq \sts{S \cup \set{\top}}{\sinkto}$, where~$\sinkto$ is
  defined inductively as follows:
  \[
    p \to q \;\;\Longrightarrow\;\; p \sinkto q
    \qquad
    p \Nst{} \;\;\Longrightarrow\;\; p \sinkto \top
    \qquad
    \top \sinkto \top
  \]
\end{definition}
A maximal execution of~$\mathit{Sink}(S, \to)$ is always infinite, and it
corresponds (in classical logic) to either an infinite execution of~$S$ or a
maximal execution of~$S$ followed by infinitely many~$\top$.
We finally prove the converse of Lemma~\ref{lem:intensional-implies-extensional}.

\begin{proposition}
  \label{prop:ext-impl-int}
  Given a countably branching STS~$\sts{S}{\to}$, and a decidable predicate~$Q$
  on~$S$, we have that, for all~$s \in S$, $\mathsf{ext}_Q(s) \implies
  \mathsf{int}_Q(s).$
\end{proposition}
\noindent

Now we thus obtain completeness of the \intentional\ predicates.
\begin{corollary}
  \label{cor:ext-int-eq-conv}
  \label{cor:ext-int-eq-must}
  \label{cor:inductive-char-must}\label{cor:inductive-char-conv}        
  For every $\server \in \States$ we have
  \begin{enumerate}
   \item
  $\state \conv $ if and only if $\state \convi$, and
\item
  for every $\client$ we have that $\Must{\server}{\client}$ if
    and only if $\musti{\server}{\client}$.
 \end{enumerate}
\end{corollary}
\begin{proof}
  Direct consequence of \rprop{ext-impl-int}, and
  \req{def-extensional-predicates} and
  \req{def-intentional-predicates} above.%
\end{proof}

As we have outlined why \rcor{inductive-char-must} is true,
from now on we use~$\convi$ and~$\opMusti$ instead of~$\conv$ and~$\opMust$.
In \rapp{properties-intensional-predicates} we prove the properties of
these predicates, that we use in the rest of the paper.

\subsubsection{The LTS of sets}

Recall that soundness of~$\asleq$ means that~$\asleq \;\subseteq\; \testleqS$. 
The naïve reasoning does not work. Fix two servers~$\serverA$ and~$\serverB$
such that $\serverA \asleq \serverB$. We need to prove that for every
client~$\client$, if $\musti{\server}{\client}$ then
$\musti{\serverA}{\client}$.
Rule induction on the predicate $\musti{\server}{\client}$ fails, as
demonstrated by the following example.

\begin{example}
  \label{ex:must-set-is-helpful}
  Consider the two servers $\server = \tau.(\co{\aa} \Par \co{b}) \extc
  \tau.(\co{\aa} \Par \co{c})$ and $\serverB = \co{\aa} \Par (\tau.\co{b} \extc
  \tau.\co{c})$ of \req{mailbox-hoisting}.
  Fix a client $\client$ such that $\musti{\serverA}{\client}$.
  Rule induction 
  yields the following inductive hypothesis:
  \begin{center}
    $\forall \serverA', \serverB' \suchthat\;
    \csys{\serverA}{\client} \st{\tau} \csys{\serverA'}{\client'} \;\land\;$
    $\serverA' \asleq \serverB' \;\Rightarrow\; \musti{\serverB'}{\client'}$.
  \end{center}
  In the proof of
  $\musti{\serverB}{\client}$ we have to consider the case where there
  is a communication between $\serverB$ and
  $\client$ such that, for instance, $\serverB \st{\co{\aa}}
  \tau.\co{b} \extc \tau.\co{c}$ and $\client \st{\aa}
  \client'$.  In that case, we need to show that $\musti{\tau.\co{b}
    \extc \tau.\co{c}\;}{\client'}$. Ideally, we would like to use the inductive
  hypothesis. This requires us to exhibit a $\server'$ such that $
  \csys{\server}{\client} \st{\tau} \csys{\server'}{\client'}$ and $
  \server' \bhvleqtwo \tau.\co{b} \extc \tau.\co{c}$.
  However, note that there is no way to derive
  $\csys{\server}{\client} \st{\tau} \csys{\server'}{\client'}
  $, because $\server
  \Nst{\co{\aa}}$.  The inductive hypothesis thus cannot be applied,
  and the naïve proof does not go through.\hfill$\qed$
\end{example}
\noindent
This example suggests that defining an auxiliary predicate~$\opMustset$ in some sense
equivalent to~$\opMusti$, but that uses explicitly {\em weak} outputs
of servers, should be enough to prove that~$\asleq$ is sound with respect to~$\testleqS$.
Unfortunately, though, there is an additional nuisance to tackle: server
nondeterminism.
\begin{example}
  Assume that we defined the predicate~$\opMusti$
  using weak transitions on the server side.
  Recall the argument 
  put forward in the previous example.
  The inductive hypothesis now becomes the following:
  \begin{center}
    For every $\serverA', \serverB', \mu$ such that
    $\serverA \wt{\mu} \serverA'$ and $\client \st{\mu} \client'$,
    $\serverA' \asleq \serverB'$ implies $\musti{\serverB'}{\client'}$.
  \end{center}
  To use the inductive hypothesis we have to choose a $\server'$ such
  that $\server \wt{\co{\aa}} \server'$ and $\server' \asleq
  \tau.\co{b} \extc \tau.\co{c}$. This is still not enough for the
  entire proof to go through, because (modulo further $\tau$-moves)
  the particular $\server'$ we pick has to be related also to either
  $\co{b}$ or $\co{c}$. It is not possible to find such a
    $\server'$, because
    the two possible candidates
  are either $\co{b}$
  or $\co{c}$; neither of which can satisfy $\server' \asleq
  \tau.\co{b} \extc \tau.\co{c}$, as the right-hand side has not
  committed to a branch yet.


  If instead of a single state $\server$ in the novel definition of
  $\opMusti$ we used a set of 
  states and a suitable
  transition relation, the choice of either $\co{b}$ or $\co{c}$ would be
  suitably delayed. It suffices for instance to have the following states and transitions:
  $\set{\serverA} \wt{\co{\aa}} \set{\co{b}, \co{c}}.$\hfill$\qed$
\end{example}

Now that we have motivated the main intuitions behind the definition of our
novel auxiliary predicate~$\opMustset$, we proceed with the formal definitions.

\begin{definition}[LTS of sets]
Let~$\pparts{ Z }$ be the set of  {\em non-empty} parts of~$Z$.
For any LTS~$\lts{\States}{ L }{~\st{}~}$,
$ X \in \pparts{ \States } $ and $\alpha \in L$, we define the sets
$$
D{(\alpha, X)} =  \setof{ \stateA }{ \exists \state \in X \suchthat \state \st{\alpha} \stateA },\quad
\WD{(\alpha, X)}  =  \setof{ \stateA }{ \exists \state \in X \suchthat \state \wt{\alpha} \stateA }.
$$
We construct the LTS $\lts{\pparts{ \States }}{ \Acttau }{ \st{} }$
by letting $ X \st{ \alpha } D{(\alpha, X)}$ whenever $D{(\alpha, X)} \neq \emptyset$.
Similarly, we have $X \wt{ \alpha } \WD{(\alpha, X)}$ whenever
$\WD{(\alpha, X)} \neq \emptyset$.  \hfill $\blacksquare$
\end{definition}
Intuitively, this definition lifts the standard notion of state derivative to sets of states.
This construction is standard \cite{DBLP:conf/avmfss/CleavelandH89,DBLP:conf/aplas/BonchiCPS13,DBLP:journals/lmcs/BonchiSV22}
  and goes back to the determinisation of nondeterministic automata.

\begin{figure}[t]
  {
  \footnotesize
 \hrulefill
  $$
  \begin{array}{ll}
    \msetnow
    &
    \msetstep
    \\[1pt]
    \begin{prooftree}
      \good{\client}
      \justifies
      \mustset{ X }{r}
    \end{prooftree}
    \hspace{4em}
    &
    \begin{prooftree}
      \begin{array}{lr}
        \lnot \good{\client} & \forall X' \wehavethat X \st{ \tau } X' 
        \implies \mustset{X'}{\client}\\
      \forall\ \serverA \in X \wehavethat \csys{ \serverA }{ \client } \st{\tau} & \forall \client' \wehavethat \client \st{ \tau } \client' \implies \mustset{X}{\client'}
      \\
      \multicolumn{2}{r}{        \forall X', \mu \in \Actfin  
        \wehavethat X \wt{\co{\mu}} X' 
        \text{ and }  \client \st{\mu} \client' \imply 
        \mustset{ X' }{ \client'}}
      \end{array}
      \justifies
      \mustset{ X }{ \client }
    \end{prooftree}
  \end{array}
  $$
  }
  \vspace{-10pt}
  \caption{Rules to define inductively the predicate $\opMustset$.}
  \label{fig:rules-mustset-main}
\hrulefill
\end{figure}

Let $\opMustset$ be defined via the rules in \rfig{rules-mustset-main}.
This predicate lets us reason on~$\opMusti$ via sets of servers,
in the following sense:
\begin{lemma}
  \label{lem:musti-if-mustset-helper}
  For every LTSs $\genlts_A, \genlts_B$ and every
  set of servers $X \in \pparts{\StatesA}$, we have that
  $\mustset{X}{\client}$ if and only if for every $\serverA \in X
  \wehavethat \musti{\serverA}{\client}$.
\end{lemma}
The LTS of sets has two important applications in this paper: 
first, it is used to define the $\mustset{}{}$ relation, on which we rely to prove the
soundness of the characterisation (see \rapp{proof-soundness}).
Additionally, it is used in the definition of the coinductive characterisation,
which is the topic of the next section.

\subsection{The action-based coinductive characterisation}
\label{sec:coind-char}

We conclude this section by introducing a characterisation of the
\mustpreorder that is more practical than $\asleq$, as it allows one
to use the usual coinductive techniques.
In addition, in the asynchronous case where processes are enhanced with
forwarding, being able to use the coinductive proof method allows us to deal
easily with the additional transitions due to forwarding.
As a demonstration, we use this preorder to prove the code hoisting
refinement shown in~\eqref{eq:mailbox-hoisting}.

First, we recall the definition of this alternative preorder,
which, like the other ones, is the same as in the synchronous case
\cite{DBLP:journals/jacm/AcetoH92,DBLP:conf/concur/LaneveP07,DBLP:journals/mscs/BernardiH16}.

\begin{definition}[Coinductive preorder]
\label{def:coinductive-char-main}
For all image-finite LTSs $\genlts_\StatesA$, $\genlts_\StatesB$ and all
$X \in \pparts{\StatesA}, \serverB \in \StatesB$,
  we let the \emph{coinductive preorder} $\coindleq$ be defined as
  the greatest relation such that whenever $X \coindleq \serverB$, the following
  requirements hold: \begin{enumerate}
\item $X \downarrow \text{ implies } \serverB \downarrow$,
\item\label{pt:coind-tau-serverB} For each $\serverB'$ such that $\serverB
\st{\tau} \serverB'$, we have that $X \coindleq \serverB'$,

\item\label{pt:coind-acceptance-sets}
  $X \downarrow$ and $\serverB \stable$ imply that there exist
  $\state \in X$ and $\stateA \in \StatesA$ such that
  $\state \wt{} \stateA \stable$ and $R(\stateA) \subseteq R(q)$,

\item\label{pt:coind-continuations-mu} For any $\mu \in \Act$,
  if $X \cnvalong{\mu}$,
  then for every  $X'$ and $\serverB'$
  such that $X \wt{\mu} X'$ and $\serverB \st{\mu} \serverB'$,
  we have that
  $X' \coindleq \serverB'$.\hfill$\blacksquare$
\end{enumerate}
\end{definition}

\noindent
This preorder characterises~$\testleqS$ when the set~$X$ of
servers is a singleton.
\begin{theorem}
\label{thm:coinductive-char-equiv-main}
For every image-finite LTS $\genlts_\StatesA$, $\genlts_\StatesB \in \obaFB$, every
$\state \in \StatesA$ and $\serverB \in \StatesB$, we have that
$\state \testleqS \serverB$ if and only if $\set{\liftFW{\state}} \coindleq
  FW(\serverB)$.
\end{theorem}
The idea of the proof is to establish that the coinductive preorder
characterises a version of the \mustpreorder that also has a set of servers on
its LHS, and is defined as follows:
\[
  X \testleqSset q \iff \forall t.\; (\forall p \in X, \musti{p}{t}) \implies
  \musti{q}{t}.
\]

Observe that \rdef{coinductive-char-main} is based on single actions, instead of
traces like~$\asleq$, thus it gives us a practical proof method.
To make this point, we now prove the code hoisting refinement \eqref{eq:mailbox-hoisting}. According
to \rthm{coinductive-char-equiv-main}, it suffices to prove:
\begin{equation}
  \label{eq:use-case-coind-preorder}
  \set{ %
    \tau.(\co{\aa} \Par \co{b}) \extc \tau.(\co{\aa} \Par \co{c}) } \;\coindleq\; \co{\aa} \Par (\tau.\co{b} \extc
   \tau.\co{c})
\end{equation}
As for proofs by induction, when using coinduction it is helpful to prove a
more general statement, which yields a useful \emph{coinductive hypothesis}.
This vocabulary corresponds to the proof theoretic point of view of coinduction,
which matches how Coq implements coinductive proofs using \mintinline{coq}{cofix}.
In practice, the prover can use the coinductive hypothesis \emph{after} the
predicate defined coinductively has been unfolded at least once.
In the set-theoretic setting used in \rdef{coinductive-char-main}, this
corresponds to choosing a relation $R$ that is closed under the operations
given in the definition. This is borne out by \rlem{technical-part-proof-by-coinduction} in \rapp{technicalities-coinductive-characterisation}.

\newcommand{\msleqtwo}{\preccurlyeq_{\textsc{m}}}

\section{Preorders based on must-sets and failure refinement}

We now establish the second \emph{standard} characterisation of the \mustpreorder, defined
using \MustSets, again thanks to  \rthm{testleqS-equals-asleq}.
As an application, we relate the \emph{failure refinement} used by the CSP
community to the $\opMust$-preorder.


We begin by defining formally the $\msleq$ preorder, and we relate it to the \mustpreorder.
\renewcommand{\after}[3]{ (#1 \, \mathsf{after} \,  #2, #3) }
For every~$X \subseteq_{\mathit{fin}} \Act$, that is for every finite
set of visible actions, with a slight abuse of notation we write
$\server \mathrel{\opMust} X$ whenever
$\server \wt{ \varepsilon } \server'$ implies that
$\server' \wt{ \mu }$ for some $\mu \in X$, and we say that~$X$ is a
\MustSet of~$\server$.  Let
$\after{ \serverA }{ s }{ \st{} } = \setof{ \serverA' }{ \serverA
  \wt{s} \serverA' }$.  For every $\genlts_\StatesA, \genlts_\StatesB$
and $\serverA \in \StatesA, \serverB \in \StatesB$, let
$\serverA \msleqtwo \serverB$ whenever $\forall \trace \in \Actfin$ we
have that $\serverA \cnvalong{ \trace }$ implies that
$(\forall X \subseteq_{\mathit{fin}} \Act$ if
$\after{\serverA}{\trace}{ \st{}_\StatesA } \mathrel{\opMust} X$ then
$\after{\serverB}{\trace}{ \st{}_\StatesB } \mathrel{\opMust} X).$

\begin{definition}
  \label{def:denicola-char}
  For all $\genlts_A, \genlts_B \in \obaFB$ and servers $\serverA
  \in A$ and $\serverB \in B$, we let $\serverA \msleq \serverB$ whenever
  $\serverA \bhvleqone \serverB \wedge
  \serverA \msleqtwo \serverB$.\hfill$\blacksquare$
\end{definition}

\begin{lemma}
  \label{lem:acceptance-sets-and-must-sets-have-same-expressivity}
  Let $\genlts_A, \genlts_B \in \obaFB$.
  For all $\serverA \in \StatesA$ and $\serverB \in \StatesB $ such that $\liftFW{\serverA} \bhvleqone \liftFW{\serverB}$,
  we have that
  $\liftFW{\serverA} \msleqtwo \liftFW{\serverB}$ if and only if
  $\liftFW{\serverA} \asleqAfw \liftFW{\serverB}$.
\end{lemma}
\noindent
As a direct consequence, we obtain the following result.
\begin{theorem}
  \label{thm:testleqS-equals-mustsetleq}
    Let $\genlts_A, \genlts_B \in \obaFB$.
  For all $\serverA \in \States$  and
  $\serverB \in \StatesB $, we have that
  $\serverA \testleqS \serverB$ if and only if
  $\liftFW{\serverA} \msleq \liftFW{\serverB}$.
\end{theorem}

\newcommand{\failleq}{\mathrel{\leq_{\textsf{fail}}}}

\paragraph{Failure refinement} %
\MustSets\ have been used mainly
by De Nicola and collaborators, for instance
in~\cite{DBLP:journals/lmcs/NicolaM23,DBLP:journals/iandc/BorealeNP02},
and are closely related to the failure refinement proposed
in~\cite{DBLP:journals/jacm/BrookesHR84} by Hoare, Brookes and Roscoe for TCSP (the process
algebra based on Hoare's language CSP
\cite{DBLP:journals/cacm/Hoare83a,DBLP:conf/icalp/Brookes83}).
Following~\cite{DBLP:journals/jacm/BrookesHR84},  
a {\em failure} of a process $\server$ is a pair $(\trace, X)$
such that $p \wt{ \trace } p'$ and $p' \Nst{\mu}$ for all $\mu \in X$.
Then, failure refinement is defined by letting $\serverA\failleq\serverB$
whenever the failures of~$\serverB$ are also failures of $\serverA$.
This refinement was designed to give a denotational semantics to
processes, and mechanisations in Isabelle/HOL have been
developed to ensure that the refinement
is well defined~\cite{HOL-CSP-AFP,DBLP:journals/acta/BaxterRC22}.  Both
Hennessy~\cite[pag. 260]{DBLP:books/daglib/0066919}
and~\cite{Castellan2023}
highlight that the failure
model can be justified operationally via the~$\opMust$ testing
equivalence: it is folklore dating back to~\cite[Section
  4]{DBLP:journals/tcs/NicolaH84} that failure equivalence and~$\testeqS$ coincide.
Thanks to \rthm{testleqS-equals-mustsetleq} we
conclude that in fact~$\testleqS$ 
coincides with the \emph{failure divergence refinement}
\cite{HOL-CSP}, that is, the intersection of~$\failleq$
and~$\bhvleqone$.

\begin{corollary}
  \label{cor:testleqS-equals-failleq}
  Let $\genlts_A, \genlts_B \in \obaFB$.
  For every $\serverA \in \States$  and
  $\serverB \in \StatesB $, we have that
  $\serverA \testleqS \serverB$ if and only if
  $\liftFW{\serverA} \bhvleqone \liftFW{\serverB}$
  and $\liftFW{\serverA} \failleq \liftFW{\serverB}$.
\end{corollary}

\noindent
Thanks to \rcor{asynleq-equals-bhvleq}, we obtain~$\msleqNF$, the analogous of
\rdef{asyn-leq} based on normal forms, and we prove that it characterises
$\testleqS$.

\newcommand{\oh}{\mathcal{O}_p}
\newcommand{\ohmy}{O}
\newcommand{\ogood}{O_{\mathsf{good}}}
\renewcommand{\States}{C}

\section{Related work}
\label{sec:bisimulation}
\label{sec:detailed-related-works}

Here we discuss in detail the works more closely related to the results of this paper.
Further discussion on related work may be found in \rapp{further-related-works}.

The first investigation on the \mustpreorder in an asynchronous
setting was put forth by \cite{DBLP:conf/fsttcs/CastellaniH98}.
While their very clear examples shed light on the preorder, their
alternative preorder (Definition 6 in that paper) is more complicated
than necessary: it uses the standard LTS of \TACCS, its lifting to an LTS
with forwarding, and two somewhat ad-hoc notions: a predicate $\rr{I}$ and a condition on multisets of inputs. 
Moreover that preorder is not complete because of a glitch in the treatment of divergence. The details of the counter-example we found to that completeness
result are given in \rapp{counterexample}.

In~\cite{DBLP:journals/jlp/Hennessy05} Hennessy outlines how to adapt the
approach of \cite{DBLP:conf/fsttcs/CastellaniH98} to a typed
asynchronous $\pi$-calculus. While the LTS with forwarding is
replaced by a Context LTS, the predicates to define the alternative
preorder are essentially the same as those used in the preceding work with
Castellani. Acceptance sets are given in Definition 3.19 there, and
the predicate~$\rr{}$ is denoted~$\searrow$, while the generalised
acceptance sets of \cite{DBLP:conf/fsttcs/CastellaniH98} are given in
Definition 3.20.  Owing to the glitch in the completeness of
\cite{DBLP:conf/fsttcs/CastellaniH98}, it is not clear that Theorem
3.28 of~\cite{DBLP:journals/jlp/Hennessy05} is correct either.


Also the authors of~\cite{DBLP:journals/iandc/BorealeNP02}
consider the \mustpreorder in \ACCS.
There is a major difference between their approach and ours.
When studying theories for asynchronous programs, one can either
\begin{enumerate}[(1)]
  \item \label{first-approach} keep the definitions used for synchronous programs,
    and enhance the LTS with forwarders;
    or
  \item adapt the definitions, and keep the standard LTS.
\end{enumerate}
In the first case, the complexity is moved into the LTS, which
becomes infinite-branching and infinite-state.
In the second case, the complexity is moved into the definitions
used to reason on the LTS (i.e. in the meta-language), and in
particular in the definition of the alternative preorder, which
deviates from the standard one. The authors of
\cite{DBLP:journals/iandc/BorealeNP02} follow the second approach. %
This essentially explains why they employ the standard LTS of \CCS
and to tackle asynchrony they reason on traces via (1) a preorder~$\preceq$ (Table 2 of that paper) that defines on {\em input} actions
the phenomena due to asynchrony; and (2)
a rather technical operation on traces, namely 
  $s \ominus s' = ( \mset{s}_{i} \setminus
  \mset{s'}_{i}) \setminus \co{( \mset{s}_{o} \setminus  \mset{s'}_{o}
    )}$.
We favour instead the approach in~(\ref{first-approach}), for 
it helps 
achieve a modular mechanisation.



The authors of~\cite{DBLP:journals/tcs/NicolaP00} give yet another account of the \mustpreorder.
Even though non-blocking outputs can be written in their calculus,
they use a left-merge operator that allows writing {\em blocking} outputs.
The contexts that they use to prove the completeness of their
alternative preorder use such blocking outputs, consequently
their arguments need not tackle the asymmetric treatment of input and
output actions. This explains why they can use smoothly a standard
LTS, while \cite{DBLP:conf/fsttcs/CastellaniH98} and
\cite{DBLP:journals/iandc/BorealeNP02} have to resort to more
complicated structures.

Theorem 5.3 of the PhD thesis by \cite{DBLP:phd/us/Thati03} states an
alternative characterisation of the \mustpreorder, but it is given
with no proof.  The alternative preorder given in Definition 5.8 of
that thesis turns out to be a mix of the ones by
\cite{DBLP:conf/fsttcs/CastellaniH98} and
\cite{DBLP:journals/iandc/BorealeNP02}.  In particular, the
definition of the alternative preorder relies on the LTS with
forwarding, there denoted~$\st{}_A$ (Point 1. in Definition 5.1
defines exactly the input transitions that forward messages into the
global buffer).  The condition that compares convergence of processes
is the same as in \cite{DBLP:conf/fsttcs/CastellaniH98}, while
server actions are compared using \MustSets, and not \AcceptanceSets.
In fact, Definition 5.7 there is titled ``\AcceptanceSets'' but it
actually defines \MustSets.

\section{Conclusion}
\label{sec:conclusion}
\label{sec:discussion}

In this paper we have shown that the standard characterisations of
  the \mustpreorder by De Nicola and
  Hennessy~\cite{DBLP:journals/tcs/NicolaH84,DBLP:books/daglib/0066919}
  are sound and complete also in an asynchronous setting, provided servers are
  enhanced with the forwarding ability.
  \rlem{liftFW-works} shows that this lifting is always possible.
  We have also shown that the standard coinductive characterisation carries over to the
  asynchronous setting.
  Our results 
  are supported by 
  the first mechanisation of the \mustpreorder, and increase proof
  (i.e. code) factorisation and reusability since the alternative preorders
  do not need to be changed when shifting between synchronous and
  asynchronous semantics: it is enough to parameterise the proofs on
  the set of non-blocking actions.
%
%
%
  \rcor{testleqS-equals-failleq} states that \mustpreorder and failure
  refinement essentially coincide.
  This might spur further interest in
  the mechanisations of failure
  refinement, carried out so far in Isabelle/HOL
  \cite{HOL-CSP-AFP,DBLP:journals/acta/BaxterRC22}, possibly
  opening up opportunities of joint efforts for automated checking.


\paragraph{Proof method for \mustpreorder}
Theorems \ref{thm:testleqS-equals-asleq}, \ref{thm:coinductive-char-equiv-main}
and \ref{thm:testleqS-equals-mustsetleq} endow researchers in programming
languages for message-passing software with a proof method for~$\testleqS$,
namely: to define for their calculi an LTS that enjoys the axioms of
output-buffered agents with feedback. 
An example of this approach is \rcor{characterisation-for-aCCS}.
%

\paragraph{Live programs have barred trees}  We argued that a proof
of $\Must{\server}{\client}$ is a proof of liveness (of the
client). This paper is thus de facto an example that proving
liveness amounts to prove that a computational tree has a bar (identified
by the predicate $\goodSym$), and hence \barinduction is a natural way
to reason constructively on liveness-preserving manipulations on programs.
  While this fact seems to be by and large 
  unexploited by the PL community, we believe that it may 
  be of interest to practitioners reasoning on liveness
  properties 
  in theorem provers in particular, and to the PL
  community at large.

\paragraph{Mechanisation} %
Boreale and Gadducci~\cite{DBLP:journals/tcs/BorealeG06} remark that the
\mustpreorder lacks a tractable proof method.  In constrast, we argue that our
contributions, in particular the coinductive characterisation
(\rthm{coinductive-char-equiv-main}), being fully mechanised in Coq, let
practitioners pursue non-trivial results about testing preorders for real-world
programming languages.  
To make this point, we have proved a form of code-hoisting using this
characterisation.
Our mechanisation lowers the barrier to entry for
researchers versed into theorem provers and wishing to use testing preorders;
adds to the toolkit of Coq users an alternative to the well-known (and already
mechanised) bisimulation equivalence \cite{DBLP:conf/lics/Pous16}; and provides
a starting point for researchers willing to study testing preorders and
analogous refinements within type theory.
Researchers working on testing preorders may benefit from it, as
there are analogies between reasoning techniques for \textsc{May}, $\opMust$,
\textsc{Compliance}, \textsc{Should}, and \textsc{Fair} testing. For instance
Baldan et al. show with pen and paper that a technique similar to forwarding
works to characterise the $\May$-preorder~\cite{DBLP:conf/birthday/BaldanBGV15}.

\paragraph{Future work}  Thanks to Theorems \ref{thm:testleqS-equals-asleq},
\ref{thm:coinductive-char-equiv-main} and \ref{thm:testleqS-equals-mustsetleq}
we can now set out to (1) devise an
axiomatisation of~$\testleqS$ for asynchronous calculi, as done in
\cite{DBLP:journals/fac/HennessyI93,DBLP:journals/iandc/BorealeN95,DBLP:books/daglib/0066919,DBLP:journals/tcs/Hennessy02}
for synchronous ones; (2) study for which asynchronous calculi~$\testleqS$ is a
pre-congruence; (3) machine-check semantic models of subtyping for session types
\cite{DBLP:journals/mscs/BernardiH16}; (4) study the decidability
of~$\testleqS$. %

More in general, given the practical relevance of asynchronous
communication, it seems crucial %
not only to adapt the large body of theory for synchronous
communication to the asynchronous setting but also to resort to
machine supported reasoning to do it. This paper is meant to be a step
in this direction.

\paragraph{Data availability}
The mechanised proofs have been archived on Zenodo\cite{artifact}.

\paragraph{Acknowledgments.}
We thank the anonymous reviewers for their helpful comments, %
and Roberto Amadio for useful remarks. %
  The first author would like to acknowledge %
his ``N+1 honoris causa'' Pierre-Evariste Dagand for the %
systematic and useful feedback, %
Guillaume Geoffroy for
suggesting the use of \barinduction, %
and Hugo Herbelin for
discussions on the topic. %
This work has received funding from the European Research Council (ERC) under
the European Union’s Horizon 2020 research and innovation programme (grant
agreement No. 101003349).

\bibliographystyle{splncs04}%
\bibliography{main.bib}%

\appendix%
\section{Further related works}
\label{sec:further-related-works}

\paragraph{Contextual preorders in functional languages}
Morris preorder is actively studied in
the pure $\lambda$-calculus \cite{DBLP:conf/rta/BreuvartMPR16,DBLP:journals/lmcs/BreuvartMR18,DBLP:journals/lmcs/IntrigilaMP19,DBLP:books/cp/BarendregtM22}, %
$\lambda$-calculus with references \cite{DBLP:conf/icalp/Prebet22,DBLP:conf/fossacs/HirschkoffJP23}, %
in PCF \cite{koutavas13} %
as well as in languages supporting shared memory concurrency \cite{DBLP:conf/popl/TuronTABD13}, and mutable references \cite{DBLP:journals/jfp/DreyerNB12}.
The more sophisticated the languages, the more intricate and larger the proofs.
The need for mechanisation became thus apparent, in particular to prove that
complex logical relations defined in the framework Iris
  (implemented in Coq) 
are sound, \ie included in the preorder \cite{DBLP:conf/popl/KrebbersTB17,DBLP:conf/lics/FruminKB18}.
The paper \cite{DBLP:journals/jlap/AubertV22} provides a framework to study
contextual equivalences in the setting of process calculi.
It is worth noting, though, that as argued in
\cite[Section~$3$]{DBLP:journals/lisp/Boudol98},
Morris equivalence coincides with $\May$-equivalence,
at least if the operational semantics at hand enjoys the Church-Rosser
property. In fact \cite{DBLP:conf/caap/BoudolL96} define Morris
preorder literally as a testing one, via tests for convergence.
The studies of the \mustpreorder in process calculi can thus be seen
as providing proof methods to adapt Morris equivalence to
\nondeterministic settings, and using contexts that are really
external observers. To sum up, one may say that
  Morris equivalence coincides with $\May$-equivalence when
  \nondeterminism\ is confluent and all states are viewed as accepting states,
  while it coincides with $\opMust$ equivalence in the presence of true
  \nondeterminism\ and when only successful states are viewed as accepting states.

  In the setting of nondeterministic and possibly concurrent
  applicative programming languages
  \cite{DBLP:conf/ppdp/Schmidt-Schauss18,DBLP:journals/corr/abs-2008-13359,birkedal-non-determinism},
  also a contextual preorder based on may and must-termination has
  been studied \cite{DBLP:journals/mscs/SabelS08,birkedal-non-determinism}.
  Our preorder $\bhvleqone$ is essentially a
  generalisation of the must-termination preorder of
  \cite{DBLP:journals/mscs/SabelS08} to traces of visible actions.

\paragraph{Theories for synchronous semantics}
Both \cite{DBLP:conf/concur/LaneveP07} and
\cite{DBLP:journals/toplas/CastagnaGP09} employed LTSs as a model of
contracts for web-services (\ie \texttt{WSCL}), and the \mustpreorder
as refinement for contracts.  The idea is that a search engine asked
to look for a service described by a contract~$\server_1$ can
actually return a service that has a contract~$\server_2$, provided
that~$\server_1 \testleqS \server_2$.

The \mustpreorder for {\em clients} proposed by
\cite{DBLP:journals/corr/BernardiH15} has partly informed the
theory of monitors by \cite{DBLP:journals/pacmpl/AcetoAFIL19},
in particular the study of preorders for monitors by
\cite{DBLP:journals/iandc/Francalanza21}.
Our results concern LTSs that are more general than those of monitors,
and thus 
our code could provide the basis to mechanise
the results of \cite{DBLP:journals/pacmpl/AcetoAFIL19}.

The first subtyping relation for binary session types was presented
in~\cite{DBLP:journals/acta/GayH05} using a syntax-oriented definition.
The semantic model of that subtyping is a refinement 
very similar to the \mustpreorder. The idea is to treat types as \textsc{CCS} terms,
assign them an LTS
\cite{DBLP:conf/ppdp/CastagnaDGP09,DBLP:conf/ppdp/Barbanerad10,DBLP:journals/iandc/RavaraRV12,DBLP:journals/mscs/BernardiH16},
and use the resulting testing preorders as semantic models of the
subtyping.
In the setting of coinductively defined higher-order session types, the
correspondence is implicitly addressed
in~\cite{DBLP:conf/ppdp/CastagnaDGP09}.
In the setting of recursive higher-order session types, 
it is given by Theorem 4.10 of~\cite{DBLP:journals/corr/BernardiH13}.

\paragraph{Models of asynchrony}
While synchronous (binary) communication requires
the simultaneous occurrence of a send and a receive action,
asynchronous communication allows a delay between a send action and
the corresponding receive action.  Different models of asynchrony
exist, depending on which medium is assumed for storing messages in
transit. In this paper, following the early work on the asynchronous
$\pi$-calculus~\cite{DBLP:conf/ecoop/HondaT91,boudol:inria-00076939,ACS96},
we assume the medium to be an unbounded
unordered mailbox, shared by all processes. Thus, no process needs to
wait to send a message, namely the send action is non-blocking.
  This model of communication is best captured via the output-buffered
  agents with feedback of \cite{DBLP:conf/concur/Selinger97}. 
  The early style LTS of the asynchronous $\pi$-calculus is a concrete example of this
  kind of LTSs.
A similar global unordered mailbox is used 
also in Chapter~$5$ of \cite{DBLP:phd/us/Thati03}, %
by \cite{Brookes2002DeconstructingCA}, which relies explicitly on a mutable
global state, and by \cite{palamidessi_2003}, which
manipulates it via two functions {\em get} and {\em set}.

More deterministic models of asynchrony are obtained assigning a data
structure to every channel. %
%
For example \cite{HYC08,HYC16} use an even more deterministic model
in which each ordered pair of processes is assigned a dedicated
channel, equipped with an {\em ordered} queue. Hence, messages along
such channels are received in the same order in which they were sent.
  This model is used for asynchronous session calculi, and mimics the
  communication mode of the TCP/IP protocol.  The obvious research
  question here is how to adapt our results to the different
  communication mechanisms and different classes of LTSs. For
  instance, both \cite{Tanti2015TowardsSR} and \cite{caruana19}
  define LTSs for \textsc{Erlang}. We will study whether at least one
  of these LTSs is an instance of output-buffered agents with
  feedback. If this is not the case, we will first try to adapt our
  results to \textsc{Erlang} LTSs.

\paragraph{May-preorder} %
  $\May$ testing and the $\May$-preorder, have been widely studied in 
  asynchronous settings. The first characterisation for \ACCS appeared 
  in \cite{DBLP:conf/fsttcs/CastellaniH98} and relies on comparing traces
  and asynchronous traces of servers. Shortly after
  \cite{DBLP:journals/iandc/BorealeNP02} presented a characterisation based
  on operation on traces.
  A third characterisation appeared in \cite{DBLP:conf/birthday/BaldanBGV15},
  where the saturated LTS $\st{}_s$ is essentially our $\sta{}$.
  That characterisation supports our claim that results about synchronous
  semantics are true also for asynchronous ones, modulo forwarding.
  Compositionality of trace inclusion, i.e. the alternative
  characterisation of the $\May$-preorder, 
  has been partly investigated in Coq by \cite{aathalye-meng} in the
  setting of IO-automata.
  The $\May$-preorder has also been studied in the setting of actor languages
  by \cite{caruana19,Tanti2015TowardsSR}.

\paragraph{Fairness} 
Van Glabbeek \cite{DBLP:conf/fossacs/Glabbeek23} argues that by amending the
semantics of parallel composition (i.e. the scheduler) different
notions of fairness can be embedded in the \mustpreorder. We would like to investigate
which notion of fairness makes the \mustpreorder coincide with the 
\textsc{Fair}-preorder of \cite{DBLP:journals/iandc/RensinkV07}.

  \paragraph{Mutable state}
  Prebet \cite{DBLP:conf/icalp/Prebet22} has recently shown an 
  encoding of the asynchronous $\pi$-calculus into a $\lambda$-calculus
  with references, which captures Morris equivalence via a bisimulation.
  This renders vividly the intuition that output-buffered agents
  manipulate a shared common state.
  We therefore see our work also as an analysis of the \mustpreorder for
  a language in which programs manipulate a global mutable store.
  Since the store is what contains output messages, and our formal
  development shows that only outputs are observable, our results
  suggest that characterisations of testing preorders for impure
  programming languages should predicate over the content of the mutable
  store, \ie the values written by programs.
  Another account of $\pi$-calculus synchronisation via a
    functional programming language is provided in
    \cite{DBLP:journals/corr/abs-2008-13359}, that explains how to
    use Haskell \textsc{M-var}s to implement $\pi$-calculus
    message passing.

\section{Bar-Induction}
\label{sec:bar-induction}


\rsec{visual-intro-bar-induction} is an informal introduction to the intuitions
behind \barinduction, together with a comment on the admissibility of
this principle in Coq.
In \rsec{properties-intensional-predicates} we gather the properties of
the predicates $\cnvalong{}$ and $\opMust$ that we use throughout our technical development.


\subsection{A visual introduction}
\label{sec:visual-intro-bar-induction}

\begin{figure*}[t]
  \hrulefill
  \begin{center}
        \begin{tikzpicture}
        \node[state,rectangle,inner sep=3pt] (a) {$\csys{\server}{\client}$};
        \node[state,below left=+20pt and +20pt of a] (b) {\myspace};
        \node[state,below right=+20pt and +20pt of a] (c) {\scalebox{.6}{$\ok$}};
        \node[state,below left=+20pt and +20pt of b] (d) {\scalebox{.6}{$\ok$}};
        \node[state,below=+20pt of c] (g) {\myspace};
        \node[state,below=+20pt of g] (h) {\myspace};
        \node[state,below right=+20pt and +20pt of b] (e) {\scalebox{.6}{$\ok$}};
        \node[below=+20pt of e] (f) {\rotatebox{90}{$\ldots$}};

        \path[->]
        (a) edge node[above,scale=.9] {$\tau$} (b)
        (b) edge node[above,scale=.9] {$\tau$} (d);
        \path[->]
        (b) edge node[above,scale=.9] {$\tau$} (e)
        (e) edge node[left,scale=.9] {$\tau$} (f);

        \path[->]
        (a) edge node[above,scale=.9] {$\tau$} (c)
        (c) edge node[left,scale=.9] {$\tau$} (g)
        (g) edge node[left,scale=.9] {$\tau$} (h);
        \end{tikzpicture}
  \end{center}

  \caption{The state transition system of client-server system.}
  \hrulefill
  \label{fig:client-server-sts}
\end{figure*}
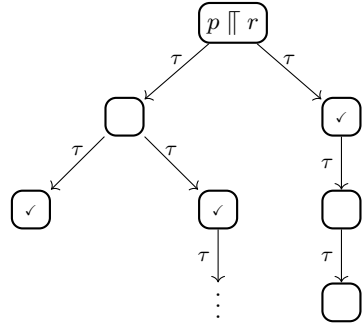

We explain the difference between {\em extensional} definitions of
predicates and {\em \intentional} ones, by discussing how the two
different approaches make us reason on computational trees.

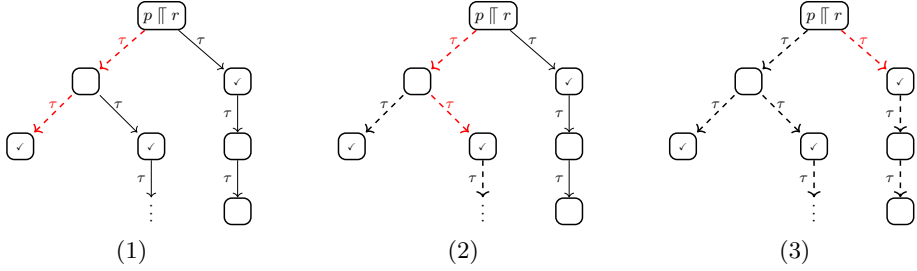
\begin{figure*}[t]
  \hrulefill
  \begin{center}
    \begin{tabular}{c@{\hskip 30pt}c@{\hskip 30pt}c}
      \scalebox{.7}{%
        \begin{tikzpicture}
          \node[state,rectangle,inner sep=3pt] (a) {$\csys{\server}{\client}$};
          \node[state,below left=+20pt and +20pt of a] (b) {\myspace};
          \node[state,below right=+20pt and +20pt of a] (c) {\scalebox{.6}{$\ok$}};
          \node[state,below left=+20pt and +20pt of b] (d) {\scalebox{.6}{$\ok$}};
          \node[state,below=+20pt of c] (g) {\myspace};
          \node[state,below=+20pt of g] (h) {\myspace};
          \node[state,below right=+20pt and +20pt of b] (e) {\scalebox{.6}{$\ok$}};
          \node[below=+20pt of e] (f) {\rotatebox{90}{$\ldots$}};

          \path[->]
          (a) edge[thick, dashed, red] node[above,scale=.9] {$\tau$} (b)
          (b) edge[thick, dashed, red] node[above,scale=.9] {$\tau$} (d);
          \path[->]
          (b) edge node[above,scale=.9] {$\tau$} (e)
          (e) edge node[left,scale=.9] {$\tau$} (f);

          \path[->]
          (a) edge node[above,scale=.9] {$\tau$} (c)
          (c) edge node[left,scale=.9] {$\tau$} (g)
          (g) edge node[left,scale=.9] {$\tau$} (h);
        \end{tikzpicture}
      }
      &
            \scalebox{.7}{%
        \begin{tikzpicture}
          \node[state,rectangle,inner sep=3pt] (a) {$\csys{\server}{\client}$};
          \node[state,below left=+20pt and +20pt of a] (b) {\myspace};
          \node[state,below right=+20pt and +20pt of a] (c) {\scalebox{.6}{$\ok$}};
          \node[state,below left=+20pt and +20pt of b] (d) {\scalebox{.6}{$\ok$}};
          \node[state,below=+20pt of c] (g) {\myspace};
          \node[state,below=+20pt of g] (h) {\myspace};
          \node[state,below right=+20pt and +20pt of b] (e) {\scalebox{.6}{$\ok$}};
          \node[below=+20pt of e] (f) {\rotatebox{90}{$\ldots$}};

          \path[->]
          (a) edge[thick, dashed,red] node[above,scale=.9] {$\tau$} (b)
          (b) edge[thick, dashed] node[above,scale=.9] {$\tau$} (d);
          \path[->]
          (b) edge[thick, dashed, red] node[above,scale=.9] {$\tau$} (e)
          (e) edge[thick, dashed] node[left,scale=.9] {$\tau$} (f);

          \path[->]
          (a) edge node[above,scale=.9] {$\tau$} (c)
          (c) edge node[left,scale=.9] {$\tau$} (g)
          (g) edge node[left,scale=.9] {$\tau$} (h);
        \end{tikzpicture}
      }
            &
                  \scalebox{.7}{%
        \begin{tikzpicture}
          \node[state,rectangle,inner sep=3pt] (a) {$\csys{\server}{\client}$};
          \node[state,below left=+20pt and +20pt of a] (b) {\myspace};
          \node[state,below right=+20pt and +20pt of a] (c) {\scalebox{.6}{$\ok$}};
          \node[state,below left=+20pt and +20pt of b] (d) {\scalebox{.6}{$\ok$}};
          \node[state,below=+20pt of c] (g) {\myspace};
          \node[state,below=+20pt of g] (h) {\myspace};
          \node[state,below right=+20pt and +20pt of b] (e) {\scalebox{.6}{$\ok$}};
          \node[below=+20pt of e] (f) {\rotatebox{90}{$\ldots$}};

          \path[->]
          (a) edge[thick, dashed] node[above,scale=.9] {$\tau$} (b)
          (b) edge[thick, dashed] node[above,scale=.9] {$\tau$} (d);
          \path[->]
          (b) edge[thick, dashed] node[above,scale=.9] {$\tau$} (e)
          (e) edge[thick, dashed] node[left,scale=.9] {$\tau$} (f);

          \path[->]
          (a) edge[thick, dashed, red] node[above,scale=.9] {$\tau$} (c)
          (c) edge[thick, dashed] node[left,scale=.9] {$\tau$} (g)
          (g) edge[thick, dashed] node[left,scale=.9] {$\tau$} (h);
        \end{tikzpicture}
      }
      \\
      \small (1) &
      \small (2) &
      \small (3)
    \end{tabular}
  \end{center}
  \caption{Extensional approach: finding successful prefixes in every
    maximal path of the computational tree.}
  \label{fig:extensional-path-by-path}
  \hrulefill
\end{figure*}

Suppose that we have a client-server system $\csys{\server}{\client}$
and that we want to prove either$\Must{\server}{\client}$
or $\musti{\server}{\client}$. For both proofs,
what matters is the state transition system (STS) of
$\csys{\server}{\client}$, i.e. 
the computation steps performed by the client-server system at issue.
In fact it is customary to treat this STS as a computational tree,
as done for instance in the proofs of \cite[Lemma 4.4.12]{DBLP:books/daglib/0066919} and
\cite[Theorem 2.3.3]{TCD-CS-2010-13}.
In the rest of this subsection we discuss the tree depicted in \rfig{client-server-sts}.
It contains three maximal computations, the middle one being infinite.
In the figures of this subsection, the states in which the client
is successful (i.e. in the predicate $\goodSym$) contain the symbol $\ok$.


\paragraph{The extensional approach}
To prove $\Must{\server}{\client}$, the extensional definition of
$\opMust$ requires checking that every maximal path in the tree in
\rfig{client-server-sts} starts with a finite prefix that leads to a
successful state.
The proof that $\Must{\server}{\client}$ amounts to looking
for a suitable prefix maximal path by maximal path, via a loop whose
iterations are suggested in \rfig{extensional-path-by-path}.
There at every iteration a different maximal path (highlighted by
dashed arrows) is checked, and each time a successful prefix is found
(indicated by a red arrow), the loop moves on to the next  maximal
path. Once a maximal path is explored, it remains dashed, to denote
that there a succesful prefix has been found.
The first iteration looks for a successful prefix in the left-most
maximal path, while the last iteration looks
for a successful prefix in the right-most path.
In the current example the loop terminates because the tree in
\rfig{client-server-sts} has conveniently a finite number of maximal
paths, but in general the mathematical reasoning has to deal with
an infinite amount of maximal path. An archetypal example is the tree in
\rfig{unbounded}: it has countably many maximal paths, each one
starting with a successful prefix.

\paragraph{The \intentional\ approach}
Consider now the predicate $\opMusti$ - which is defined \intentional ly -
and a proof that $\musti{\server}{\client}$.
The base case of $\opMusti$ ensures that all the nodes that contain a
successful client (i.e. that satisfies the predicate $Q_2$, defined on
line 553 of the submission) are in $\opMusti$. Pictorially, this is the
step from (1) to (2) in \rfig{intensional-layer-by-layer}, where the
nodes in $\opMusti$ are drawn using dashed borders, and the freshly
added ones are drawn in red.
Once the base case is established, the inductive rule of $\opMusti$
ensures that any node that inevitably goes to nodes that are in
$\opMusti$, is also in the predicate $\opMusti$. This leads to the step
from (2) to (3) and then from (3) to (4).
Note that the argument is concise, for in the tree the depth at which
successful states can be found is finite. In general though is may not
be the case. The tree in \rfig{unbounded} is again the archetypal
example: every maximal path there contains a finite prefix that leads
to a successful state, but there is no upper bound on
the length on those prefixes.

\begin{figure*}[t]
  \hrulefill
  \begin{center}
    \begin{tabular}{c@{\hskip 30pt}c}
      \scalebox{.7}{%
        \begin{tikzpicture}
          \node[state,rectangle,inner sep=3pt] (a) {$\csys{\server}{\client}$};
          \node[state,below left=+20pt and +20pt of a] (b) {\myspace};
          \node[state,below right=+20pt and +20pt of a] (c) {\scalebox{.6}{$\ok$}};
          \node[state,below left=+20pt and +20pt of b] (d) {\scalebox{.6}{$\ok$}};
          \node[state,below=+20pt of c] (g) {\myspace};
          \node[state,below=+20pt of g] (h) {\myspace};
          \node[state,below right=+20pt and +20pt of b] (e) {\scalebox{.6}{$\ok$}};
          \node[below=+20pt of e] (f) {\rotatebox{90}{$\ldots$}};

          \path[->]
          (a) edge node[above,scale=.9] {$\tau$} (b)
          (b) edge node[above,scale=.9] {$\tau$} (d);
          \path[->]
          (b) edge node[above,scale=.9] {$\tau$} (e)
          (e) edge node[left,scale=.9] {$\tau$} (f);

          \path[->]
          (a) edge node[above,scale=.9] {$\tau$} (c)
          (c) edge node[left,scale=.9] {$\tau$} (g)
          (g) edge node[left,scale=.9] {$\tau$} (h);
        \end{tikzpicture}
      }
      &
      \scalebox{.7}{%
        \begin{tikzpicture}
          \node[state,rectangle,inner sep=3pt] (a) {$\csys{\server}{\client}$};
          \node[state,below left=+20pt and +20pt of a] (b) {\myspace};
          \node[state,dashed,red,thick,below right=+20pt and +20pt of a] (c) {\scalebox{.6}{$\ok$}};
          \node[state,dashed,red,thick,below left=+20pt and +20pt of b] (d) {\scalebox{.6}{$\ok$}};
          \node[state,dashed,red,thick,below right=+20pt and +20pt of b] (e)
               {\scalebox{.6}{$\ok$}};
               \node[below=+20pt of e] (f) {\phantom{\rotatebox{90}{$\ldots$}}};

          \path[->]
          (a) edge node[above,scale=.9] {$\tau$} (b)
          (b) edge node[above,scale=.9] {$\tau$} (d);
          \path[->]
          (b) edge node[above,scale=.9] {$\tau$} (e);

          \path[->]
          (a) edge node[above,scale=.9] {$\tau$} (c);
        \end{tikzpicture}
      }
      \\
      \small (1) &
      \small (2)
      \\[2.5em]
      \scalebox{.7}{%
                \begin{tikzpicture}
          \node[state,rectangle,inner sep=3pt] (a) {$\csys{\server}{\client}$};
          \node[state,dashed,red,thick,below left=+20pt and +20pt of a] (b) {\myspace};
          \node[state,dashed,below right=+20pt and +20pt of a] (c) {\scalebox{.6}{$\ok$}};
          \node[state,dashed,below left=+20pt and +20pt of b] (d) {\scalebox{.6}{$\ok$}};
          \node[state,dashed,below right=+20pt and +20pt of b] (e)
               {\scalebox{.6}{$\ok$}};
               \node[below=+20pt of e] (f) {\phantom{\rotatebox{90}{$\ldots$}}};

          \path[->]
          (a) edge node[above,scale=.9] {$\tau$} (b)
          (b) edge node[above,scale=.9] {$\tau$} (d);
          \path[->]
          (b) edge node[above,scale=.9] {$\tau$} (e);

          \path[->]
          (a) edge node[above,scale=.9] {$\tau$} (c);
                \end{tikzpicture}

      }
      &
      \scalebox{.7}{%
                \begin{tikzpicture}
          \node[state,dashed,red,thick,rectangle,inner sep=3pt] (a) {$\csys{\server}{\client}$};
          \node[state,dashed,below left=+20pt and +20pt of a] (b) {\myspace};
          \node[state,dashed,below right=+20pt and +20pt of a] (c) {\scalebox{.6}{$\ok$}};
          \node[state,dashed,below left=+20pt and +20pt of b] (d) {\scalebox{.6}{$\ok$}};
          \node[state,dashed,below right=+20pt and +20pt of b] (e)
               {\scalebox{.6}{$\ok$}};
               \node[below=+20pt of e] (f) {\phantom{\rotatebox{90}{$\ldots$}}};

          \path[->]
          (a) edge node[above,scale=.9] {$\tau$} (b)
          (b) edge node[above,scale=.9] {$\tau$} (d);
          \path[->]
          (b) edge node[above,scale=.9] {$\tau$} (e);

          \path[->]
          (a) edge node[above,scale=.9] {$\tau$} (c);
                \end{tikzpicture}
      }
      \\
      \small (3) &
      \small (4)
    \end{tabular}
  \end{center}
  \caption{\Intentional\ approach: visiting the tree bottom-up, starting from the bar.}
  \label{fig:intensional-layer-by-layer}
  \hrulefill
\end{figure*}
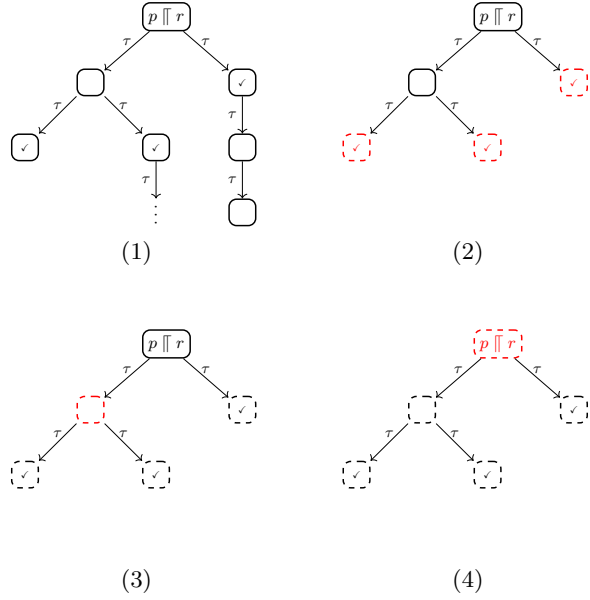

\paragraph{Do \extensional and  \intentional predicates coincide ?}
Extensional and \intentional\ definitions make us reason on
computational trees in strikingly different fashions:
{\em extensionally} we reason maximal path by maximal path,
while {\em \intentional ly} we reason bottom-up, starting
from the nodes in a predicate that bars the tree.\footnote{Whence the
name {\em bar}-induction.}
It is natural to ask whether reasoning in these different manners
ultimately leads to the same outcomes.
In our setting this amounts to proving that the predicates
$\opMust$ and $\opMusti$ are logically equivalent, and similarly for
the convergence predicates $\conv$ and $\convi$.
The proof that $\musti{\server}{\client}$ implies
$\Must{\server}{\client}$ is - obviously - by induction on the derivation
of $\musti{\server}{\client}$.
Proving that the extensional predicates imply the \intentional\ ones
is, on the other hand, delicate, because we may have to deal with unbounded structures.
The tree in \rfig{unbounded} is once more the archetypal example:
it has countably many maximal paths, and there is no upper bound on the
depth at which successful states (i.e. nodes in the {\em bar}) are found.

In classical logic one can prove that
$\Must{\server}{\client}$ implies $\musti
{\server}{\client}$ by contradiction.
As we wish to avoid this reasoning principle, the only tool we have
is the axiom of \Barinduction, which states exactly that under
suitable hypotheses, extensionally defined predicates imply their
\intentional ly defined counter-parts.

\begin{figure}[t]
\hrulefill
\begin{center}
\begin{tikzpicture}
  \node[state](t){$\csys{\server}{\client}$};

  \node[state][below=  of t](p20){\myspace};
  \node[state][below=  of p20](p21){\myspace};
  \node[state][below=  of p21](p22){\scalebox{.6}{$\ok$}};

  \node[state][left=of p20] (p10) {\myspace};
  \node[state][below=of p10] (p11) {\scalebox{.6}{$\ok$}};

  \node[state][left =of p10] (p00) {\scalebox{.6}{$\ok$}};

  \node[state][right=of p20](p30){\myspace};
  \node[state][below=of p30](p31){\myspace};
  \node[state][below=of p31](p32){\myspace};
  \node[state][below=of p32](p33){\scalebox{.6}{$\ok$}};

  \node[][right=of p30](p5){\vdots\ldots\vdots};

\path (t) edge[to] node[action,swap,scale=.9] {$\tau$} (p00)
      (t) edge[to] node[action,scale=.9] {$\tau$} (p10)
      (t) edge[to] node[action,scale=.9] {$\tau$} (p20)
      (t) edge[to] node[action,swap,scale=.9] {$\tau$} (p30)
      (t) edge[to] node[action,scale=.9] {$\tau$} (p5);

\path (p30) edge[to] node[action,scale=.9] {$\tau$} (p31)
      (p31) edge[to] node[action,scale=.9] {$\tau$} (p32)
      (p32) edge[to] node[action,scale=.9] {$\tau$} (p33)
      (p20) edge[to] node[action,scale=.9] {$\tau$} (p21)
      (p21) edge[to] node[action,scale=.9] {$\tau$} (p22)
      (p10) edge[to] node[action,scale=.9] {$\tau$} (p11);
\end{tikzpicture}
\end{center}
  \caption{An infinite branching computational tree where the
    {\em bar}\ $\ok$ is at unbounded depth.}
  \label{fig:unbounded}
\hrulefill
\end{figure}
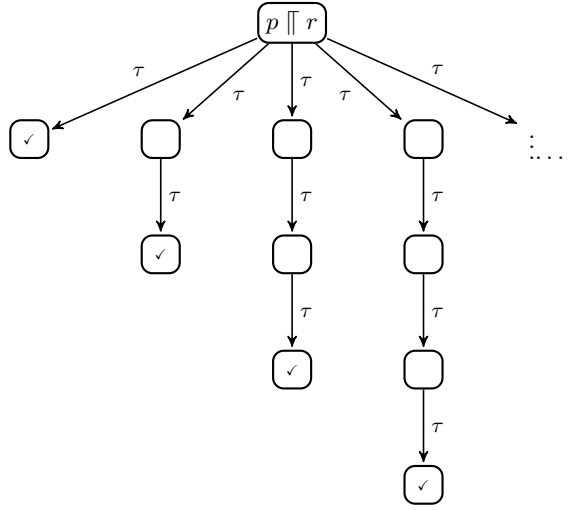

\paragraph{Admissibility.}
To show that the principle is admissible, we prove that it follows
from the Classical Epsilon (CE) axiom of the Coq standard library.
In short, CE gives a choice function $\epsilon$ such that if $p$ is a proof of
$\exists x: A, P x$, then $\epsilon(p)$ is an element of $A$ such that $P
(\epsilon(p))$ holds. It implies Excluded Middle, and thus classical reasoning,
because $A \vee \neg A$ is equivalent to $\exists b: bool, (b=true \land A) \vee
(b=false \land \neg A)$.
Since CE is guaranteed by the Coq developers to be admissible, our statement of
\barinduction is also admissible.


\subsection{Properties of inductively defined predicates}
\label{sec:properties-intensional-predicates}
Convergence along traces is preserved by the strong transitions $\st{}$.

\begin{lemma}
  \label{lem:acnv-aux}
  In every LTS,
  $\forevery \state, \stateA \in \States$ and $ \trace \in \Actfin$
  the following facts are true,
  \begin{enumerate}
    \item
       if $\state \cnvalong \trace$
      and $\state \st{\tau} \stateA$ then $\stateA \cnvalong \trace$,\label{pt:acnv-one-step-tau}
    \item
      $\forevery \mu \in \Act \wehavethat \state \cnvalong
      \mu.\trace$ and $ \state \st{\mu} \stateA \imply \state \cnvalong \trace$.\label{pt:acnv-one-step-mu}
  \end{enumerate}
\end{lemma}

\begin{lemma}
  \label{lem:cnvalong-implies-finite-branching}
  For every $\trace \in \Actfin$ and
  $\state \in \ACCS$, if $ \state \cnvalong \trace$ then
  ${\cardinality{\setof{ \stateB }{\state \wt{ \trace} \stateB} }} \in \N.$
\end{lemma}
The hypothesis of convergence in \rlem{cnvalong-implies-finite-branching} is necessary.
This is witnessed by the process $p =  \rec{(x \Par \mailbox{\co{a}})}$,
which realises an ever lasting addition of a message to the mailbox:
$$
p \st{\tau}
p \Par \mailbox{\co{a}} \st{\tau}  p \Par \mailbox{\co{a} \Par \co{a}}
\st{\tau}  p \Par \mailbox{\co{a} \Par \co{a} \Par \co{a}} \st{\tau}
\ldots
$$
In more general languages also image-finiteness may fail. An example
is given on page 267 of \cite{DBLP:conf/mfcs/HennessyP80}.


The predicate $\opMusti$ is preserved by atoms freely changing their locations in systems.
This is coherent with the intuition that the mailbox is a global and
shared one. For instance the systems
$ \csys{ a.\Nil \Par \mailbox{ \co{d} } }{ d.\Unit }$  and  $ \csys{  a.\Nil }{ d.\Unit
  \Par \mailbox{\co{d}} }$, which in the mechanisation are respectively
\begin{minted}{coq}
  (pr_par (pr_input a pr_nil) (pr_out d), pr_input d pr_succes)
\end{minted}
and
\begin{minted}{coq}
  (pr_input a pr_nil, pr_par (pr_input a pr_succes) (pr_out d))
\end{minted}
have the same mailbox, namely $\mailbox{ \co{d} }$.


The predicate $\opMusti$ enjoys three useful properties:
it ensures convergence of servers interacting with clients
that are not in a good state; it is preserved by internal computation
of servers; and it is preserved also by interactions with unhappy clients.
The arguments to show these facts are by rule induction on the
hypothesis $\musti{ \server }{\client }$.
The last fact is a consequence of a crucial property of $\opMusti$,
namely \rlem{must-output-swap-l-fw}.

\begin{lemma}
  \label{lem:must-terminate}
  Let $\genlts_A \in \obaFW$ and $\genlts_B \in \obaFB$.
  For every $\server \in \StatesA$, $\client \in \StatesB$ we have that
  $\musti{ \server }{\client }$ implies that $\server \convi$ or $\good{\client}$.
\end{lemma}

\begin{lemma}
  \label{lem:must-lts-a-tau}
  \label{lem:musti-preserved-by-left-tau}
  Let $\genlts_A \in \obaFW$ and $\genlts_B \in \obaFB$.
  For every $\server, \server' \in \StatesA$, $\client \in \StatesB$ we have that
  $\musti{\server}{\client} $ and $\server \st{\tau} \server'$ imply
  $\musti{\serverB}{\client}$.
\end{lemma}

\begin{lemma}
  \label{lem:good-preserved-by-lts-output-iff}
  For every $\genlts_B \in \oba$, $\client \in \StatesB$
  and name $\aa \in \Names$ such that
  $\serverA \st{\co{\aa}} \serverA'$ then
  $\good{\serverA}$ iff $\good{\serverA'}$.
\end{lemma}
\begin{proof}
  This is a property of \lstinline{Good}, more specifically
  \lstinline{good_preserved_by_lts_output} and \lstinline{good_preserved_by_lts_output_converse}.
\end{proof}

\noindent%

The next lemma states that when reasoning on~$\opMusti$,
outputs can be freely moved from the client to the server side of
systems, if servers 
have the forwarding ability.
Its proof uses {\em all} the axioms for output-buffered agents with
feedback, and the lemma itself is used in the proof of the main
result on acceptance bsets, namely \rlem{completeness-part-2.2-auxiliary}.
\begin{lemma}[ Output swap ]
  \label{lem:must-output-swap-l-fw}
  Let $\genlts_A \in \obaFW$ and
  $\genlts_B \in \obaFB$.
  $\Forevery \serverA_1, \serverA_2 \in \StatesA$,
  every $\client_1, \client_2 \in \StatesB$ and name $\aa \in \Names$ such that
  $\serverA_1 \st{\co{\aa}} \serverA_2$ and
  $\client_1 \st{\co{\aa}} \client_2$,
  if $\musti{\serverA_1}{\client_2}$ then $\musti{\serverA_2}{\client_1}$.
\end{lemma}
\begin{proof}
  We proceed by induction on $\musti{\serverA_1}{\client_2}$.
  In the base case  $\musti{\serverA_1}{\client_2}$ is derived using the rule \mnow
  and thus $\good{\client_2}$.
  \rlem{good-preserved-by-lts-output-iff} implies that $\good{\client_1}$,
  and so we prove $\musti{\serverA_2}{\client_1}$ using rule \mnow.
  We are done with the base case.

  In the inductive case, the hypothesis $\musti{\serverA_2}{\client_1}$ has been derived
  via an rule \mstep, and we therefore know the following facts:
  \begin{enumerate}
  \item
    \label{must-output-swap-l-fw-h-2-1}
    \label{pt:output-swap-inductive-case-fact-1}
    $\csys{\serverA_1}{\client_2} \st{\tau} \csys{\hat{\server}}{\hat{\client}} $, and
  \item
    \label{pt:output-swap-inductive-case-fact-2}
    $\Forevery \serverA', \client'$ such that
    $\csys{\serverA_1}{\client_2}\st{\tau}\csys{\serverA'}{\client'}$ we have that
    $\musti{\serverA'}{\client'}$.
  \end{enumerate}

  We prove $\musti{\server_2}{\client_1}$ by applying rule \mstep. In turn this requires us to show that \begin{enumerate}[(i)]
  \item
    \label{must-output-swap-l-fw-g-1}
    $\csys{\server_2}{\client_1} \st{\tau}$, and that
  \item
    \label{must-output-swap-l-fw-g-2}
    for each $\server'$ and $\client'$ such that
    $\csys{\server_2}{\client_1} \st{\tau} \csys{\serverA'}{\client'}$,
    we have $\musti{\server'}{\client'}$.
  \end{enumerate}

  We prove (\ref{must-output-swap-l-fw-g-1}). The argument starts with
  a case analysis on how the transition
  (\ref{pt:output-swap-inductive-case-fact-1}) has been derived.
  There are the following three cases:
    \begin{description}
  \item[\stauserver]
    \label{must-output-swap-l-fw-g-1-1}
    \label{pt:must-output-proof-of-i-stauserver}
    a $\tau$-transition performed by the server such that
    $\serverA_1 \st{\tau} \hat{\server}$ and that $\hat{\client} = \client_2$, or
  \item[\stauclient]
    \label{must-output-swap-l-fw-g-1-2}
    \label{pt:must-output-proof-of-i-stauclient}
    a $\tau$-transition performed by the client such that
    $\client_2 \st{\tau} \hat{\client}$ and that $\hat{\server} = \server_1$, or
  \item[\scom]
    \label{must-output-swap-l-fw-g-1-3}
    \label{pt:must-output-proof-of-i-scom}
    an interaction between the server $\server_1$ and the client
    $\client_2$ such that 
    $\serverA_1 \st{\mu} \hat{\server}$ and that $\client_2 \st{\co{\mu}} \hat{\client}$.
    \end{description}

    In case \stauserver\ 
    we use the \outputtau axiom together with
    the transitions $\server_1 \st{\co{\aa}} \server_2$ and
    $\server_1 \st{\tau} \hat{\serverA}$
    to obtain that either:
    \begin{itemize}
  \item
    there exists a $\server_3$ such that $\serverA_2 \st{\tau} \server_3$ and
    $\hat{\server} \st{\co{\aa}} \server_3$, or
  \item
    $\server_2 \st{\aa} \server_3$.
  \end{itemize}
    In the first case $\server_2 \st{\tau} \server_3$ let us construct the transition
  $\csys{\serverA_2}{\client_1} \st{\tau} \csys{\server_3}{\client_1}$ as required.
    In the second case recall that by hypothesis $\client_1 \st{\co{\aa}}
  \client_2$, and thus the transition  $\server_2 \st{\aa}
  \hat{\server}$ and rule \scom\ let us construct the desired reduction
    $\csys{\server_2}{\client_1} \st{\tau} \csys{\hat{\server}}{\client_2}$.

  In case \stauclient\ 
  we use the \outputcommutativity axiom together with
  the transitions $\client_1 \st{\co{\aa}} \client_2 \st{\tau} \hat{\client}$
  to obtain a $\client_3$ such that $\client_1 \st{\tau} \client_3 \st{\co{\aa}} \hat{\client}$
  and it follows that there exists the silent move $\csys{\serverA_2}{\client_1} \st{\tau} \csys{\serverA_2~}{~\client_3}$.

  In case \scom\ 
  we have that $\serverA_1 \st{\mu} \hat{\server}$ and $\client_2 \st{\co{\mu}} \hat{\client}$.
  We distinguish whether $\mu = \co{\aa}$ or not.
  If $\mu = \co{\aa}$ then observe that
  $\client_1 \st{\co{\aa}} \client_2 \st{\aa} \hat{\client}$.
  Since by hypothesis $\client_1,\client_2
  \in \StatesB$ and $\genlts_B \in \obaFB$ we apply 
  \outputfeedback axiom to these transitions and obtain $\client_1 \st{\tau}
  \hat{\client}$. An application of \scom\ let us construct
  the desired transition $\csys{\server_2}{\client_1} \st{\tau}
  \csys{\serverA_2}{\hat{\client}}$.

  If $\mu \neq \co{\aa}$ we apply the \outputconfluence axiom to the transitions
  $\server_1 \st{\co{\aa}} \server_2$ and $\serverA_1 \st{\mu} \hat{\server}$
  to obtain a $\server_3$ such that
  $\serverA_2 \st{\mu} \server_3$ and $\hat{\server} \st{\co{\aa}} \server_3$.
  We then apply the \outputcommutativity axiom to obtain
  $\client_1 \st{\co{\mu}} \client_3 \st{\co{\aa}} \hat{\client}$ for
  some $\client_3$.
  Finally, we have the desired $\csys{\server_2}{\client_1} \st{\tau} \csys{\hat{\server}}{\client_3}$
  thanks to the existence of an interaction between $\server_2$ and $\client_1$
  that follows from
  $\server_2 \st{\mu} \server_3$ and $\client_1 \st{\co{\mu}} \client_2$.
  This concludes the proof of  (\ref{must-output-swap-l-fw-g-1}).

  We now tackle (\ref{must-output-swap-l-fw-g-2}). First of all,
  note that the inductive hypothesis states the following fact,
  \begin{center}
  For every
  $\serverA', \client', \serverA_0$ and $\client_0$,
  such that
  $\csys{\serverA_1}{\client_2}\st{\tau}\csys{\serverA'}{\client'}$,
  $\serverA'\st{\co{\aa}}\serverA_0$
  and
  $\client_0\st{\co{\aa}}\client'$
  then
  $\musti{\serverA_0}{\client_0}$.
  \end{center}
  
  Fix a transition
  $$
  \csys{\serverA_2}{\client_1} \st{\tau} \csys{\serverA'}{\client'},
  $$
  we must show $\musti{\serverA'}{\client'}$.
  We proceed by case analysis on the rule used to derive the
  transition at issue, and the cases are as follows,
  \begin{enumerate}[(a)]
  \item
    \label{must-output-swap-l-fw-g-2-1}
    a $\tau$-transition performed by the server such that
    $\serverA_2 \st{\tau} \serverA'$ and that $\client' = \client_1$, or
  \item
    \label{must-output-swap-l-fw-g-2-2}
    a $\tau$-transition performed by the client such that
    $\client_1 \st{\tau} \client'$ and that $\serverA' = \serverA_2$, or
  \item
    \label{must-output-swap-l-fw-g-2-3}
    an interaction between the server $\serverA_2$ and the client
    $\client_1$ such that
    $\serverA_2 \st{\mu} \serverA'$
    and that 
    $\client_1 \st{\co{\mu}} \client'$.
  \end{enumerate}

    In case (\ref{must-output-swap-l-fw-g-2-1})
  we have $\server_2 \st{\tau} \server'$
  and $\client' = \client_1$ and hence we must show $\musti{\server'}{\client_1}$.
  We apply the \outputcommutativity axiom to
  the transitions $\server_1 \st{\co{\aa}} \server_2 \st{\tau} \server'$
  to obtain a $\server_3$ such that
  $\server_1 \st{\tau} \server_3 \st{\co{\aa}} \server'$.
  We apply the inductive hypothesis with
  $\serverA' = \server_3, \client' = \client_2, \server_0 = \server'$
  and $\client_0 = \client_1$ and obtain
  $\musti{\serverA_2}{\client_1}$ as required.

    In case (\ref{must-output-swap-l-fw-g-2-2})
  we have $\client_1 \st{\tau} \client'$ 
  and $\server' = \server_2$, we therefore must show $\musti{\server_2}{\client'}$.
  We apply the \outputtau axiom to
  the transitions $\client_1 \st{\tau} \client'$ and $\client_1 \st{\co{\aa}} \client_2$
  to obtain that
  \begin{itemize}
  \item either
    there exists a $\hat{\client}$ such that $\client_2 \st{\tau} \hat{\client}$ and
    $\client' \st{\co{\aa}} \hat{\client}$,
  \item or 
    $\client_2 \st{\aa} \client'$.
  \end{itemize}
  In the first case we apply the inductive hypothesis with
  $\serverA' = \serverA_1, \client' = \hat{\client}, \serverA_0 = \serverA_2$
  and $\client_0 = \client'$ and obtain $\musti{\serverA_2}{\client'}$ as required.
  In the second case,
  the transitions $\server_1 \st{\co{\aa}} \server_2$ and $\client_2
  \st{\aa} \client'$ and rule \scom\ let us prove
  $\csys{\server_1}{\client_2} \st{\tau} \csys{\server_2}{\client'}$.
  We apply \rpt{output-swap-inductive-case-fact-2} to
  obtain $\musti{\serverA_2}{\client'}$ as required.

    We now consider the case (\ref{must-output-swap-l-fw-g-2-3}) in which
  $\serverA_2 \st{\mu} \serverA'$ and $\client_1 \st{\co{\mu}} \client'$.
  We must show $\musti{\serverA'}{\client'}$ and to do so we distinguish
  whether $\mu = \aa$ or not.

  If $\mu = \aa$ then we apply the $\outputdeterminacy$ axiom to the transitions
  $\client_1 \st{\co{\aa}} \client_2$ and $\client_1 \st{\co{\mu}} \client'$ to obtain
  that $\client_2 = \client'$.
  Since by hypothesis $\server_1, \server_2 \in \StatesA$ and $\genlts_\StatesA \in
  \obaFW$ we apply the \fwdfeedback\ axiom to the transitions
  $\server_1 \st{\co{\aa}} \server_2 \st{\aa} \server'$ to prove that 
  either $\server_1 \st{\tau} \server'$ or $\server_1 = \serverA'$ must hold.
  If $\server_1 \st{\tau} \server'$ then we have that
  $\csys{\server_1}{\client_2} \st{\tau} \csys{\server'}{\client_2}$.
  The property in (\ref{pt:output-swap-inductive-case-fact-2}) ensures that
  $\musti{\server'}{\client_2}$ and from $\client_2 = \client'$ we have that
  the required $\musti{\server'}{\client'}$ holds too.
  If $\serverA_1 = \serverA'$ then $\musti{\serverA'}{\client_2}$
  is a direct consequence of the hypothesis $\musti{\serverA_1}{\client_2}$.

  If $\mu \neq \aa$ then we are allowed to apply the
  \outputconfluence axiom to the transitions
  $\client_1 \st{\co{\aa}} \client_2$ and $\client_1 \st{\co{\mu}} \client'$
  to obtain a $\hat{\client}$ such that
  $\client_2 \st{\co{\mu}} \hat{\client}$
  and $\client' \st{\co{\aa}} \hat{\client}$.
  An application of the \outputcommutativity axiom to the transitions
  $\server_1 \st{\co{\aa}} \server_2 \st{\mu} \server'$
  provides us with a $\hat{\server}$ such that
  $\server_1 \st{\mu} \hat{\server} \st{\co{\aa}} \server'$.
  We now apply the inductive hypothesis with
  $\serverA' = \hat{\serverA}, \client' = \hat{\client}, \serverA_0 = \serverA'$
  and $\client_0 = \client'$ and obtain $\musti{\serverA_2}{\client'}$
  as required. This concludes the proof of
  (\ref{must-output-swap-l-fw-g-2}),
  and therefore of the lemma.
\end{proof}

\begin{lemma}
  \label{lem:musti-presereved-by-actions-of-unsuccesful-tests}
  Let $\genlts_A \in \obaFW$ and $\genlts_B \in  \obaFB$.
  For every $\server, \server' \in \StatesA$, $\client, \client' \in \StatesB$
  and every action $\mu \in \Act$ such that
  $\server \st{ \mu } \server'$ and $\client \st{ \co{\mu}} \client'$
  we have that $\musti{\server}{\client}$ and $\lnot \good{\client}$ implies $\musti{\server'}{\client'}$.%
\end{lemma}
\noindent
\begin{proof}
  By case analysis on the hypothesis that $\musti{\server}{\client}$.
\end{proof}

\section{Forwarders}
\label{sec:appendix-forxarders}

The intuition behind forwarders, quoting \cite{DBLP:conf/ecoop/HondaT91},
is that ``any message can come into the configuration, regardless of the forms of
inner receptors. [\ldots] As the experimenter is not synchronously
interacting with the configuration [\ldots], he may send any message
as he likes.''

In this appendix we give the technical results to ensure that
the function $\liftFWSym$ builds an LTS that satisfies
the axioms of the class \texttt{LtsEq}.

For any LTS $\genlts =\lts{ A }{L}{\st{}}$ of output-buffered agents we assume a function $\outputmultisetSym : A \rightarrow MO$
\begin{enumerate}[(i)]
\item
  $\co{\aa} \in O(\server)$ if and only if $\co{\aa} \in \outputmultiset{\server}$, and
\item
  for every $\server'$, if $\server \st{\co{\aa}} \server'$ then $\outputmultiset{\server} = \mset{\co{\aa}} \uplus \outputmultiset{\server'}$.
\end{enumerate}
Note that by definition $\outputmultiset{ \server }$ is a finite multiset.

\begin{definition}\label{def:strip-def}
  For any LTS $\genlts = \lts{\States}{L}{\st{}} \in \obaFB$,
  we define the function $\stripSym : \States \longrightarrow \States$ by induction on $
  \outputmultiset{ \server }$ as follows: if
  $\outputmultiset{\server}  = \varnothing $ then
  $\strip{\server} = \server$, while
  if $\exists \co{\aa} \in \outputmultiset{\server}$ and $\server \st{ \co{\aa} } \server'$ then
  $ \strip{ \server } =  \strip{ \server' } $.
  Note that $\strip{ \server }$ is well-defined thanks to the \outputdeterminacy
  and the \outputcommutativity axioms.\hfill$\blacksquare$
\end{definition}



We now wish to show that $\liftFW{\genlts} \in \obaFW$ for any LTS
$\genlts$ of output-buffered agents with feedback.
Owing to the structure of our typeclasses, we have first to construct an
equivalence~$\doteq$ over $\liftFW{\genlts}$ that is compatible with the
transition relation, \ie satisfies the axiom in \rfig{Axiom-LtsEq}.
We do this in the obvious manner, \ie by combining the equivalence $\simeq$
over the states of~$\genlts$ with an equivalence over mailboxes.

\begin{definition}
  \label{def:fw-eq}
  For any LTS $\genlts \in \obaFB$, two states $\serverA \triangleright M$
  and $\serverB \triangleright N$ of $\liftFW{\genlts}$ are
  equivalent, denoted
  $\serverA \triangleright M \doteq \serverB \triangleright N$, if
  $ \strip{ \serverA } \simeq \strip{ \serverB }$ and
  $M \uplus \outputmultiset{\serverA} = N \uplus
  \outputmultiset{\serverB}$.\hfill$\blacksquare$
\end{definition}

\begin{lemma}
  \label{lem:harmony-sta}\coqLTS{harmony_a}
  For every $\genlts_\StatesA$ and every
  $\serverA \triangleright M, \serverB \triangleright N \in \StatesA \times MO$,
  and every $\alpha \in L$, if
  $
  \serverA \triangleright M \mathrel{({\doteq} \cdot {\sta{\alpha}})}
  \serverB \triangleright N
  $ then
  $
  \serverA  \triangleright M \mathrel{({\sta{\alpha}} \cdot {\doteq})} \serverB' \triangleright N'.
  $
\end{lemma}

\begin{lemma}
  \label{lem:fw-eq-id-mb}
  For every $\genlts_\StatesA \in \obaFB$ and every
  $\serverA, \serverB \in \StatesA$, $M \in MO$,
  if $\serverA \simeq \serverB$ then $\serverA \triangleright M \doteq \serverB \triangleright M$.
\end{lemma}

\begin{proof}
This follows from the fact that if $\serverA \simeq \serverB$ then
$\strip{\serverA} \simeq \strip{\serverB}$ and $\outputmultiset{\serverA} = \outputmultiset{\serverB}$.
\end{proof}


\noindent
\textbf{\rlem{liftFW-works}}. \textit{For every LTS~$\genlts \in \obaFB$, $\liftFW{\genlts} \in \obaFW$}.

\begin{proof}
  We must show that, given an LTS~$\genlts = \lts{\States}{L}{\st{}} \in \obaFB$, we have that
  $\liftFW{\genlts} \in \obaFW$.
  To do so, we need to show that $\liftFW{\genlts}$ obeys to the axioms given in \req{axioms-forwarders},
  namely \boom and \fwdfeedback.
  We first show that $\liftFW{\genlts}$ obeys to the \boom axiom.

  We pick a process $\serverA \in \States$, a mailbox $M \in \MO$ and a name $\aa \in \Names$.
  The axiom \boom requires us to exhibit a process $\serverA' \in A$ and a mailbox $M' \in \MO$ such that
  the following transitions hold.
  $$
  \serverA \triangleright M \sta{\aa} \serverA' \triangleright M' \sta{\co{\aa}} \serverA \triangleright M
  $$

  We choose $\serverA' = \serverA$ and $M' = (\mset{\co{a}} \uplus M)$.
  An application of the rule \stminplift and then the rule \stmoutlift from \rfig{rules-liftFW}
  allows us to derive the required sequence of transitions as shown below.
  $$
  \serverA \triangleright M \sta{\aa} \serverA \triangleright (\mset{\co{a}} \uplus M) \sta{\co{\aa}} \serverA \triangleright M
  $$

  We now show that $\liftFW{\genlts}$ obeys to the \fwdfeedback axiom.
  To begin we pick three processes $\serverA_1, \serverA_2, \serverA_3 \in \States$, three
  mailboxes $M_1, M_2, M_3 \in \MO$ and a name $\aa \in \Names$ such that:

  $$
  \serverA_1 \triangleright M_1 \sta{\co{\aa}} \serverA_2 \triangleright M_2  \sta{\aa} \serverA_3 \triangleright M_3
  $$

  We need to show that either:
  \begin{enumerate}
  \item $\serverA_1 \triangleright M_1 \sta{\tau} \serverA_3 \triangleright M_3$, or
  \item $\serverA_1 \triangleright M_1 \doteq \serverA_3 \triangleright M_3$
  \end{enumerate}

  We proceed by case analysis on the last rule used to derive the transition
  $\serverA_1 \triangleright M_1 \sta{\co{\aa}} \serverA_2 \triangleright M_2$.
  This transition can either be derived by the rule \stmoutlift or the rule \stproclift .

  We first consider the case where the transition has been derived using the rule
  \stmoutlift. We then have that $\serverA_1 = \serverA_2$ and $M_1 =  (\mset{\co{a}} \uplus M_2)$.
  We continue by case analysis on the last rule used to derive the transition
  $\serverA_2 \triangleright M_2  \sta{\aa} \serverA_3 \triangleright M_3$.
  If this transition has been derived using the rule \stminplift then it must be the case that
  $\serverA_2 = \serverA_3$ and that $(\mset{\co{a}} \uplus M_2) = M_3$.
  This lets us conclude by the following equality to show that
  $\serverA_1 \triangleright M_1 \doteq \serverA_3 \triangleright M_3$.
  $$
  \serverA_1 \triangleright M_1
  = \serverA_2 \triangleright (\mset{\co{a}} \uplus M_2)
  = \serverA_3 \triangleright M_3
  $$
  Otherwise, this transition has been derived using the rule \stproclift, which implies that
  $\serverA_2 \st{\aa} \serverA_3$ together with $M_2 = M_3$.
  An application of the rule \stcommlift ensures the following transition and allows us to
  conclude this case with $\serverA_1 \triangleright M_1 \sta{\tau} \serverA_3 \triangleright M_3$.
  $$
  \serverA_1 \triangleright M_1 = \serverA_2 \triangleright (\mset{\co{a}} \uplus M_2) \st{\tau}
  \serverA_3 \triangleright M_2 = \serverA_3 \triangleright M_3
  $$

  We now consider the case where the transition
  $\serverA_1 \triangleright M_1 \sta{\co{\aa}} \serverA_2 \triangleright M_2$ has been
  derived using the rule \stproclift such that
  $\serverA_1 \st{\co{\aa}} \serverA_2$ and $M_1 = M_2$.

  Again, we continue by case analysis on the last rule used to derive the transition
  $\serverA_2 \triangleright M_2  \sta{\aa} \serverA_3 \triangleright M_3$.
  If this transition has been derived using the rule \stminplift then it must be the case that
  $\serverA_2 = \serverA_3$ and $(\mset{\co{a}} \uplus M_2) = M_3$.
  Also, note that, as $\genlts \in \obaFB$, the transition $\serverA_1 \st{\co{\aa}} \serverA_2$
  implies
  $\outputmultiset{\serverA_1} = \outputmultiset{\serverA_2} \uplus \mset{\co{\aa}}$.
  In order to prove $\serverA_1 \triangleright M_1 \doteq \serverA_3 \triangleright M_3$,
  it suffices to show the following:
  \begin{enumerate}[(a)]
  \item $\strip{ \serverA_1 } \simeq \strip{ \serverA_3 }$, and
  \item $M_1 \uplus \outputmultiset{\serverA_1} = M_3 \uplus \outputmultiset{\serverA_3}$
  \end{enumerate}

  We show that $\strip{ \serverA_1 } \simeq \strip{ \serverA_3 }$ by definition of
  $\strip{}$ together with the transition
  $\serverA_1 \st{\co{\aa}} \serverA_2$ and the equality $\serverA_2 = \serverA_3$.

  The following ensures that
  $M_1 \uplus \outputmultiset{\serverA_1} = M_3 \uplus \outputmultiset{\serverA_3}$.

  $$
  \begin{array}{llll}
  & M_1 \uplus \outputmultiset{\serverA_1}\\
  = & M_1 \uplus \outputmultiset{\serverA_2} \uplus \mset{\co{\aa}} && \text{from } \outputmultiset{\serverA_1} = \outputmultiset{\serverA_2} \uplus \mset{\co{\aa}}\\
  = & M_2 \uplus \outputmultiset{\serverA_3} \uplus \mset{\co{\aa}} && \text{from } M_1 = M_2, \serverA_2 = \serverA_3\\
  = & (M_2 \uplus \mset{\co{\aa}}) \uplus \outputmultiset{\serverA_3}\\
  = & M_3 \uplus \outputmultiset{\serverA_3} && \text{from } M_3 = M_2 \uplus \mset{\co{\aa}}\\
  \end{array}
  $$

  If the transition $\serverA_2 \triangleright M_2  \sta{\aa} \serverA_3 \triangleright M_3$
  has been derived using the rule \stproclift then it must be the case that
  $\serverA_2 \st{\aa} \serverA_3$ and $M_2 = M_3$.
  As $\genlts \in \obaFB$, we are able to call the axiom
  \outputfeedback together with the transitions
  $\serverA_1 \st{\co{\aa}} \serverA_2$ and $\serverA_2 \st{\aa} \serverA_3$
  to obtain a process $\serverA'_3$
  such that $\serverA_1 \st{\tau} \serverA'_3$ and $\serverA'_3 \simeq \serverA_3$.
  An application of \rlem{fw-eq-id-mb} and rule \stproclift allows us to conclude that
  $\serverA_1 \triangleright M_1 \mathrel{({\sta{\tau}} \cdot {\doteq})} \serverA_3 \triangleright M_3$. \qed
\end{proof}

\section{Completeness}
\label{sec:proof-completeness}
\label{sec:bhv-completeness}

\begin{table*}
  \hrulefill\\

    \begin{minipage}{300pt}%
      $\forall \trace \in \Actfin, \forall \aa \in \Names,$
      \begin{enumerate}[(1)]
      \item
        \label{gen-spec-ungood}
        $ \lnot \good{f(\trace)}$
      \item
        \label{gen-spec-mu-lts-co}
        $ \forall \mu \in \Act, f (\mu.\trace) \st{\co{\mu}} f(\trace)$
      \item
        \label{gen-spec-mu-out-ex-tau}
        $ f (\co{a}.\trace) \st{\tau} $
      \item
        \label{gen-spec-out-good}
        $ \forall \client \in \States, f (\co{a}.\trace)
        \st{\tau} \client$ implies $\good{ \client }$
      \item
        \label{gen-spec-out-mu-inp}
        $ \forall \client \in \States, \mu \in \Act,$
        $f (\co{a}.\trace) \st{\mu} \client$ implies $\mu = \aa$ and
        $\client = f(s)$
      \end{enumerate}
    \end{minipage}
    \\[2em]
    \begin{minipage}{300pt}
      $\forall E \subseteq \Names$,
      \begin{enumerate}[(t1)]
      \item\label{gen-spec-acc-nil-stable-tau}
        $\testacc{\varepsilon}{E} \Nst{\tau}$
      \item\label{gen-spec-acc-nil-stable-out}
        $\forall \aa \in \Names, \testacc{\varepsilon}{E} \Nst{\co{\aa}}$
      \item\label{gen-spec-acc-nil-mem-lts-inp}
        $\forall \aa \in \Names, \testacc{\varepsilon}{E} \st{\aa}$ if and only if
        $\aa \in E$
      \item\label{gen-spec-acc-nil-lts-inp-good}
        $\forall \mu \in \Act, \client \in \States,
        \testacc{\varepsilon}{E} \st{\mu} \client$ implies $\good{\client}$
      \end{enumerate}
    \end{minipage}
  \\[2em]
      \begin{minipage}{300pt}
      \begin{enumerate}[(c1)]
      \item
        $\forall \mu \in \Act, \testconv{\varepsilon} \Nst{\mu}$ \qquad
      \item
        $\exists \client, \testconv{\varepsilon} \st{\tau} \client$ \qquad
      \item
        $\forall \client, \testconv{\varepsilon} \st{\tau} \client$ implies
        $\good{ \client }$
        \\
      \end{enumerate}
      \end{minipage}
  \caption{Properties of the functions that generate clients.}
\hrulefill
\label{tab:properties-functions-to-generate-clients}
\end{table*}

This section is devoted to the proof that the alternative preorder
given in \rdef{accset-leq} includes the \mustpreorder.
First we present a general outline of the main technical results to obtain
the proof we are after. Afterwards, in Subsection~(\ref{sec:appendix-completeness})
we discuss in detail on all the technicalities.

Proofs of completeness of characterisations of contextual preorders usually
require using, as the name suggests, syntactic contexts.
Our calculus-independent setting, though, does not allow us to
define them.  Instead we phrase our arguments using two functions
$
\testconvSym :  \Actfin \rightarrow \States,$ and
$\testaccSym :  \Actfin \times \parts{\Names} \rightarrow \States$
where \lts{ \States }{L}{ \st{} } is some LTS of \obaFB.
In \rtab{properties-functions-to-generate-clients} we gather all the
{\em properties} of~$\testconvSym$ and~$\testaccSym$ that are sufficient to give our
arguments. The properties (1) - (5) must hold for both~$\testconvSym$ and
$\testacc{\varepsilon}{-}$ for every set of names~$O$, the properties (c1) -
(c2) must hold for~$\testconvSym$, and (t1) - (t4) must hold for~$\testaccSym$.

We use the function~$\testconvSym$ to test the convergence of servers, and the
function~$\testaccSym$ to test the acceptance sets of servers.

A natural question is whether such~$\testconvSym$ and~$\testaccSym$ can actually exist.
The answer depends on the LTS at hand. In \rapp{client-generators},
and in particular \rfig{client-generators}, we define these functions for
the standard LTS of~\ACCS, and it should be obvious how to adapt those
definitions to the asynchronous $\pi$-calculus \cite{DBLP:journals/jlp/Hennessy05}.

In short, our proofs show that~$\asleq$ is complete with respect to~$\testleqS$
in any LTS of output-buffered agents with feedback wherein the
functions~$\testconvSym$ and~$\testaccSym$ enjoying the properties in
\rtab{properties-functions-to-generate-clients} can be defined.

\renewcommand{\States}{\ensuremath{A}}

First, converging along a finite trace~$\trace$~is logically
equivalent to passing the client~$\testconv{ \trace}$.  In other
words, there exists a bijection between the proofs (i.e. finite
derivation trees of~$\musti{ \server }{\testconv{ \trace}}$) and
the ones of~$ \server \cnvalong{ \trace }$. We first give the
proposition, and then discuss the auxiliary lemmas to prove it.

\begin{proposition}\coqCom{must_iff_acnv}
  \label{prop:must-iff-acnv}
  For every $\genlts_{\States} \in \obaFW$,
  $\server \in \States$, and
  $\trace \in \Actfin$ we have that $\musti{\server}{ \testconv{ \trace} }$
  if and only if~$\server \cnvalong \trace$.
\end{proposition}
\noindent
The {\em if} implication is \rlem{acnv-must} and the {\em only if}
implication is \rlem{must-cnv}.
\noindent
The hypothesis that $\genlts_{\States} \in \obaFW$,
\ie the use of forwarders, is necessary to show that convergence
implies passing a client, as shown by the next example.
\begin{example}
    \label{ex:forwarders-necessary}
    Consider a server~$\server$ in an LTS $\genlts \in \obaFB$
    whose behaviour amounts to the following transitions:
    $\server \st{ \ab } \Omega \st{ \tau }  \Omega \st{ \tau } \ldots$
    Note that this entails that $\genlts$ does not
    {\em not} enjoy the axioms of forwarders.

    Now let $ \trace = \aa.\ab $. Since $\server \conv$ and $\server
    \Nwt{\aa}$ we know that $ \server \cnvalong \aa.\ab$.
    On the other hand \rtab{properties-functions-to-generate-clients}(\ref{gen-spec-mu-lts-co})
    implies that the client $\testconv{\trace}$ performs the transitions
    $\testconv{ \trace } \st{ \co{\aa}} \testconv{ \ab }  \st{ \co{\ab}} \testconv{ \varepsilon }$.
    Thanks to the \outputcommutativity axiom we obtain $\testconv{ \trace } \st{ \co{\ab}} \st{ \co{\aa}} \testconv{ \varepsilon }$.
    \rtab{properties-functions-to-generate-clients}(\ref{gen-spec-ungood}) implies that the states
    reached by the client are unsuccessful, and so by zipping the traces performed
    by $\server$ and by $\testconv{\trace}$
    we build a maximal computation of
    $\csys{\server}{\testconv{\trace}}$ that is unsuccessful,
    and thus $\Nmusti{\server}{\testconv{\trace}}$.\hfill$\qed$
  \end{example}
  \noindent
  This example explains why in spite of \rlem{musti-obafb-iff-musti-obafw}
  output-buffered agents with feedback do not suffice to use the
  standard characterisations of the \mustpreorder.


We move on to the more involved technical results, \ie the next
two lemmas, that we use to reason on acceptance sets of servers.


\begin{lemma}
  \label{lem:completeness-part-2.2-diff-outputs}
  Let $\genlts_A \in \obaFW$.
  For every $\server \in \States$, $\trace \in \Actfin$,
  and every $L, E \subseteq \Names$, if
  $\co{L} \in \accht{ \server }{ \trace }$
  then $\Nmusti{ \server }{ \testacc{\trace}{E \setminus L}}$.
\end{lemma}

\begin{lemma}\coqCom{must_gen_a_with_s}
  \label{lem:completeness-part-2.2-auxiliary}
  Let $\genlts_A \in \obaFW$.
  $\Forevery \server \in \States, \trace \in \Actfin$,
  and every finite set $\ohmy \subseteq \co{\Names}$,
  if $\server \cnvalong s$ then either
  \begin{enumerate}[(i)]
      \item
    $\musti{\server}{\testacc{ \trace }{ \bigcup \co{ \accht{p}{s}
          \setminus \ohmy }}}$, or
  \item
    there exists $\widehat{\ohmy} \in \accht{ \server }{ \trace }$ such that $\widehat{\ohmy} \subseteq \ohmy$.
  \end{enumerate}
\end{lemma}

\renewcommand{\traceA}{s_1}
\renewcommand{\traceB}{s_2}


\renewcommand{\stateB}{q}
\renewcommand{\traceB}{\traceC}

\renewcommand{\traceA}{s_1}
\renewcommand{\traceB}{s_2}
\renewcommand{\traceC}{s_3}

\renewcommand{\traceA}{s_1}
\renewcommand{\traceB}{s_2}
\renewcommand{\traceC}{s_3}

We can now show that the alternative preorder~$\asleq$
includes~$\testleqS$ when used over LTSs of forwarders.
\begin{lemma}
  \label{lem:completeness}
  For every $\genlts_A, \genlts_B \in \obaFW$ and
  servers $\serverA \in \StatesA, \serverB \in \StatesB $,
  if ${ \serverA } \testleqS { \serverB }$
  then ${ \serverA } \asleq { \serverB }$.
\end{lemma}
\begin{proof}
  Let ${ \serverA } \testleqS { \serverB }$.
  To prove that ${ \serverA } \bhvleqone { \serverB }$,
  suppose
  ${\serverA} \cnvalong \trace$ for some trace $\trace$.
  \rprop{must-iff-acnv} implies $\musti{ {\serverA} }{\testconv{ \trace} }$,
  and so by hypothesis $ \musti{ { \serverB } }{\testconv{ \trace }}$.
  \rprop{must-iff-acnv} ensures that ${\serverB} \cnvalong \trace$.

  We now show that ${ \serverA } \bhvleqtwo { \serverB }$.
  Thanks to
  \rlem{conditions-on-accsets-logically-equivalent}, it is enough to prove
  that $\serverA \asleqAfw \serverB$. So, we show that
  for every trace $\trace \in \Actfin$, if ${ \serverA } \acnvalong \trace$
  then $\accht{{ \serverA }}{\trace} \ll
  \accht{{ \serverB }}{\trace} $.
  Fix an $O \in \accht{ {\serverB} }{ \trace }$.
  We have to exhibit a set $\widehat{O} \in \accht{{ \serverA
  }}{\trace}$ such that $\widehat{O} \subseteq O$.

  By definition $O \in \accht{ {\serverB} }{ \trace }$
  means that for some $q'$ we have ${\serverB} \wt{ \trace } q' \stable$
  and $O(\serverB') = O$.
  Let $E = \bigcup \accht{ { \serverA }}{\trace}$ and $X = E \setminus O $.
  The hypothesis that ${\serverA}~\cnvalong~\trace$,
  and the construction of the set~$X$
  let us apply \rlem{completeness-part-2.2-auxiliary}, which implies that
  either \begin{enumerate}[(a)] \item
  $\musti{{\serverA}}{\testacc{\trace}{ \co{ X }}}$, or
  \item there exists
  a $\widehat{O} \in \accht{{ \serverA }}{\trace}$ such that $\widehat{O}
  \subseteq O(\serverB')$.
  \end{enumerate}
  Since (b) is exactly what we are after, to conclude the
  argument it suffices to prove that (a) is false.
  This follows from \rlem{completeness-part-2.2-diff-outputs}, which proves
  $\Nmusti{ {\serverB}  }{ \testacc{ \trace }{ \co{ X }} }$,
  and the hypothesis ${ \serverA } \testleqS { \serverB }$,
  which ensures  $\Nmusti{ {\serverA}  }{ \testacc{ \trace }{ \co{ X }} }$.
\end{proof}

The fact that the \mustpreorder can be captured via the function $\liftFW{-}$
and $\asleq$ is a direct consequence of \rlem{musti-obafb-iff-musti-obafw}
and \rlem{completeness}.
\begin{proposition}[Completeness]
  \label{prop:bhv-completeness}
  For every $\genlts_A, \genlts_B \in \obaFB$ and
  servers $\serverA \in \StatesA, \serverB \in \StatesB $,
  if $\serverA \testleqS \serverB$ then $\liftFW{ \serverA } \asleq \liftFW{ \serverB }$.
\end{proposition}

\label{completeness-part-2.2-auxiliary-proof}
\label{sec:appendix-completeness}

  We now gather all the technical auxiliary lemmas and then discuss the
  proofs of the main ones.

By assumption, outputs preserve the predicate~$\goodSym$.
For
stable clients, they also preserve the negation of this predicate.
\begin{lemma}
    \label{lem:st-wtout-Nok}
  For all $\client, \client' \in \States$ and trace $\trace \in \co{\Names}^\star$,
  $\client \stable$, $\lnot \good{\client}$ and $\client \wt{ \trace }
  \client'$ implies $\lnot \good{\client'} $.
\end{lemma}





\subsection{Testing convergence}
We start with preliminary facts, in particular two lemmas that follow
from the properties in \rtab{properties-functions-to-generate-clients}.

A process $\state$ {\em converges along} a trace $\trace$ if for every
$\state'$ reached by $\state$ performing any prefix of $\trace$, the
process $\state'$ converges.

\begin{lemma}
  \label{lem:cnvalong-iff-prefix}\label{lem:acnvalong-iff-prefix}
  For every $\lts{\States}{L}{\st{}}$, $\state \in \States$, and $\trace \in \Actfin$,
  $ \state \cnvalong \trace$ if and only if
  $ \state \wt{\trace'} \stateA $ implies $\stateA \conv$
  for every $\trace'$ prefix of $\trace$.
\end{lemma}

Traces of output actions impact neither the stability of servers,
nor their input actions.
\begin{lemma}
  \label{lem:st-wtout-st}  \label{lem:st-wtout-inp}
  For every $\genlts_\StatesA$,
  every $\server, \server' \in \StatesA$ and every trace $\trace \in \co{\Names}^\star$,
  \begin{enumerate}
    \item 
      $\server \stable$ and $\server \wt{ \trace } \server'$ implies $\server' \stable$.
    \item 
      $\server \stable$ and $\server \wt{ \trace } \server'$ implies $I(\server) = I(\server')$.
  \end{enumerate}
\end{lemma}

\begin{lemma}
  \label{lem:testconv-always-reduces}
  For every $\trace \in\Actfin$, $\testconv{ \trace} \st{\tau}$.
\end{lemma}

The \outputdeterminacyinv axiom is used in the proof of the next lemma.
\begin{lemma}\coqCom{inversion_gen_test_mu}
  \label{lem:inversion-gen-mu}
  For every $\trace \in \Actfin$, 
  if $\testconv{\trace} \st{ \mu } \client $
  then either
  \begin{enumerate}[(a)]
  \item\label{pt:inversion-gen-mu-left}
    $\good{ \client }$, or
  \item\label{pt:inversion-gen-mu-right}
    $s = \traceA. \co{\mu}. \traceB$ for some $\traceA \in \Names^{\star}$ and $\traceB \in \Actfin$
    such that $\client \simeq \testconv{\traceA.\traceB}$.
  \end{enumerate}
\end{lemma}

\begin{lemma}\coqCom{inversion_gen_conv_tau}
  \label{lem:inversion-feeder-tau}
  For every $\trace \in \Actfin$,
  if $\testconv{\trace} \st{\tau} \client$
    then either:
  \begin{enumerate}[(a)]
  \item\label{inversion-feeder-tau-ok} $\good{ \client }$, or
  \item\label{inversion-feeder-tau-split}
    there exist $\ab$, $\traceA, \traceB$ and $\traceC$ with
    $\traceA.\ab.\traceB \in \Names^\star$ such that
    $s = \traceA.\ab.\traceB.\co{\ab}.\traceC$ and
    $\client \simeq \testconv{\traceA.\traceB.\traceC}$.
    \end{enumerate}
\end{lemma}

\renewcommand{\stateB}{q}
\renewcommand{\traceB}{\traceC}

\renewcommand{\traceA}{s_1}
\renewcommand{\traceB}{s_2}
\renewcommand{\traceC}{s_3}

\begin{lemma}
  \label{lem:must-gen-conv-wt-mu}
  Let $\genlts_A \in \obaFW$. For every server $\server, \server' \in \States$,
  trace $\trace \in \Actfin$ and
  action $\mu \in \Act$ such that $\server \wt{\mu} \server'$
  we have that
  $\musti{\server}{\testconv{\mu.\trace}} \implies \server' \musti{\server'}{\testconv{\trace}}$.
\end{lemma}
\begin{proof}
  By rule induction on the reduction $\server \wt{\mu} \server'$ together with
  \rlem{musti-preserved-by-left-tau} and \rlem{musti-presereved-by-actions-of-unsuccesful-tests}.
\end{proof}

\begin{lemma}
  \label{lem:must-cnv}
  Let $\genlts_A \in \obaFW$. For every server $\server \in \States$, trace $\trace \in \Actfin$
  we have that
  $\musti{\server}{\testconv{\trace}} \implies \server \cnvalong \trace$.
\end{lemma}
\begin{proof}
  We proceed by induction on the trace $\trace$.
  In the base case $\trace$ is $\varepsilon$.
  \rtab{properties-functions-to-generate-clients}(\ref{gen-spec-ungood}) states that
  $\lnot \good{\testconv{\varepsilon}}$ and we apply \rlem{must-terminate} to obtain $\server \convi$,
  and thus $\server \cnvalong \varepsilon$.
  In the inductive case $\trace$ is $\mu.\trace'$ for some $\mu \in \Act$ and $\trace' \in \Actfin$.
  We must show the following properties,
  \begin{enumerate}
  \item $\server \convi$, and
  \item for every $\server'$ such that $\server \wt{\mu} \server'$, $\server' \cnvalong \trace'$.
  \end{enumerate}
  We prove the first property as we did in the base case,
  and we apply \rlem{must-gen-conv-wt-mu} to prove the second property.
\end{proof}

\begin{lemma}
  \label{lem:acnvalong-preserved-by-operations}
  \label{lem:acnvalong-preserved-by-annihilation}
  \label{lem:acnv-drop-in-the-middle}
  Let $\genlts_A \in \obaFW$. For every $\server \in \States$, $\traceA \in \Names^\star$ and $\traceC \in \Actfin$ we have that
  \begin{enumerate}
  \item\label{pt:acnvalong-preserved-by-drop-in-the-middle}
        for every $\mu \in \Act$,
        if $\state \acnvalong \traceA.\mu.\traceC$ and $\state \st{\mu} \stateB$ then $\stateB \acnvalong \traceA.\traceC$,
  \item\label{pt:acnvalong-preserved-by-annihilation}
  for every $\aa.\traceB \in \Names^\star$ if
  $\state \acnvalong \traceA.\aa.\traceB.\co{\aa}.\traceC$ then
  $\state \acnvalong \traceA.\traceB.\traceC$.
  \end{enumerate}
\end{lemma}

\begin{lemma}\coqCom{terminate_must_gen_conv_nil}
  \label{lem:terminate-must-gen-conv-nil}
  For every LTS $\genlts_{\States}$ and
  every $\server \in \States$,
  $\server \convi$ implies $\musti{\server}{\testconv{\varepsilon}}$.
\end{lemma}
\begin{proof}
Rule induction on the derivation of $p \convi$.
\end{proof}

\begin{lemma}\coqCom{acnv_must}
  \label{lem:acnv-must}
  For every $\genlts_A \in \obaFW$, every $\server \in \StatesA$, and $\trace \in\Actfin$,
  if $\server \cnvalong \trace$ then $\musti{\server}{\testconv{ \trace}}$.
\end{lemma}

\begin{proof}
The hypothesis $\server \cnvalong \trace$ ensures $p \convi$.
We show that $\musti{\server}{\testconv{ \trace}}$ reasoning by complete
induction on the length of the trace~$\trace$.
The base case is \rlem{terminate-must-gen-conv-nil} 
and here we discuss the inductive case, i.e. when $\len{ \trace } = n + 1$ for some $n \in \N$.

We proceed by rule induction on $p \convi$.
In the base case $p \Nst{\tau}$,
and the reduction at hand 
is due to either a $\tau$ transition in~$\testconv{ \trace}$, or
a communication between~$p$ and~$\testconv{ \trace}$.

In the first case~$\testconv{ \trace} \st{\tau} r$, and so \rlem{inversion-feeder-tau} ensures
that one of the following conditions holds,
\begin{enumerate}
\item $\good{ \client }$, or
\item there exist $\aa \in \Names$, $\traceA, \traceB$ and $\traceC$ with
  $\trace = \traceA.\aa.\traceB.\co{\aa}.\traceC$ and
  $r \simeq \testconv{\traceA.\traceB.\traceC}$.
\end{enumerate}

If $\good{\client}$ then we conclude via rule \mnow; %
otherwise \rptlem{acnvalong-preserved-by-operations}{acnvalong-preserved-by-annihilation}
and the hypothesis that $\server \cnvalong \trace$ imply $\server \cnvalong \traceA.\traceB.\traceC$,
thus prove $\musti{\serverA}{\client}$ 
via the inductive hypothesis of the complete induction on $\trace$.

We now consider the case when the transition is due to a communication,
\ie $\serverA \st{\mu} \serverA'$ and $\testconv{ \trace} \st{\co{\mu}} \client$.
\rlem{inversion-gen-mu}
tells us that either $\good{\client}$
or there exist $\traceA$ and $\traceB$ such that
$\trace = \traceA.\mu.\traceB$ and $\client \simeq \testconv{\traceA.\traceB}$.
In the first case we conclude via rule \mnow.
In the second case we apply
\rptlem{acnvalong-preserved-by-operations}{acnvalong-preserved-by-drop-in-the-middle}
to prove $\server' \cnvalong \traceA.\traceB$, and thus~$\musti{\server'}{\client}$ follows
from the inductive hypothesis of the complete induction.
In the inductive case of the rule induction on~$\server \convi$, we know that
$\server \st{\tau} \server'$ for some process $\server'$.
We reason again by case analysis on how the reduction
we fixed 
has been derived, \ie either via a~$\tau$ transition in~$\testconv{ \trace}$,
or via a communication between~$p$ and~$\testconv{ \trace}$, or via
a~$\tau$ transition in~$p$.
In the first two cases we reason as we did
for the base case of the rule induction.
In the third case~$\server \cnvalong{s}$ and
$\server \st{\tau} \server'$ imply $\server' \cnvalong{s}$, we
thus obtain $\musti{\server'}{\testconv{ \trace}}$ thanks to
the inductive hypothesis of the rule induction
which we can apply because the
tree to derive~$\server' \convi$ is smaller than the tree
to derive that~$\server \convi$.
\end{proof}

\subsection{Testing acceptance sets}

\renewcommand{\state}{p}
\renewcommand{\stateA}{p'}

In this section we present the properties of the function $\testacc{-}{-}$
that are sufficient to obtain completeness.
To begin with, $\testacc{-}{-}$ function enjoys a form of monotonicity with respect to its second argument.

\begin{lemma}
  \label{lem:tacc-monotonicity}
  \label{lem:must-f-gen-a-subseteq}
  Let $\genlts_A \in \obaFB$.
  $\Forevery \serverA \in \StatesA$, trace $\trace \in \Actfin$, and sets of outputs $O_1, O_2$,
  if $\musti{\serverA}{\testacc{\trace}{O_1}}$
  and $O_1 \subseteq O_2$ then
  $\musti{\serverA}{\testacc{\trace}{O_2}}$.
\end{lemma}
\begin{proof}
  Induction on the derivation of $\musti{ \server }{\testacc{ \trace }{ O_1 }}$.
\end{proof}

Let $\oba$ denote the set of LTS of output-buffered agents. Note that any $\genlts \in \oba$ need not enjoy the \outputfeedback axiom.

\begin{lemma}{\label{lem:aft-output-must-gen-acc}}
  Let $\genlts_A \in \oba$, and
  $\genlts_B \in \oba$.
  For every $\serverA \in \StatesA$, trace $\trace \in \Actfin$,
  set of outputs $O$ and name $\aa \in \Names$,
  such that
  \begin{enumerate}[(i)]
  \item
    \label{aft-output-must-gen-acc-h-1}
    $\serverA \convi$ and,
  \item
    \label{aft-output-must-gen-acc-h-2}
    $\Forevery \server' \in \StatesA$, $\server \wt{\co{\aa}} \serverA'$ implies
    $\musti{\server'}{\testacc{\trace}{ \co{O} }}$,
  \end{enumerate}
  we have that $\musti{\server}{\testacc{\co{\aa}.\trace}{ \co{O} }}$.
\end{lemma}
\begin{proof}
  We proceed by induction on the hypothesis $\server \convi$.

  \paragraph{(Base case: $\serverA$ is stable)}
  We prove $\musti{\server}{\testacc{\co{\aa}.\trace}{ \co{O} }}$ by applying rule \mstep.
  Since \rtab{properties-functions-to-generate-clients}(\ref{gen-spec-mu-out-ex-tau}) implies that
  $\csys{ \serverA }{\testacc{\co{\aa}.\trace}{ \co{ O } }} \st{\tau}$, all we need to prove is
  the following fact,
  \begin{equation}
    \tag{$\star$}
    \label{eq:aft-output-must-gen-acc-AIM}
    \forall \server' \in \StatesA, \client \in \StatesB \wehavethat
    \text{ if }
    \csys{ \server }{\testacc{\co{\aa}.\trace}{ \co{ O } }}
    \st{\tau} \csys{\server'}{\client} \text{ then } \musti{\server'}{\client}.
  \end{equation}
  %
  Fix a transition $\csys{ \serverA }{\testacc{\varepsilon}{\co{O}}} \st{\tau} \csys{\serverA'}{\client}$.
  As $\serverA$ is stable, this transition can either be due to:
  \begin{enumerate}
  \item a $\tau$-transition performed by the client such that
    $\testacc{\co{\aa}.\trace}{ \co{O} } \st{\tau} \client$, or
  \item an interaction between the server $\serverA$ and the client
    $\testacc{\co{\aa}.\trace}{ \co{O} }$.
  \end{enumerate}
  In the first case \rtab{properties-functions-to-generate-clients}(\ref{gen-spec-out-good}) implies $\good{\client}$, and hence we obtain $ \musti{\server'}{\client}$ via rule \mnow.
  In the second case there exists an action $\mu$ such that
  $$
  \serverA \st{\mu} \serverA'
  \quad\text{and}\quad
  \testacc{\co{\aa}.\trace}{\co{O}} \st{\co{\mu}} \client
  $$
  \rtab{properties-functions-to-generate-clients}(\ref{gen-spec-out-mu-inp}) implies $\mu$ is $\co{\aa}$
  and $\client = \testacc{\trace}{ \co{O} }$.
  We then have $\serverA \st{\co{\aa}} \serverA'$ and thus the reduction
  $\serverA \wt{\co{\aa}} \serverA'$,
  which allows us to apply the hypothesis (\ref{aft-output-must-gen-acc-h-2}) and obtain
  $\musti{\serverA'}{\client}$ as required.

  \paragraph{(Inductive case: $\server \st{\tau} \server'$ implies $\server'$)}
  The argument is similar to one for the base case, except that
  we must also tackle the case when the transition
  $\csys{ \server }{\testacc{ \co{\aa}.\trace }{ \co{O} }}
  \st{\tau} \csys{\server'}{\client}$ is due to a $\tau$ action performed
  by $\server$, \ie $\server \st{\tau} \server'$ and $\client = \testacc{ \co{\aa}.\trace }{ \co{O} }$.
  The inductive hypothesis tells us the following fact:
  \begin{center}
    For every $\serverA_1$ and $\aa$,
    such that
    $\serverA \st{\tau} \serverA_1$, for every $\serverA_2$,
   if $\serverA_1 \wt{\co{\aa}} \serverA_2$
    then $\musti{\serverA_2}{\testacc{\trace}{O}}$.
  \end{center}
  To apply the inductive hypothesis we have to show that
  for every $\serverA_2$ such that $\serverA' \wt{\co{\aa}} \serverA_2$
  we have that $\musti{\serverA_2}{\testacc{\trace}{ \co{ O }}}$.
  This is a consequence of the hypothesis (\ref{aft-output-must-gen-acc-h-2})
  together with the reduction $\serverA \st{\tau} \serverA' \wt{\co{\aa}} \serverA_2$,
  and thus concludes the proof.
\end{proof}

\begin{lemma}{\label{lem:must-gen-acc-stable}}
  Let $\genlts_{\StatesA} \in \obaFB$.
  $\Forevery \serverA \in \StatesA$ and set of outputs $O$,
  if $\serverA$ is stable then either
  \begin{enumerate}[(a)]
  \item $\musti{\serverA}{\testacc{\varepsilon}{\co{ \outof{\serverA} \setminus O}} }$, or
  \item $\outof{\serverA} \subseteq O$.
  \end{enumerate}
\end{lemma}
\begin{proof}
We distinguish whether $\outof{\serverA} \setminus O$ is empty or not.
In the first case, $\outof{\serverA} \setminus O = \emptyset$ implies
$\outof{\serverA} \subseteq O$, and we are done.

In the second case, there exists $\co{\aa} \in \outof{\serverA}$ such that
$\co{\aa} \notin O$.
Note also that
\rtab{properties-functions-to-generate-clients}(\ref{gen-spec-ungood})
ensures that  $\lnot \good{\testacc{\varepsilon}{\co{ \outof{\serverA} \setminus O}}}$,
and thus we construct a derivation of
$\musti{ \serverA }{\testacc{\varepsilon}{\co{ \outof{\serverA} \setminus O}}}$
by applying the rule \mstep. This requires us to show the following facts,
\begin{enumerate}
\item\label{must-gen-acc-stable-2-2} $\csys{\serverA}{\testacc{\trace}{\co{ \outof{\serverA} \setminus O}}} \st{}$, and
\item\label{must-gen-acc-stable-2-3} for each $\serverA'$, $r$ such that
  $\csys{\serverA}{\testacc{\trace}{\co{ \outof{\serverA} \setminus O}} } \st{\tau} \csys{\serverA'}{r}$,
  $\musti{\serverA'}{r}$ holds.
\end{enumerate}

To prove (\ref{must-gen-acc-stable-2-2}), we show that an interaction between
the server $\serverA$ and the test $\testacc{\trace}{\co{ \outof{\serverA} \setminus O }}$ exists.
As $\co{\aa} \in \outof{\serverA}$, we have that $\serverA \st{\co{\aa}}$.
Then $\co{\aa} \in \outof{\serverA} \setminus O$
together with (\ref{gen-spec-acc-nil-mem-lts-inp}) ensure
that $\testacc{\trace}{\co{ \outof{\serverA} \setminus O }} \st{\aa}$.
An application of the rule \rname{s-com} gives us the required transition
$\csys{\serverA}{\testacc{\trace}{\co{ \outof{\serverA} \setminus O}}} \st{}$.

To show (\ref{must-gen-acc-stable-2-3}), fix a silent transition
$\csys{\server}{\testacc{\trace}{\co{ \outof{\server} \setminus O}}} \st{\tau} \csys{\server'}{r}$.
We proceed by case analysis on the rule used to derive the
transition under scrutiny.
Recall that the server $\server$ is stable by hypothesis, and that
$\testacc{\trace}{\co{ \outof{\serverA} \setminus O }}$ is stable
thanks to
\rtab{properties-functions-to-generate-clients}(\ref{gen-spec-acc-nil-stable-tau}).
This means that the silent transition must have been derived via rule \scom.
Furthermore, \rtab{properties-functions-to-generate-clients}(\ref{gen-spec-acc-nil-stable-out}) implies that the test
$\testacc{\trace}{\co{ \outof{\serverA} \setminus O }}$ does not perform any output.
As a consequence, if there is an interaction it must be because the test
performs an input and becomes $r$.
\rtab{properties-functions-to-generate-clients}(\ref{gen-spec-acc-nil-lts-inp-good})
implies that $\good{r}$, and hence we obtain the required
$\musti{\serverA'}{r}$ applying rule \mnow.
\end{proof}

\noindent%
\textbf{\rlem{completeness-part-2.2-diff-outputs}}
Let $\genlts_A \in \obaFW$.
  For every $\server \in \States$, $\trace \in \Actfin$,
  and every $L, E \subseteq \Names$, if
  $\co{L} \in \accht{ \server }{ \trace }$
  then $\Nmusti{ \server }{ \testacc{\trace}{E \setminus L}}$.
\begin{proof}
  By hypothesis there exists a set $\co{ L } \in \accht{ \server }{ \trace }$,
  \ie for some $\server'$ we have $\server \wt{ \trace } \server' \stable$ and $O(p') = \co{L}$.
  We have to show that $\Nmusti{ \server }{\testacc{ \trace }{E \setminus L}}$, \ie
  $\musti{ \server }{\testacc{ \trace }{E \setminus L}}$ implies $\bot$.
  For convenience, let $X = E \setminus L$.


  We proceed by induction on the derivation of the weak transitions $\server \wt{ \trace } \server'$.
  In the base case the derivation consists in an application of rule \wtrefl,
  which implies that $\server = \server'$ and $\trace = \varepsilon$.
  We show that there exists no derivation of judgement $\musti{\server}{\testacc{ \trace }{X}}$.
  By definition, $\lnot \good{\testacc{ \trace }{ X }}$ and
  thus no tree that ends with \mnow can have $\musti{\server}{\testacc{ \trace }{ X }}$
  as conclusion.
  The hypotheses ensure that $\server$ is stable,
  and $\testacc{ \varepsilon }{ X }$ is stable by definition.
  The set of inputs of $\testacc{\varepsilon}{ X }$ is $X$,
  which prevents an interaction between $\server$ and
  $\testacc{ \trace }{ X }$, i.e. an application of rule \scom.
  This proves that $\csys{ \server }{\testacc{ \trace }{ X }}$ is stable,
  thus a side condition of \mstep is false, and the rule cannot
  be employed to prove $\musti{ \server }{\testacc{ \trace }{ X }}$.

  \renewcommand{\traceB}{t}

  In the inductive cases $p \wt{ \trace } p'$ is derived using either:
  \begin{enumerate}[(i)]
  \item \label{completeness-part-2.2-diff-outputs-ind-case-1}
    rule \wttau such that $p \st{\tau} \widehat{ \server } \wt{ \trace } \server'$, or
  \item \label{completeness-part-2.2-diff-outputs-ind-case-2}
    rule \wtmu such that $p \st{\mu} \widehat{ \server } \wt{\traceB} \server'$, with $\trace = \mu.\traceB$.
  \end{enumerate}

  In the first case, applying the inductive hypothesis requires us to show
  $\musti{\widehat{ \server } }{ \testacc{ \trace }{ X }}$, which is true since
  $\musti{p}{\testacc{ \trace }{ X }}$ is preserved by the $\tau$-transitions
  performed by the server.

  In the second case, applying the inductive hypothesis
  requires us to show $\musti{\widehat{ \server } }{ \testacc{ \traceB }{ X }}$.
  \rtab{properties-functions-to-generate-clients}(\ref{gen-spec-mu-lts-co}) implies that $\testacc{\mu.\traceB}{ X } \st{\co{\mu}} \testacc{\traceB}{ X }$.
  Then we derive via \scom~the transition
  $\csys{ \server }{ \testacc{\mu.\traceB}{ E }} \st{ \tau } \csys{
    \widehat{\server} }{ \testacc{\traceB}{ E } }$.
  Since $\musti{ \server }{\testacc{ \trace }{ X }}$ is preserved by the interactions
  occurring between the server and the client, which implies
  $\musti{\widehat{ \server }}{\testacc{ \traceB }{ X }}$ as required.
\end{proof}

\noindent%
\textbf{\rlem{completeness-part-2.2-auxiliary}} \coqCom{must_gen_a_with_s}
 Let $\genlts_A \in \obaFW$.
  $\Forevery \server \in \States, \trace \in \Actfin$,
  and every finite set $\ohmy \subseteq \co{\Names}$,
  if $\server \cnvalong s$ then either
  \begin{enumerate}[(i)]
      \item
    $\musti{\server}{\testacc{ \trace }{ \bigcup \co{ \accht{p}{s}
          \setminus \ohmy }}}$, or
  \item
    there exists $\widehat{\ohmy} \in \accht{ \server }{ \trace }$ such that $\widehat{\ohmy} \subseteq \ohmy$.
  \end{enumerate}
\begin{proof}
  We proceed by induction on the trace $\trace$.

  \paragraph{(Base case, $\trace = \varepsilon$)}
  The hypothesis $p \acnvalong{\varepsilon}$ implies $p \convi$ and we continue by induction
  on the derivation of $p \convi$.\footnote{Recall that the definition of~$\convi$ is in \req{def-intentional-predicates}}
  In the base case $p \convi$ was proven using rule \mnow, and hence $\server  \stable$.
  We apply \rlem{must-gen-acc-stable}
  to obtain either:
  \begin{enumerate}[(i)]
  \item\label{completeness-part-2.2-auxiliary-1-1-1}
    $\musti{\serverA}{\testacc{\varepsilon}{ \co{ \outof{\serverA} \setminus O}} }$, or
  \item\label{completeness-part-2.2-auxiliary-1-1-2}
    $\outof{\serverA} \subseteq O$.
  \end{enumerate}
  In case (\ref{completeness-part-2.2-auxiliary-1-1-1}) we are done.
  In case (\ref{completeness-part-2.2-auxiliary-1-1-2}), as $\serverA$ is stable we have
  $\setof{\serverA'}{\serverA \wt{\varepsilon} \serverA' \stable} = \set{\serverA}$ and thus
  $\accht{\serverA}{\varepsilon} = \set{O(\serverA)}$ and we conclude by letting
  $\widehat{O} = O(\serverA)$.

  In the inductive case~$\server \convi$ was proven using rule \mstep.
    We know that $ \server \st{ \tau }$,
    and the inductive hypothesis states that for any $\server'$ such that
    $\serverA \st{\tau} \server'$, either:
    \begin{enumerate}[(a)]
    \item\label{completeness-part-2.2-auxiliary-1-2-1-appendix}
      $\musti{\server'}{\testacc{\varepsilon}{\co{ \outof{\server'} \setminus O}}}$, or
    \item\label{completeness-part-2.2-auxiliary-1-2-2-appendix}
      there exists $\widehat{\ohmy} \in \accht{ \server' }{ \trace }$
      such that $\widehat{\ohmy} \subseteq \ohmy$.
    \end{enumerate}
    It follows that either
    \begin{description}
    \item[($\forall$)]
      \label{completeness-part-2.2-auxiliary-1-2-1}
      for each $\serverA' \in \setof{\serverA'}{\serverA \st{\tau} \serverA'}$,
      $\musti{\serverA'}{\testacc{\varepsilon}{\bigcup \co{\accht{ \server' }{ \trace } \setminus O}}}$, or
  \item[($\exists$)]
    \label{completeness-part-2.2-auxiliary-1-2-2}
    there exists a $\server' \in \setof{\server'}{\serverA \st{\tau} \server'}$
    and a $\widehat{\ohmy} \in \accht{ \server' }{ \varepsilon }$
    such that $\widehat{\ohmy} \subseteq \ohmy$,
    \end{description}
    \noindent
    We discuss the two cases. If ($\exists$) the argument is straightforward:
    we pick the existing $\server'$ such that $\server \st{\tau}
    \server'$. The definition of $\accht{-}{-}$ ensures that
    and show that $\accht{\server'}{\varepsilon} \subseteq
    \accht{\serverA}{\varepsilon}$, and thus we conclude by choosing $\widehat{O}$.

  Case ($\forall$) requires more work.
  We are going to show that $\musti{\serverA}{\testacc{\varepsilon}{\bigcup \co{\accht{ \server }{ \trace } \setminus O}}}$ holds.
  To do so we apply the rule \mstep and we need to show the following facts,
  \begin{enumerate}[(a)]
  \item $\csys{ p }{\testacc{\varepsilon}{\bigcup \co{\accht{ \server }{ \trace } \setminus O}}} \st{\tau}$, and
  \item for each $\csys{p'}{r'}$ such that
    $\csys{ p }{\testacc{\varepsilon}{\bigcup \co{\accht{ \server }{ \trace } \setminus O}}}
    \st{\tau} \csys{p'}{r'}$, we have $\musti{\server'}{r'}$.
  \end{enumerate}

  The first requirement follows from the fact that $\server$ is not
  stable. %
  To show the second requirement we proceed by case analysis on the transition
  $\csys{ p }{\testacc{\varepsilon}{\bigcup \co{\accht{ \server }{ \trace } \setminus O}}} \st{\tau} \csys{p'}{r'}$.
  As $\testacc{\varepsilon}{\co{ \outof{\serverA} \setminus O }}$ is stable by
  $(\ref{gen-spec-acc-nil-stable-tau})$, it can either be due to:
  \begin{enumerate}
  \item a $\tau$-transition performed by the server $\serverA$ such that
    $\serverA \st{\tau} \serverA'$, or
  \item an interaction between the server $\serverA$ and the client
    $\testacc{\varepsilon}{\bigcup \co{\accht{ \server }{ \trace } \setminus O}}$.
  \end{enumerate}

  In the first case we apply the first part of the inductive hypothesis 
  to prove that 
  $\musti{\serverA'}{\testacc{\varepsilon}{\bigcup \co{\accht{ \serverA' }{ \trace } \setminus O}}}$,
  and we conclude via \rlem{must-f-gen-a-subseteq} to get the required
  $$
  \musti{\serverA'}{\testacc{\varepsilon}{\bigcup \co{\accht{ \server }{ \trace } \setminus O}}}.
  $$
  In the second case, there exists a $\mu \in \Act$ such that
  $$
  \serverA \st{\mu} \serverA' \text{ and }
  \testacc{\varepsilon}{\bigcup \co{\accht{ \server }{ \trace } \setminus O}} \st{\co{\mu}} r
  $$

  Thanks to \rtab{properties-functions-to-generate-clients}(\ref{gen-spec-acc-nil-lts-inp-good}) we apply
  rule \mnow\ to prove that $\musti{\serverA'}{ \client }$ and we are done
  with the base case of the main induction on the trace~$\trace$.

  \paragraph{(Inductive case, $\trace =  \mu.s'$)}
  By induction on the set $\setof{\serverA'}{\serverA \wt{\mu} \serverA'}$
  and an application of the inductive hypothesis we know that either:
  \begin{enumerate}[(i)]
  \item\label{completeness-part-2.2-auxiliary-2-2-2}
    there exists $\serverA' \in \setof{\serverA'}{\serverA \wt{\mu} \serverA'}$
    and $\widehat{\ohmy} \in \accht{\serverA'}{\trace'}$
    such that $\widehat{\ohmy} \subseteq \ohmy$, or
  \item\label{completeness-part-2.2-auxiliary-2-2-1}
    for each $\serverA' \in \setof{\serverA'}{\serverA \wt{\mu} \serverA'}$ we have that
    $\musti{\serverA'}{\testacc{\trace'}{\bigcup \co{\accht{\serverA'}{\trace'}}}}$.
  \end{enumerate}

  In the first case, the inclusion $\accht{\serverA'}{\trace'} \subseteq \accht{\serverA}{\mu.\trace'}$
  and $\widehat{\ohmy} \in \accht{\serverA'}{\trace'}$ imply
  $\widehat{\ohmy} \in \accht{\serverA}{\trace}$ and we are done.

  In the second case, we show
  $\musti{\serverA}{\testacc{\trace}{\bigcup \co{\accht{\serverA}{\trace}}}}$
  by case analysis on the action~$\mu$, which can be either an input or an output.
  \begin{itemize}
  \item If $\mu$ is an input, $\mu = \aa$ for some $\aa \in \Names$.
    An application of the axiom of forwarders gives us a $\serverA'$ such that
    $\server \st{\aa} \server' \st{\co{\aa}} \server$.
    An application of \rtab{properties-functions-to-generate-clients}(\ref{gen-spec-mu-lts-co})
    gives us the following transition,
    $$
    \testacc{\aa.s'}{\bigcup \co{ \accht{\serverA}{\aa.\trace'} \setminus O}}
    \st{\co{\aa}}
    \testacc{s'}{\bigcup \co{ \accht{\serverA}{\aa.\trace'} \setminus O}}
    $$
    By an application of \rlem{must-output-swap-l-fw} it is enough to show
    $$
    \musti{\serverA'}{\testacc{s'}{\bigcup \co{\accht{\serverA}{\aa.\trace'} \setminus O}}}
    $$
    to obtain the required
    $\musti{\serverA}{\testacc{\aa.s'}{\bigcup
        \co{\accht{\serverA}{\aa.\trace'} \setminus O}}}$.


  \item If $\mu$ is an output, $\mu = \co{\aa}$ for some $\aa \in \Names$ and
    we must show that
    $$
    \musti{\serverA}{\testacc{\co{\aa}.\trace'}{\bigcup \co{\accht{\serverA}{\co{\aa}.\trace'} \setminus O}}}.
    $$
    We apply \rlem{aft-output-must-gen-acc} together with (\ref{completeness-part-2.2-auxiliary-2-2-1}) to
    obtain
    $
    \musti{\serverA}{\testacc{\co{\aa}.\trace'}{\bigcup \co{\accht{\serverA'}{\trace'} \setminus O}}}
    $.
    Again, \rlem{must-f-gen-a-subseteq}
    together with the inclusion $\accht{\serverA'}{\trace'} \subseteq \accht{\serverA}{\co{\aa}.\trace'}$
    ensures the required
    $\musti{\serverA}{\testacc{\trace}{\bigcup \co{\accht{\serverA}{\trace} \setminus O}}}$.
  \end{itemize}
\end{proof}

\section{Soundness}
\label{sec:proof-soundness}
\label{sec:bhv-soundness}\label{sec:soundness-bhv}

\newcommand{\cnvleqset}{\mathrel{\preccurlyeq^{\mathsf{set}}_{\mathsf{cnv}}}}
\newcommand{\accleqset}{\mathrel{\preccurlyeq^{\mathsf{set}}_{\mathsf{acc}}}}
\newcommand{\Naccleqset}{\mathrel{{\not\preccurlyeq}^{\mathsf{set}}_{\mathsf{acc}}}}
\newcommand{\asleqset}{\mathrel{\preccurlyeq^{\mathsf{set}}_{\mathsf{AS}}}}
\newcommand{\Nasleqset}{\mathrel{\not\preccurlyeq^{\mathsf{set}}_{\mathsf{AS}}}}

In this section we prove the converse of \rprop{bhv-completeness},
i.e. that~$\asleq$ is included in~$\testleqS$. 
We remark immediately that a naïve reasoning does not work.
Fix two servers~$\serverA$ and~$\serverB$ such that $\serverA \asleq \serverB$.

To lift the predicates~$\bhvleqone$ and~$\bhvleqtwo$ to sets of servers,
we let $\accht{ X }{ \trace } = \setof{ O }{ \exists \server \in X . O
  \in \accht{\server}{ \trace }}$, and for every finite $X \in
\pparts{ \States }$, we write $X \conv$ to mean $\forall \state \in X \suchthat\state \conv$, we write $  X \cnvalong \trace$ to mean $\forall \state \in X \suchthat \state \cnvalong \trace$, and let

\begin{itemize}
\item $ X \cnvleqset \serverB$ to mean $\forall \trace \in \Actfin,
  \text{ if } X \cnvalong \trace
  \text{ then } \serverB \cnvalong \trace$,

\item
  $X \accleqset \serverB$ to mean $\forall \trace \in \Actfin,
  X \cnvalong{\trace} \implies \accht{ X }{ \trace } \ll \accht{ \serverB }{ \trace }$,

\item
  $X \asleqset \serverB$ to mean $X \cnvleqset \serverB \wedge X \accleqset \serverB$.
\end{itemize}

These definitions imply immediately the following equivalences,
  $
\set{\serverA} \accleqset \serverB \Longleftrightarrow p \bhvleqone \serverB$,
$\set{\serverA} \cnvleqset \serverB \Longleftrightarrow p \bhvleqtwo \serverB $
and thereby the following lemma.
\begin{lemma}
  \label{lem:alt-set-singleton-iff}
  For every LTS $\genlts_A, \genlts_B, \serverA \in \StatesA$,
  $\serverB \in \StatesB$,
  $\serverA \asleq \serverB$ if and only if $\set{\serverA} \asleqset \serverB$.
\end{lemma}


\newcommand{\completewrt}[2]{\ensuremath{\mathsf{cwrt} \, (#1,#2)}}

The preorder~$\asleqset$ is preserved by $\tau$-transitions on
its right-hand side, and by visible transitions
on both sides.
We reason separately on the two auxiliary
preorders $\cnvleqset$ and $\accleqset$.
We need one
further notion.
\begin{lemma}
  \label{lem:bhvleqone-preserved}
  Let $\genlts_\StatesA, \genlts_\StatesB \in \obaFW$.
  For every set $X \in \pparts{ \StatesA }$, and
  $\serverB \in \StatesB$, such that
  $X \cnvleqset \serverB$,
  \begin{enumerate}
  \item\label{pt:bhvleqone-preserved-by-tau}
    $\serverB \st{ \tau } \serverB'$ implies $X \cnvleqset \serverB'$,
  \item\label{pt:bhvleqone-preserved-by-mu}
    $X \convi$, $X \wt{\mu} X'$ and $\serverB \st{\mu} \serverB'$ imply $X' \cnvleqset \serverB'$.
  \end{enumerate}
\end{lemma}

\begin{lemma}
  \label{lem:bhvleqtwo-preserved}
  Let $\genlts_\StatesA, \genlts_\StatesB \in \obaFW$.
  For every
  $X, X' \in \pparts{ \StatesA }$ and $ \serverB \in \StatesB$,
  such that $X \accleqset \serverB$, then
  \begin{enumerate}
  \item\label{pt:bhvleqtwo-preserved-by-tau}
    $\serverB \st{ \tau } \serverB'$ implies $X \accleqset \serverB'$,
  \item\label{pt:bhvleqtwo-preserved-by-mu}
    if $X \convi$ then for every $\mu \in \Act$, every $\serverB'$ and $X'$ such that $\serverB \st{\mu} \serverB'$ and $X \wt{\mu} X'$
    we have $X' \accleqset \serverB'$.
  \end{enumerate}
\end{lemma}

The main technical work for the proof of soundness is carried out by the next lemma.
\begin{lemma}
  \label{lem:soundness-set}
  Let $\genlts_\StatesA, \genlts_\StatesB \in \obaFW$ and
  $\genlts_\StatesC \in \obaFB$.
  For every set of servers $X \in \pparts{ \StatesA }$,
  server $\serverB \in \StatesB$ and client $\client \in \StatesC$,
  if $\mustset{X}{\client}$ and $X \asleqset \serverB$ then $\musti{\serverB}{\client}$.
\end{lemma}

\begin{proposition}[Soundness]
  \label{prop:bhv-soundness}
  For every $\genlts_A, \genlts_B \in \obaFB$ and
  servers $\serverA \in \StatesA, \serverB \in \StatesB $,
  if $\liftFW{ \serverA } \asleq \liftFW{ \serverB }$ then $\serverA \testleqS \serverB$.
\end{proposition}
\begin{proof}
  \rlem{musti-obafb-iff-musti-obafw}
  ensures that the result follows if we prove that
$\liftFW{\serverA} \testleqS \liftFW{\serverB}.$
Fix a client $\client$ such that $\musti{\liftFW{\serverA}}{\client}$.
\rlem{soundness-set} implies the required
$\musti{\liftFW{\serverB}}{\client}$, if we show that
\begin{enumerate}[(i)]
  \item $\mustset{ \set{ \liftFW{\serverA} } }{ \client }$, and that
  \item $\set{ \liftFW{ \serverA } } \mathrel{\preccurlyeq^{\mathsf{set}}_{\mathsf{AS}}} \liftFW{\serverB}$.
\end{enumerate}
The first fact follows from the assumption that $\musti{\liftFW{\serverA}}{\client}$
and \rlem{musti-if-mustset-helper} applied to the singleton
$\set{\liftFW{\serverA}}$.
The second fact follows from the hypothesis that $\liftFW{ \serverA }
\asleq \liftFW{ \serverB }$ and \rlem{alt-set-singleton-iff}.
\end{proof}

\subsection{Technical results to prove soundness}
\label{sec:appendix-soundness}

We now discuss the proofs of the main technical results
behind \rprop{bhv-soundness}. The predicate $\opMustset$ is  monotonically decreasing with respect to its first
argument, and it enjoys properties analogous to the ones
of~$\opMusti$ that have been shown in \rlem{must-terminate} and
\rlem{musti-preserved-by-left-tau}.

\begin{lemma}
  \label{lem:mx-sub}
  For every LTS  $\genlts_\StatesA, \genlts_\StatesB$ and
  every set $X_1 \subseteq X_2 \subseteq \StatesA$, client $\client \in \StatesB$,
  if $\mustset{X_2}{\client}$ 
  then $\mustset{X_1}{\client}$.
\end{lemma}

\begin{lemma}
  \label{lem:mustx-terminate-ungood}
  Let $\genlts_\StatesA \in \obaFW$ and $\genlts_\StatesB \in \obaFB$.
  For every set $X \in \pparts{ \StatesA }$, client $\client \in \StatesB$,
  if $\lnot \good{\client}$ and $ \mustset{X}{\client}$ then $X \convi$.
\end{lemma}

\begin{lemma}
  \label{lem:wt-nil-mx}
  For every $\genlts_\StatesA, \genlts_\StatesB$,
  every set $X_1, X_2 \in \pparts{\StatesA}$, and client $\client \in \StatesB$,
  if $\mustset{X_1}{\client}$ and $X_1 \wt{\varepsilon} X_2$ then $\mustset{X_2}{\client}$.
\end{lemma}

\begin{lemma}
  \label{lem:mx-preservation}
  For every LTS $\genlts_A, \genlts_B$ and every
  $X \in \pparts{\StatesA}$ and $\client \in \StatesB$,
  if $\mustset{X}{\client}$ then for every $X'$ such that
  \begin{enumerate}[(a)]
  \item\label{pt:mx-preservation-wt-nil}
    If $X \wt{\varepsilon} X'$ then $\mustset{X'}{\client}$,
  \item\label{pt:mx-preservation-wt-mu}
    For any $\mu \in \Act$ and client $\client'$,
    if $X \wt{\mu} X'$, $\client \st{\co{\mu}} \client'$ and $\lnot \good{\client}$,
    then $\mustset{X'}{\client'}$.
  \end{enumerate}
\end{lemma}

\begin{lemma}
  \label{lem:mx-forall}
  Given two LTS $\genlts_A$ and $\genlts_B$ then for every
  $X \in \pparts{\StatesA}$ and $\client \in \StatesB$,
  if for each $\serverA \in X$ we have that $\musti{\serverA}{\client}$,
  then $\mustset{X}{\client}$.
\end{lemma}

\begin{lemma}
  \label{lem:stability-Nbhvleqtwo}
  Let $\genlts_\StatesA, \genlts_\StatesB \in \obaFW$ and $\genlts_\StatesC \in \obaFB$.
  For every $X \in \pparts{ \StatesA }$ and
  $\serverB \in \StatesB $ such that
  $X \asleqset \serverB$, for every $\client \in \StatesC$
  if $\lnot \good{\client}$ and $\mustset{X}{\client}$
  then $\csys{\serverB}{\client} \st{\tau}$.
\end{lemma}
\begin{proof}
  If either $\serverB \st{\tau}$ or $\client \st{\tau}$
  then we prove that $\csys{\serverB}{\client}$ performs
  a $\tau$-transition vis \stauserver or \stauclient, so suppose that both
  $\serverB$ and $\client$ are stable.
  Since $\serverB$ is stable we know that
  $$
  \accht{q}{ \varepsilon } = \set{ O(q) }
  $$
  The hypotheses $\lnot \good{\client}$ and $ \mustset{X}{\client}$ together with
  \rlem{mustx-terminate-ungood} imply $X \convi$ and thus $X \cnvalong \varepsilon$.
  The hypothesis $X \accleqset q$ with $\trace = \varepsilon$,
  gives us a $\serverA'$ such that $\serverA \wt{\varepsilon} \serverA' \stable$
  and $O(\serverA') \subseteq O(\serverB)$.
  By definition there exists the weak silent trace
  $ X \wt{ } X'$ for some set $X'$ such that $ \set{ \serverA' }
  \subseteq X' $.
  The hypothesis  $\mustset{X}{\client}$ together with
  \rlem{wt-nil-mx} and \rlem{mx-sub} ensure that
  $\mustset{\set{\serverA'}}{\client}$.

  As $\lnot \good{\client}$, $\mustset{\set{\serverA'}}{\client}$
  must have been derived using rule \mstep which implies that $\csys{\serverA'}{\client} \st{\tau}$.
  As both $\client$ is stable by assumption, and $\serverA'$ is stable
  by definition, this $\tau$-transition must have been
  derived using \rname{s-com}, and so %
    $\serverA' \st{\mu}$ and $\client \st{\co{\mu}}$ for some $\mu \in \Act$.
  Now we distinguish whether $\mu$ is an input or an output.
  In the first case $\mu$ is an input. Since $\genlts_{\StatesB} \in
  \obaFW$ we use the \boom axiom to prove $\serverB \st{ \mu }$, and thus
  $\csys{ \serverB }{ \client} \st{ \tau }$ via rule \scom.
  In the second case $\mu$ is an output, and so the inclusion
  $O(\serverA') \subseteq O(\serverB)$ implies that $\serverB
  \st{ \mu }$, and so we conclude again applying rule \scom.
\end{proof}

  \begin{lemma}
  \label{lem:empty-nleqx}
  \label{lem:X-must-perform-visible-action}
  Let $\genlts_\StatesA, \genlts_\StatesB \in \obaFW$.
  For every $X \in \pparts{ \StatesA }$ and
  $\serverB, \serverB' \in \StatesB$,
  such that $X \asleqset \serverB$, then
  for every $\mu \in \Act$, if
  $X \cnvalong \mu$ and $\serverB \st{\mu} \serverB'$ then $X \wt{\mu}$.
  \end{lemma}
  \begin{proof}
  Then, from $X \cnvleqset \serverB$ and $X \cnvalong \mu$ we have that
  $\serverB \cnvalong \mu$ and thus $\serverB' \convi$.
  As~$\serverB'$ converges, there must exist $\serverB''$ such that
  $$
  \serverB \wt{\mu} \serverB' \wt{\varepsilon} \serverB'' \stable
  $$
  and so $\accfwp{ \serverB }{ \mu }{ \st{}_\StatesB } \neq \emptyset$.
  An application of the hypothesis $X \accleqset \serverB$ implies that there exists
  a set $\widehat{ O } \in \accfwp{ X }{ \mu }{ \st{}_\StatesA }$,
  and thus there exist two servers $\serverA' \in X$ and $\serverA''$
  such that $\serverA' \wt{\mu} \serverA'' \stable$.
  Since $ \serverA' \in X $ it follows that $X \wt{\mu}$.
  \end{proof}

\noindent%
\textbf{\rlem{bhvleqone-preserved}}
  Let $\genlts_\StatesA, \genlts_\StatesB \in \obaFW$.
  For every set $X \in \pparts{ \StatesA }$, and
  $\serverB \in \StatesB$, such that
  $X \cnvleqset \serverB$ then
  \begin{enumerate}
  \item
    $\serverB \st{ \tau } \serverB'$ implies $X \cnvleqset \serverB'$,
  \item
    if $X \convi$ and $\serverB \st{\mu} \serverB'$
    then for every set $X \wt{\mu} X'$ 
    we have that $X' \cnvleqset \serverB'$.
  \end{enumerate}
\begin{proof}
  We first prove \rpt{bhvleqone-preserved-by-tau}.
  Let us fix a trace $\trace$ such that $X \cnvalong \trace$.
  We must show $\serverB' \cnvalong \trace$.
  An application of the hypothesis $X \cnvleqset \serverB$ ensures $\serverB \cnvalong \trace$.
  From the transition $\serverB \st{ \tau } \serverB'$ and the fact that
  convergence is preserved by the $\tau$-transitions we have that
  $\serverB' \cnvalong \trace$ as required.

  We now prove \rpt{bhvleqone-preserved-by-mu}.
  Fix a trace $\trace$ such that $X' \cnvalong \trace$.
  Since  $\serverB \st{\mu} \serverB'$, the required  $\serverB' \cnvalong \trace$ follows
  from $ \serverB \cnvalong \mu.\trace $.
  Thanks to the hypothesis $X \cnvleqset \serverB$ it suffices to show that
  $X \cnvalong \mu.\trace'$, \ie that
  $$
  \forall \server \in X . \server \cnvalong \mu.\trace'
  $$
  Fix a server $\server \in X$.
  We must show that
  \begin{enumerate}
  \item $\server \convi$ and that
  \item for any $\server'$ such that $\server \wt{\mu} \server'$ we have $\server' \cnvalong \trace$.
  \end{enumerate}
  The first requirement follows from the hypothesis $X \convi$.
  The second requirement follows from the transition $\server \wt{\mu} \server'$,
  from the assumption $X' \cnvalong \trace$,
  and the hypothesis that  $X \wt{\mu} X'$, which ensures that
  $\server' \in X'$ and thus by definition of $X' \cnvalong \trace$ that $\server' \cnvalong \trace$.
\end{proof}

\noindent%
\textbf{\rlem{bhvleqtwo-preserved}}
  Let $\genlts_\StatesA, \genlts_\StatesB \in \obaFW$.
  For every
  $X, X' \in \pparts{ \StatesA }$ and $ \serverB \in \StatesB$,
  such that $X \accleqset \serverB$, then
  \begin{enumerate}
  \item
    $\serverB \st{ \tau } \serverB'$ implies $X \accleqset \serverB'$,
  \item
    for every $\mu \in \Act$,
    if $X \convi$, then for every  $\serverB \st{\mu} \serverB'$ and set $X \wt{\mu} X'$
    we have $X' \accleqset \serverB'$.
  \end{enumerate}
  \newcommand{\attaboy}{\server_{\mathsf{attaboy}}}
  \begin{proof}
      To prove \rpt{bhvleqtwo-preserved-by-tau}
      fix a trace $\trace \in \Actfin$ such that
      $X \acnvalong \trace$.
      We have to explain why
      $\accht{ X }{ \trace } \ll \accht{ \serverB' }{ \trace }$.
      By unfolding the definitions, this amounts to showing that
      \begin{equation}
        \label{eq:attaboy-1}
        \tag{$\star$}
        \forall O \in \accht{ \serverB' }{\trace} \wehavethat
        \exists \attaboy \in X \suchthat
        \exists \widehat{O} \in \accht{ \attaboy }{ \trace } \suchthat
        \widehat{O} \subseteq O
      \end{equation}

      Fix a set $O \in \accht{\serverB'}{\trace}$.
      By definition there exists some $\serverB''$ such that
      $\serverB' \wt{\trace} \serverB'' \stable$, and that $O =
      O(\serverB'')$.
      The definition of $\accht{-}{-}$ and the silent move $\serverB \st{ \tau } \serverB'$
      ensures that $O \in \accht{\serverB}{ \trace }$.
      The hypothesis $ X \accleqset \serverB $ and that
      $X \cnvalong{ \trace } $ now imply that
      $\accht{ X }{ \trace } \ll \accht{ \serverB }{ \trace }$,
      which together with $ O \in \accht{\serverB}{ \trace }$
      implies exactly \req{attaboy-1}.

    We now prove \rpt{bhvleqtwo-preserved-by-mu}.
    To show $X' \accleqset \serverB'$ fix a trace $\trace \in \Actfin$
    such that $X' \cnvalong \trace$.

    We have to explain why $\accht{X}{\trace} \ll
    \accht{\serverB'}{\trace}$.
    By unfolding the definitions we obtain our aim,
      \begin{equation}
      \label{eq:aim}
      \tag{$\star\star$}
      \forall O \in \accht{ \serverB' }{ \trace } \wehavethat
      \exists \attaboy \in X' \suchthat
      \exists \widehat{ O } \in \accht{ \attaboy }{ \trace } \suchthat
      \widehat{ O } \subseteq O
    \end{equation}

      To begin with, we prove that $X \cnvalong \mu.\trace$.
      Since $X \wt{\mu} X'$ we know that $ X \wt{ \mu } X'$.
      This, together with $X \convi$ and $X' \cnvalong{ \trace} $
      implies the convergence property we are after.

      Now fix a set $ O \in \accht{ \serverB' }{ \trace }$.
      Thanks to the transition $ \serverB \st{ \mu } \serverB'$,
      we know that $ O \in \accht{ \serverB }{ \mu.\trace }$.
      The hypothesis $ X \accleqset \serverB $ together with
      $ X \cnvalong \mu.\trace$ implies that there exists
      a server $\attaboy \in X$ such that there exists an
      $ \widehat{O} \in \accht{ \attaboy }{ \mu.\trace }$.
      This means that $ \attaboy \wt{ \mu } \attaboy'$ and that
      $ \widehat{O} \in \accht{ \attaboy' }{ \trace } $.
      Since $X \wt{\mu} X'$ we know that $\attaboy' \in X' $
      and this concludes the argument.

  \end{proof}

\begin{lemma}
  \label{lem:ungood-cnv-mu}
  For every $\genlts_\StatesA \in \obaFW$, $\genlts_\StatesB \in \obaFB$,
  every set of processes $X \in \pparts{ \StatesA }$, every $\client \in \StatesB$, and every $\mu \in \Act$,
  if $\mustset{X}{\client}$, $\lnot \good{\client}$ and $\client \st{\mu}$ then $X \cnvalong \co{\mu}$.
\end{lemma}

\noindent%
\textbf{\rlem{soundness-set}}
Let $\genlts_\StatesA, \genlts_\StatesB \in \obaFW$ and
$\genlts_\StatesC \in \obaFB$.
For every set of processes $X \in \pparts{ \StatesA }$,
server $\serverB \in \StatesB$ and client $\client \in \StatesC$,
if $\mustset{X}{\client}$ and $X \asleqset \serverB$ then
$\musti{\serverB}{\client}$.
\begin{proof}
  We proceed by induction on the derivation of $\mustset{X}{\client}$.
  In the base case, $\good{\client}$ so we trivially derive
  $\musti{\serverB}{\client}$. %
  In the inductive case the proof of the hypothesis  $\mustset{X}{\client}$
  terminates with an application of \rname{Mset-step}.
  Since $\lnot \good{\client}$, we show the result applying \mstep.
  This requires us to prove that
  \begin{enumerate}[(1)]
  \item $ \csys{\serverB}{\client} \st{ \tau }$, and that
  \item for all $\serverB', \client'$ such that $\csys{\serverB}{\client} \st{ \tau } \csys{\serverB'}{\client'} $
    we have $\musti{\serverB'}{\client'} $.
  \end{enumerate}
  The first fact is a consequence of \rlem{stability-Nbhvleqtwo},
  which we can apply because $\lnot \good{\client}$ and thanks to the
  hypothesis $X \accleqset q$ and $\mustset{X}{\client}$.
  To prove the second fact, fix a transition $
  \csys{\serverB}{\client} \st{ \tau } \csys{\serverB'}{\client'}$. We
  have to explain why the following properties are true,
  \begin{enumerate}[(a)]
  \item\label{pt:aim-soundness-1}
    for every
    $\serverB' \suchthat \serverB \st{\tau} \serverB' \implies
    \musti{\serverB'}{\client}$,
  \item\label{pt:aim-soundness-2}
    for every
    $\client \suchthat \client \st{\tau} \client' \implies
    \musti{\serverB}{\client'}$,
  \item\label{pt:aim-soundness-3}
    for every $\serverB', \client'$ and $\mu \in \Act$,
    $\serverB \st{\mu} \serverB'$ and
    $\client \st{\co{\mu}} \client'$ imply
    $\musti{\serverB'}{\client'}$.
  \end{enumerate}

  First, note that $\mustset{ X  }{ \client }$, $\lnot
  \good{\client}$, and \rlem{mustx-terminate-ungood} imply~$X \conv$. Second, the inductive hypotheses
  state that for every $\client'$, non-empty set $X'$, and $\serverB$ the following facts hold,
  \begin{enumerate}[(i)]
  \item\label{soundness-IH1}
    $X \st{\tau} X'$ and $X' \asleqset \serverB$ implies
    $\musti{\serverB}{\client}$,
  \item\label{soundness-IH2}
    $\client \st{\tau} \client'$ and $X \asleqset \serverB$ implies
    $\musti{\serverB}{\client'}$,
  \item\label{soundness-IH3}
    for every and $\mu \in \Act$,
    $X \wt{\co{\mu}} X'$ and $\client \st{\mu} \client'$, and $X' \asleqset \serverB$ implies
    $\musti{\serverB}{\client'}$.
  \end{enumerate}

  To prove (\ref{pt:aim-soundness-1}) we use $X \conv$ and the hypothesis
  $X \cnvleqset q$ to obtain~$q \convi$. A rule induction on $q
  \convi$ now suffices: in the base case (\ref{pt:aim-soundness-1}) is
  trivially true and in the inductive case (\ref{pt:aim-soundness-1}) follows from
  \rptlem{bhvleqone-preserved}{bhvleqone-preserved-by-tau} and
  \rptlem{bhvleqtwo-preserved}{bhvleqtwo-preserved-by-tau},
  and the inductive hypothesis.

  The requirement (\ref{pt:aim-soundness-2}) follows directly from
  the hypothesis $X \asleqset \serverB$ and
  part (\ref{soundness-IH2}) of the inductive hypothesis.

  To see why (\ref{pt:aim-soundness-3}) holds, fix an action $\mu$ such that $\serverB \st{\mu} \serverB'$
  and $\client \st{\co{\mu}} \client'$.
  Since $\lnot \good{ \client }$ \rlem{ungood-cnv-mu} implies that $ X \cnvalong \mu $,
  and so \rlem{X-must-perform-visible-action} proves that $X \wt{ \mu }$.
  In turn this implies that there exists a $X'$ such that $X \wt{\mu} X'$, and thus
  \rptlem{bhvleqone-preserved}{bhvleqone-preserved-by-mu} and
  \rptlem{bhvleqtwo-preserved}{bhvleqtwo-preserved-by-mu}
  prove that $X' \asleqset \serverB'$ holds, and
  (\ref{soundness-IH3}) ensures the result, \ie that $\musti{\serverB'}{\client'}$.
    %
    %
\end{proof}

\section{Technical details about coinduction}
\label{sec:technicalities-coinductive-characterisation}

Here we outline the lemma that lets us prove \req{use-case-coind-preorder}.

\begin{lemma}
  \label{lem:technical-part-proof-by-coinduction}
  For every mailbox $\mailbox{M}$ and every set $X$ of processes, we have
\begin{equation}
   \label{eq:ch}
   \{ \tau.(\co{\aa} \Par \co{b}) \extc
   \tau.(\co{\aa} \Par \co{c}) \triangleright \mailbox{M} \} \cup X \;\coindleq\; \co{\aa} \Par (\tau.\co{b} \extc
   \tau.\co{c}) \triangleright \mailbox{M}.
   \tag{CH}
\end{equation}
\end{lemma}
\begin{proof}[sketch]
Let $\mailbox{M}$ be a mailbox and $X$ be a set of processes.
We focus on the last point of \rdef{coinductive-char-main}, for the other cases
are similar.
Consider an arbitrary reduction of the RHS of \req{ch}:
  \begin{equation}
    \label{eq:hoisting-red}
  \co{\aa} \Par (\tau.\co{b} \extc \tau.\co{c}) \triangleright \mailbox{M}
  \;\st{\mu}\; \serverB
  \end{equation}
We proceed by case analysis.
  \begin{enumerate}
\item%
The cases that correspond to a reduction of the mailbox $\mailbox{M}$
are dealt with directly using the coinductive hypothesis \eqref{eq:ch},
since the mailbox is quantified universally in \eqref{eq:ch}.
In more detail, consider the case where the reduction \eqref{eq:hoisting-red} is of the form:
\begin{equation*}
\co{\aa} \Par (\tau.\co{b} \extc \tau.\co{c}) \triangleright \mailbox{M}
\;\st{b}\; \co{\aa} \Par (\tau.\co{b} \extc \tau.\co{c}) \triangleright
  \mailbox{\{ \co{b} \} \uplus M}
\end{equation*}
      Let $X'$ be $\WD(b, \{ \tau.(\co{\aa} \Par \co{b}) \extc \tau.(\co{\aa} \Par
      \co{c}) \triangleright \mailbox{M} \} \uplus X)$, that is, the unique set of processes such that
\[
  \{ \tau.(\co{\aa} \Par \co{b}) \extc \tau.(\co{\aa} \Par \co{c}) \triangleright \mailbox{M} \} \cup X
      \wt{b} X'
\]
It is easy to prove that 
\(
  \{ \tau.(\co{\aa} \Par \co{b}) \extc \tau.(\co{\aa} \Par \co{c})
      \triangleright \{ \co{b} \} \uplus \mailbox{M} \} \;\in\; X'
\)
and thus we conclude by applying the coinductive hypothesis~\eqref{eq:ch}.
\item
If \eqref{eq:hoisting-red} corresponds to a transition of the process,
it must be that $\mu = \co{a}$ and $q = \tau.\co{b} \extc
\tau.\co{c} \triangleright \mailbox{M}$.
In that case, the set $X'$ of processes reached from
$\{ \tau.(\co{\aa} \Par \co{b}) \extc \tau.(\co{\aa} \Par \co{c})
\triangleright \mailbox{M} \} \uplus X$ while outputting $a$
contains $\co{b}$ and $\co{c}$, so that $X' \coindleq q$ follows from
the general fact that $\{p, q\} \testleqSset \tau. p + \tau. q$. \qed

\end{enumerate}
\end{proof}

\section{Traces in normal form}
\label{sec:normal-forms}

Let $\MI$ be the set of multisets of input actions, and let $\chopSym : \Actfin \longrightarrow ( \MI \times \MO )^\star$ 
be the function 
$$
\chop{ \trace } = (\I_0, M_0),( \I_1, M_2), \ldots , (\I_n,M_n)
$$
which is defined inductively in \rfig{nf-trace-def}.  The intuition is
that given a trace $\trace$, the function $\chopSym$ forgets the order
of actions in sequences of consecutive inputs, and in sequences of
consecutive outputs, thereby transforming them in multisets. On the
other hand $\chopSym$ preserves the causal order in the trace, in the
sense that the execution of all inputs in $\I_k$ is necessary to execute any of the outputs in $M_k$,
as well as the actions in the multisets at index $h$ with $k < h$.
For instance:
$$
\chop{ c\aa\co{bdd}\aa\co{ef}e} =
(\mset{c, \aa}, \mset{\co{\ab},\co{d},\co{d}}),
(\mset{\aa}, \mset{\co{e},\co{f}}),
(\mset{e},\varnothing)
$$

\begin{figure}[t]
  \hrulefill
$$
\begin{array}{rcl}
  \chopSym(\varepsilon) & = & \varepsilon \\[3pt]
  \chopSym(s) & = & \chopSym'(s, \varnothing, \varnothing) \\[10pt]

  \chopSym'(\varepsilon, I, M) & = & (I,M) \\[3pt]
  \chopSym'(\aa.s, I, M) & = & \chopSym'(s, \mset{\aa} \uplus I, M)\\[3pt]
  \chopSym'(\co{\aa}.s, I, M) & = & 
    \begin{cases}
                                          (I, \mset{\co{\aa}} \uplus M), \chopSym'(s, I,
    \varnothing)& \text{if $s = b.s'$ for some $b, s'$} \\[3pt]
     \chopSym'(s, I,
    \mset{\co{\aa}} \uplus M)& \text{otherwise} 
\end{cases}                                
\end{array}
$$


  \caption{Definition of the trace normalisation function $\chopSym$}
  \label{fig:nf-trace-def}
  \hrulefill
\end{figure}

Let $\sigma$ range over the set $( \MI \times \MO )^\star$. We
say that $\sigma$ is a trace in {\em normal form}, and we write
$ p \wt{ \sigma } q$, whenever there exists $\trace \in \Actfin$
such that $ p \wt{ \trace } q$ and $\chop{\trace} = \sigma $.

\begin{definition}
  \label{def:asyn-leq}%
We lift in the obvious way the predicates $\bhvleqone, \bhvleqtwo, $
and $\asleqAfw$ to traces in normal forms.
For every $\genlts_\StatesA, \genlts_\StatesB$ and $\serverA \in \StatesA,
  \serverB \in \StatesB$, let
\begin{itemize}
  \item 
    $\serverA \asynleqone \serverB$
    to mean 
  $\forall \sigma \in (\MI \times \MO)^\star \wehavethat \serverA \cnvalong{ \sigma }$
    implies $\serverB \cnvalong{ \sigma }$,
    \item 
      $\serverA \asynleqtwo \serverB$ to mean
      $\forall \sigma \in (\MI \times \MO)^\star \wehavethat \serverA \acnvalong{\sigma}$
      implies $\accht{ \serverA }{ \sigma } \ll \accht{ \serverB }{ \sigma }$,
  \item $\serverA \asleqNF \serverB$ whenever
    $\serverB \asynleqone \serverA  \wedge  \serverA \asynleqtwo \serverB$.
    $\hfill\blacksquare$
  \end{itemize}
\end{definition}

In LTSs with forwarding, 
the transition relation $\st{}$ is {\em input-receptive}
(Axiom (IB4), Table 2 of \cite{DBLP:conf/concur/Selinger97}),
and thus \rlem{weak-a-swap} shows that $\st{}$ enjoys a restricted version
of \restrictedinputcommutativity, and that so does its weak version.
Sequences of input actions $s \in \Names^\star$ enjoy a form of
diamond property in $\wt{}$. The crucial fact pertains consecutive input actions. Recall that $\simeq$ denotes any equivalence that satisfies the compatibility property depicted in \rfig{Axiom-LtsEq}.

\begin{lemma}
  \label{lem:weak-a-swap}
  For every $\genlts = \lts{A}{\Act}{\st{}} \in \obaFW,$ every $\serverA, \serverB \in \States$ and $\aa, \ab \in \Names$,
  if $p \wt{\aa.\ab} q$ then $p \wt{\ab.\aa} \cdot \simeq q$.
\end{lemma}

\rlem{weak-a-swap}, together with an induction on traces, allows
us to prove that $\chopSym$ preserves convergence and acceptance sets.

\begin{lemma}
  \label{lem:normalisation-preserves-predicates}
  For every $\genlts = \lts{A}{\Act}{\st{}} \in \obaFW$, every $\server \in \States$ and $\trace \in \Actfin$ we have that
  \begin{enumerate}
  \item 
    $ p \wt{ \trace } q$ iff $ p \wt{ \chop{\trace} } \cdot \simeq q$,
    and if the first trace does not pass through a successful state then the normal form does not either,
  \item 
    $ p \acnvalong{\chop{\trace} }$ iff $ p \acnvalong{ \trace }$,
  \item 
    $ \accht{p}{ \trace } = \accht{p}{ \chop{\trace} }$.
  \end{enumerate}
\end{lemma}
\noindent

We thereby obtain two other characterisations of the contextual preorder~$\testleqS$:
\rthm{testleqS-equals-bhvleq} and \rlem{normalisation-preserves-predicates} ensure that the
preorders $\testleqS$ and $\asleqNF$ coincide.
\begin{corollary}
  \label{cor:asynleq-equals-bhvleq}
  For every $\genlts_\StatesA, \genlts_\StatesB \in \obaFB$,
  every $\serverA \in \StatesA$ and $\serverB \in \StatesB$,
  \[
    \serverA \testleqS \serverB \;\;\;\text{iff}\;\;\;
  \liftFW{\serverA} \asleqNF \liftFW{\serverB}.
  \]
\end{corollary}

\renewcommand{\mailbox}[1]{\fcolorbox{gray}{yellow!40}{\ensuremath{#1}}}

\section{Asynchronous CCS}
\label{sec:accs}

\begin{figure}[t]
\hrulefill
$$
\begin{array}{l@{\hskip 3pt}ll@{\hskip 3pt}ll@{\hskip 3pt}l}
\rinput
&
  \begin{prooftree}
    \justifies
    \aa.p \st{ \aa } p
  \end{prooftree}
&
  \rtau
&
  \begin{prooftree}
    \justifies
    \tau.p \st{ \tau } p
  \end{prooftree}
  &

\\[2em]

\mboxelim
&
\begin{prooftree}
  \justifies
  \mailbox{\co{\aa}} \st{\co{\aa}} \Nil
\end{prooftree}


&
\unfold
&
  \begin{prooftree}
    \justifies
    \rec{p} \st{\tau} p \subst{ \rec{p} }{ x }
  \end{prooftree}
&

 \\[2em]

  \extL
  &
\begin{prooftree}
  p \st{  \alpha } p'
  \justifies
  p \extc q \st{  \alpha } p'
\end{prooftree}

&
\extR
&
\begin{prooftree}
  q \st{ \alpha } q'
  \justifies
  p \extc q \st{ \alpha } q'
\end{prooftree}

\\[2em]
  \parL
  &
\begin{prooftree}
{ \state \st{\alpha} \stateA }
  \justifies
      {\state \Par \stateB \st{\alpha} \stateA \Par \stateB}
\end{prooftree}

&
\parR&
\begin{prooftree}
  \stateB \st{\alpha} \stateB'
  \justifies
      {\state \Par \stateB \st{\alpha} \state \LTSPar \stateB'}
\end{prooftree}

\\[2em]
\com&
\begin{prooftree}
  \state \st{ \mu } \stateA \quad \stateB \st{\co{ \mu }} \stateB'
  \justifies
  \state \LTSPar \stateB \st{\tau} \stateA \LTSPar \stateB'
\end{prooftree}

\end{array}
$$
\caption{The LTS of processes. 
  The meta-variables are
  $\aa \in \Names, \mu \in \Act, \and \alpha \in
  \Acttau$.
}  
\label{fig:rules-LTS}
\hrulefill
\end{figure}

Here we recall the syntax and the LTS of asynchronous \CCS, or \ACCS
for short, a version of \CCS where outputs have no continuation
and sum is restricted to input- and $\tau$-guards.
$$
p,q,r ::= 
\out{\aa}
\BNFsep g
\BNFsep p \Par p
\BNFsep \rec{p}
\BNFsep x,
\qquad
g ::= 
\Nil
\BNFsep \Unit 
\BNFsep a.p
\BNFsep \tau.p
\BNFsep g \extc g
$$

This calculus,
which is inspired by the variant of the asynchronous $\pi$-calculus
considered by~\cite{ACS96,ACS98} for their study of asynchronous
bisimulation, was first investigated by~\cite{DBLP:conf/concur/Selinger97},
and subsequently resumed by other authors such as~\cite{DBLP:journals/iandc/BorealeNP02}.
  Different asynchronous variants of~\CCS were studied in the
  same frame of time by~\cite{pugliesephd}, whose calculus included
  output prefixing and operators from ACP, and
  by~\cite{DBLP:conf/fsttcs/CastellaniH98}, whose calculus TACCS
  included asynchronous output prefixing and featured two forms of
  choice, internal and external, in line with previous work on testing
  semantics~\cite{DBLP:conf/tapsoft/NicolaH87}.

The syntax of terms is given in \req{syntax-processes}.
As usual,~$\rec{p}$ binds the variable~$x$~in~$p$, and we use
standard notions of free variables, open and closed terms.
Processes, ranged over by $p, q, r, \dots$ are {\em closed} terms.
The operational semantics of processes is given by the LTS
$\lts{\ACCS}{\Acttau}{\st{}}$ specified by the rules in \rfig{rules-LTS}.

The prefix $ \aa.p $ represents a {\em blocked} process, which waits
to perform the
input $\aa$, \ie to interact with the atom $\co{ \aa}$, and then becomes $p$; and atoms
$\mailbox{\co{ \aa }}, \mailbox{\co{ \ab }}, \ldots$ represent output
messages.  We will discuss in detail the role played by atoms in the
calculus, but we first overview the rest of the syntax. We
include~$\Unit$ to syntactically denote successful states.
The prefix $ \tau.p $ represents a process that does one step of
  internal computation and then becomes~$p$.
  The sum $g_1 \extc g_2$ is a process that can behave as $g_1$ or $g_2$, but not both.
  Thus, for example $ \tau.p \extc \tau.q$ models an \texttt{if \ldots then \ldots else}, while
$ \aa.p \extc \ab.q$ models a \texttt{match \ldots with}.
Note that the sum operator is only defined on \emph{guards}, namely it can only
take as summands $\Nil, \Unit$ 
  or input-prefixed and $\tau$-prefixed
processes.
While the restriction to guarded sums is a standard one,
widely adopted in process calculi,
the restriction to input and $\tau$ guards is specific to asynchronous
calculi. We will come back to this point after discussing atoms and
mailboxes.
%
Parallel composition $p \Par q$ runs $p$ and $q$ concurrently,
allowing them also to interact with each other, thanks to rule \com. For example
\begin{equation}
  \label{eq:example1}
  \ab.\aa.\Nil \Par \ab.c.\Nil \Par \mailbox{ ( \co{\aa} \Par \co{\ab} \Par \co{c} ) }
\end{equation}
represents a system in which two concurrent processes,
namely $\ab.\aa.\Nil$ and $ \ab.c.\Nil $, are both
ready to consume the message $ \co{b} $
from a third process, namely $\mailbox{\co{a} \Par \co{b} \Par \co{c}}$.
This last process is a parallel product of atoms, and it is not
guarded, hence it is best viewed as an unordered mailbox shared by {\em all}
the processes running in parallel with it. For instance in (\ref{eq:example1})
the terms $ b.a.\Nil$ and $b.c.\Nil$ share the mailbox $\mailbox{ \co{a} \Par \co{b}
  \Par \co{c} }$.  Then, depending on which process consumes $\co{b}$,
the overall process will evolve to either
$b.c.\Nil \Par \mailbox{\co{c}} $ or $b.a.\Nil \Par \mailbox{\co{a}}$,
which are both stuck.\footnote{The global shared mailbox that we
    treat is reminiscent but less general than the chemical ``soup''
    of \cite{DBLP:journals/tcs/BerryB92}.  In that context the
    components of the soup are not just atoms, but whole parallel
    components: in fact, the chemical soup allows parallel components
    to come close in order to react with each other, exactly as
      the structural congruence of~\cite{DBLP:conf/icalp/Milner90}, which indeed was
    inspired by the Chemical Abstract Machine.}

Concerning the sum construct, we follow previous work on
asynchronous calculi
(\cite{ACS96,ACS98,DBLP:books/daglib/0004377,DBLP:journals/iandc/BorealeNP02})
and only allow input-prefixed or $\tau$-prefixed terms as
  summands. 
The reason for forbidding atoms in sums is that the
  \nondeterministic sum is essentially a synchronising operator: the
  choice is solved by executing an action in one of the summands and
  \emph{simultaneously} discarding all the other summands. Then, if an
  atom were allowed to be a summand, this atom could be discarded by
  performing an action in another branch of the choice. This would
  mean that a process would have the ability to withdraw a message
  from the mailbox without consuming it, thus contradicting the
  intuition that the mailbox is a shared entity which is out of the
  control of any given process, and with which processes can only
  interact by feeding a message into it or by consuming a message from
  it.  In other words, this restriction on the sum operator
ensures that atoms 
indeed represent messages in a global
mailbox. For further details see the discussion on page 191 of
\cite{DBLP:books/daglib/0004377}.

A structural induction on the syntax ensures that processes perform only a finite
number of outputs:
\begin{lemma}
  \label{lem:output-sets-finite}
  For every $\state \in \ACCS \wehavethat {\cardinality{O(\state)}} \in \N$.
\end{lemma}
\noindent
Together with \rlem{output-shape}, this means that at any point
of every execution the global mailbox contains a finite number of
messages. Since 
the LTS is image-finite under any visible action,
a consequence of \rlem{output-sets-finite} is that the number of reducts
of a program is finite.

\renewcommand{\stateA}{p'}
\begin{lemma}
  \label{lem:st-finite-image}
  \label{lem:sttau-finite-image}
  For every $\state \in \ACCS \wehavethat {\cardinality{\reducts{
        \state }{\ACCS}{\st}}} \in \N.$
\end{lemma}
\begin{proof}
  Structural induction on~$p$. The only non-trivial case is if $ p = p_1 \Par p_2$.
  In this case the result is a consequence of the inductive hypothesis, of \rlem{output-sets-finite} and
  of the following fact: $ p \st{\tau} q $ iff
  \begin{enumerate}
    \item $p_1 \st{\tau} p'_1 $ and $q = p'_1 \Par p_2$,
    \item $p_2 \st{\tau} p'_2 $ and $q = p_1 \Par p'_2$,
    \item $p_1 \st{ a } p'_1 $ and $ p_2 \st{\co{a}} p'_2 $ and  $q =  p'_1 \Par p'_2$,
    \item $p_1 \st{ \co{a} } p'_1 $ and $ p_2 \st{ a } p'_2 $ and  $q =  p'_1 \Par p'_2$.
  \end{enumerate}
  In the third case the number of possible output actions $\co{a}$ is finite thanks to \rlem{output-sets-finite},
  and so is the number of reducts $p'_1$ and $p'_2$, so the set of term $ p'_1 \Par p'_2 $ is decidable.
  The same argument works for the fourth case.
\end{proof}

Thanks to 
\rlem{output-sets-finite} and
\rlem{sttau-finite-image}, \rlem{st-finite-image} holds also for the LTS modulo structural congruence,
i.e. \lts{\modulo{\ACCS}{\equiv}}{\modulo{\st{}}{\equiv}}{\Acttau}.


\subsection{Structural equivalence and its properties}
\label{sec:equiv}
\label{sec:structural-congruence}

\begin{figure}[t]
  \hrulefill
  \begin{minted}{coq}
    Class LtsEq (A L : Type) `{Lts A L} := {
      eq_rel : A → A → Prop;

      eq_rel_refl p : eq_rel p p;
      eq_symm p q : eq_rel p q → eq_rel q p;
      eq_trans p q r :
      eq_rel p q → eq_rel q r → eq_rel p r;

      eq_spec p q (α : Act L) :
      (∃ p', (eq_rel p p') ∧ p' ⟶{α} q)
      →
      (∃ q', p ⟶{α} q' ∧ (eq_rel q' q))
    }.
  \end{minted}
  \caption{A typeclass for LTSs where a structural congruence exists over states.}
  \label{fig:LtsEq}
  \hrulefill
\end{figure}

\begin{figure}[t]
$$
\begin{array}{r@{\hskip 3pt}ll}

\\[1em]
\rulename{S-szero} & p \extc \Nil \equiv p \\
\rulename{S-scom} & p \extc q \equiv q \extc p \\
\rulename{S-sass} &  (p \extc q) \extc r \equiv p \extc (q \extc r)
\\[1em]
\rulename{S-pzero} & p \Par \Nil \equiv p \\
\rulename{S-pcom} & p \Par q \equiv q \Par p \\
\rulename{S-pass} &  ( p \Par q ) \Par r \equiv p \Par (q \Par r)
\\[1em]
\rulename{S-refl}& p \equiv p \\
\rulename{S-symm} & p \equiv q & \mathit{if} \ q \equiv p\\
\rulename{S-trans} & p \equiv q & \mathit{if} \ p \equiv p' \ \mathit{and} \ p' \equiv q
\\[1em]
\rulename{S-prefix} & \alpha.p \equiv \alpha.q  & \mathit{if} p \equiv q\\
\rulename{S-sum} & p \extc q \equiv p' \extc q & \mathit{if}  p \equiv p'\\
\rulename{S-ppar} & p \Par q \equiv p' \Par q & \mathit{if}  p \equiv p'
\end{array}
$$
\caption{Rules to define structural congruence on \ACCS.}
\label{fig:equiv}
\hrulefill
\end{figure}

To manipulate the syntax of processes we use
a standard structural congruence denoted~$\equiv$,
stating that~$\ACCS$ is a commutative monoid with identity~$\Nil$
with respect to both sum and parallel composition. 

A first fact is the following one.
\begin{lemma}
  \label{lem:weak-output-swap}
  For every $\mu \in \co{\Names}$ and $\alpha \in \Acttau$,
  if $p \wt{\mu.\alpha} q$
  then $p \wt{\alpha.\mu} \cdot \equiv q$.
\end{lemma}



As sum and parallel composition are commutative monoids, we
use the notation
$$
\begin{array}{lll}
  \Sigma \set{g_0, g_1, \ldots g_n} &\text{ to denote }&g_0 \extc g_1 \extc \ldots \extc g_n \\
  \Pi \set{p_0, p_1, \ldots p_n}&\text{ to denote }&p_0 \Par p_1 \Par \ldots \Par p_n
\end{array}
$$
This notation is useful to treat the global shared mailbox.
In particular, if $\mset{ \mu_0, \mu_1, \ldots \mu_n} $
is a multiset of output actions, then the syntax $\Pi  \mset{ \mu_0,
  \mu_1, \ldots \mu_n} $ represents the shared mailbox that contains
the messages $\mu_i$; for instance we have $\Pi \mset{ \co{a}, \co{a}, \co{c}
} = \mailbox{ \co{ a  } \Par \co{ a } \Par \co{c}}$.
We use the colour $\mailbox{ - }$ to highlight the content of the mailbox.
Intuitively a shared mailbox contains the messages
that are ready to be read, i.e. the outputs 
that are immediately available (i.e. not guarded by any prefix operation).
For example in
$$
\mailbox{\co{c}} \Par a.(\co{b} \Par c.d.\Unit) \Par \mailbox{\co{d}} \Par \tau.\co{e}
$$
the mailbox is $\mailbox{\co{c} \Par \co{d}}$.
The global mailbox that we denote with $\mailbox{ - }$ is exactly the
  buffer $B$ in the {\em configurations} of
  \cite{DBLP:phd/us/Thati03}, and reminiscent of the $\omega$ used by
  \cite{DBLP:conf/fossacs/BravettiLZ21}. The difference is that $\omega$
  represents an unbounded {\em ordered} queue, while our mailbox is
  an unbounded {\em unordered} buffer.

As for the relation between output actions in the LTS and the global
mailbox, an output $\co{\aa}$ can take place if and only if 
the message $\co{\aa}$ appears in the mailbox:
\begin{lemma}
  \label{lem:output-shape}
  For every $\state \in \ACCS$,
  \begin{enumerate}
    \item
      for every $\aa \in \Names \wehavethat \state \st{
        \co{\aa} } \stateA$ implies 
      $\state \equiv  \stateA \Par
      \mailbox{\co{\aa}}$,
    \item
      there exists $ \stateA $ such that
      $\state \equiv  \stateA \Par \mailbox{ \Pi M }$,
      and  $\stateA$ performs no output action.
  \end{enumerate}
\end{lemma}
\noindent
This lemma and \rlem{weak-output-swap} essentially hold, because, as already pointed out in \rsec{preliminaries}, the syntax enforces outputs to have no continuation.

The following lemma states a fundamental fact
(\cite[Lemma 2.13
]{DBLP:books/daglib/0018113},
\cite[Proposition 5.2
]{DBLP:books/daglib/0098267},
\cite[Lemma 1.4.15]{DBLP:books/daglib/0004377}).
Its proof is so tedious that even the references we have given only
sketch it. In this paper we follow the masters example, and give
merely a sketch. However, we have a complete machine-checked proof.

\begin{lemma}
  \label{lem:st-compatible-with-equiv}
  For every $p,q \in \ACCS$ and $\alpha \in \Act_{\tau} \wehavethat p \equiv \cdot \st{\alpha} q \implies p  \st{\alpha} \cdot \equiv q$.
\end{lemma}
\begin{proof}[Proof sketch]
  We need to show that if there exists a process $p'$ such that $p \equiv p'$
  and $p' \st{\alpha} q$ then there exists a process $q'$ such that $p \st{\alpha} p'$
  and $p' \equiv q$.
  The proof is by induction on the derivation $p \equiv p'$.

  We illustrate one case with the rule \rulename{S-trans}.
  The hypotheses tell us that there exists $\hat{p}$ such that $p \equiv \hat{p}$ and $\hat{p} \equiv p'$,
  that $p' \st{\alpha} q$,  and the inductive hypotheses that
  \begin{enumerate}[(a)]
  \item\label{pt:preharmony-case-1}
    for all $q'$ \textit{s.t} $\hat{p} \st{\alpha} q'$ implies that there exists a $\hat{q}$ such that $p \st{\alpha} \hat{q}$ and $\hat{q} \equiv q'$
  \item\label{pt:preharmony-case-2}
    for all $q'$ \textit{s.t} $p' \st{\alpha} q'$ implies that there exists a $\hat{q}$ such that $\hat{p} \st{\alpha} \hat{q}$ and $\hat{q} \equiv q'$
  \end{enumerate}
  By combining \rpt{preharmony-case-2} and $p' \st{\alpha} q$ we obtain a
  $\hat{q}_1$ such that $\hat{p} \st{\alpha} \hat{q}_1$ and $\hat{q}_1 \equiv q$.
  Using \rpt{preharmony-case-1} together with $\hat{p}_1 \st{\alpha} \hat{q}_1$ we have that
  there exists a $\hat{q}_2$ such that $p \st{\alpha} \hat{q}_2$ and $\hat{q}_2 \equiv \hat{q}_1$.
  We then have that $p \st{\alpha} \hat{q}_2$ and it remains to show that $\hat{q}_2 \equiv q$.
  We use the transitivity property of the structural congruence relation to show that
  $\hat{q}_2 \equiv \hat{q}_1$ and $\hat{q}_1 \equiv q$ imply $\hat{q}_2 \equiv q$ as required
  and we are done with this case.
\end{proof}
\noindent
Time is a finite resource. The one spent to machine check
\rlem{st-compatible-with-equiv} would have been best invested into bibliographical research.
Months after having implemented the lemma we realised that
\cite{DBLP:journals/entcs/AffeldtK08} already had 
an analogous result for a mechanisation of the $\pi$-calculus.
\rlem{st-compatible-with-equiv} is crucial to prove the Harmony Lemma,
which states that~$\tau$-transitions coincide with the standard reduction
relation of \ACCS. This is out of the scope of our discussion,
and we point the interested reader to Lemma 1.4.15 of
\cite{DBLP:books/daglib/0004377}, and to the list of problems
presented on the web-page of \textsc{The Concurrent Calculi Formalisation Benchmark}.\footnote{\url{https://concurrentbenchmark.github.io/}}

We give a corollary that is useful to prove \rlem{ACCSmodulos-equiv-is-out-buffered-with-feedback}.
\begin{corollary}
  \label{cor:equiv-preserves-transitions-modulo-equiv}
  \label{cor:equiv-preserves-transitions}
For every $p,q \in \ACCS, \and \alpha \in \Acttau \wehavethat p \equiv q$ implies that
$p \st{ \alpha } \cdot \equiv r$ if and only if $q \st{ \alpha } \cdot \equiv r$.
\end{corollary}
\begin{proof}
  Since $q \equiv p \st{ \alpha } p' \equiv r$ \rlem{st-compatible-with-equiv} implies $q \st{\alpha} \cdot \equiv p'$,
  thus $q \st{\alpha} \cdot \equiv r$ by transitivity of $\equiv$.
  The other implication follows from the same argument and the symmetry of $\equiv$.
\end{proof}

A consequence of \rlem{output-shape} is that the LTS
\lts{\modulo{\ACCS}{\equiv}}{\modulo{\st{}}{\equiv}}{\Acttau}
enjoys the axioms in \rfig{axioms}, and thus it is \obaFB.
\cite[Theorem 4.3]{DBLP:conf/concur/Selinger97} proves it reasoning modulo
bisimilarity, while we reason modulo structural equivalence.
\begin{lemma}
\label{lem:ACCSmodulos-equiv-is-out-buffered-with-feedback}
$\Forevery p \in \ACCS,$ and $\aa \in \Names$ the following properties are true,
\begin{itemize}
\item
  $\forevery \alpha \in \Acttau \wehavethat p \st{\co{\aa}}\st{\alpha} p_3$ implies $ p \st{ \alpha }\st{ \co{\aa}} \cdot \equiv p_3$;
\item
  $\forevery \alpha \in \Acttau \suchthat
  \alpha \not\in\set{ \tau, \co{\aa}} \wehavethat p \st{\co{\aa}} p'
  \text{ and } p \st{\alpha} p''$ imply that $p'' \st{\co{\aa}} q \text{ and }p' \st{\alpha} q$ for some $q$;
\item
  $ p \st{\co{\aa}} p' \text{ and } p \st{\co{\aa}} p'' \imply p' \equiv p''$;
\item
  $p \st{\co{\aa}} p' \st{\aa} q \implies p \st{\tau} \cdot \equiv  q$;
\item
  $ p \st{\co{\aa}} p' \text{ and } p \st{\tau} p'' \imply$ that $p' \st{\tau} q$ and $p'' \st{\co{\aa}} q$; or that $p' \st{\tau} p''$.
  \item  $\forevery p'$ if there exists a $\hat{p}$ such that
   $p \st{\co{\aa}} \hat{p}$ and $p' \st{\co{\aa}} \hat{p}$
   then $p \equiv p'$
\end{itemize}
\end{lemma}
\begin{proof}
  To show \axiom{feedback} we begin via \rlem{output-shape}
  which proves $p \equiv p' \Par \mailbox{\co{\aa}}$.
  We derive $p' \Par \mailbox{\co{\aa}} \st{ \tau } q'$
  and apply \rcor{equiv-preserves-transitions-modulo-equiv} to obtain
  $ p \st{\tau} \cdot \equiv q $.

  We prove \axiom{Output-Tau}.
  The hypothesis and \rlem{output-shape} imply that $ p \equiv p' \Par \mailbox{\aa}$.
  Since $ p \st{ \tau } p'' $ it must be the case that
  $ p' \st{ \tau } \hat{p}$ for some $ \hat{p}$,
  and $p'' =  \hat{p} \Par \mailbox{\aa}$.
  Let $q = \hat{p}$. We have that $ p'' \st{ \aa } \hat{p} \Par \Nil \equiv q$.
\end{proof}
\noindent
Processes that enjoy \axiom{Output-Tau} are called {\em
non-preemptive} in \cite[Definition
10]{DBLP:conf/lics/CleavelandZ91}.

Each time a process $\state$ reduces to a {\em stable} process
$\stateA$, it does so by consuming at least part of the mailbox, for
instance a multiset of outputs $N$, thereby arriving in a state
$\stateB$ whose inputs cannot interact with what remains of the
mailbox, i.e. $M \setminus N$, where $M$ is the original mailbox.

\begin{lemma}
  \label{lem:completeness-part-2.2-squigly-02}
  $\Forevery M \in \MO$, $p, \stateA \in \ACCS, $
  if $p \Par \mailbox{ \Pi M } \wt{\varepsilon} \stateA \stable$
  then there exist an $N \subseteq M$ and some $\stateB \in \ACCS$
  such that $p \wt{ \co{N} } \stateB  \stable$, $O(\stateB) \subseteq O( \stateA )$, and
  $\disjoint{\co{ I( \stateB ) }}{(M \setminus N)}$.
\end{lemma}
\begin{proof}By induction on the derivation of $p \Par \mailbox{ \Pi M } \wt{\varepsilon} \stateA$.%
  In the base case this is due to \wtrefl, which ensures that
  $$
  p  \Par \mailbox{ \Pi M } = \stateA,
  $$
  from which we obtain $p  \Par \mailbox{ \Pi M } \stable$.
  This ensures that  $\disjoint{ \co{I(p)} }{ M }$.
  We pick as $\stateB$ and $N$ respectively $p  \Par \mailbox{ \Pi M }$
  and $\emptyMset$ as $p \Par \mailbox{ \Pi M } \wt{ \emptyMset } p \Par \mailbox{ \Pi M }$ by reflexivity,
  and $O( p  \Par \mailbox{ \Pi M } ) = O( \stateA )$.

    In the inductive case the derivation ends with an application of \wttau and
    $$
    \begin{prooftree}
      \state  \Par \mailbox{ \Pi M } \st{\tau} p'
      \quad
      \begin{prooftree}
        \vdots
        \justifies
        p'\wt{\varepsilon} \stateA
        \end{prooftree}
      \justifies
       \state  \Par \mailbox{ \Pi M } \wt{\varepsilon} \stateA
    \end{prooftree}
    $$

    We continue by case analysis on the rule used to infer the transition $\state  \Par \mailbox{ \Pi M  } \st{\tau} p'$.
    As by definition $\mailbox{ \Pi M } \stable$, the rule
    is either \parL, i.e. a $\tau$-transition performed by $\state$,
    or \com, i.e. an interaction between $\state$ and $\mailbox{ \Pi M  }$.

    \paragraph{Rule \parL:}
    In this case $ \state \st{\tau} p'' $ for some $p''$, thus $p''
    \Par \mailbox{ \Pi M  } \wt{ \varepsilon } \stateA $ and the
    result follows from the inductive hypothesis.

    \paragraph{Rule \com:}
    The hypothesis of the rule ensure that~$\state \st{\aa} p''$
    and~$\mailbox{  \Pi M  } \st{\co{\aa}} \stateA$,
    and as the process~$\mailbox{ \Pi M  }$ does not perform any
    input, it must be the case that
    $\aa \in \Names$, that  $\co{\aa} \in M$, and that $q \equiv
    \mailbox{ \Pi (M \setminus \mset{\co{\aa}} ) }$.\footnote{In terms
    of LTS with mailboxes, $\stateA = (M \setminus \mset{\co{\aa}})
    $.}
    Note that
    $p' \equiv p'' \Par \mailbox{  \Pi (M \setminus \mset{\co{\aa}}) } $.

    The inductive hypothesis ensures that for some $N' \subseteq M
    \setminus \mset{\co{\aa}}$ and some $\stateC \in \ACCS$
    we have
    \begin{enumerate}[(a)]
    \item $p' \wt{ \co{N'} } \stateC  \stable$,
    \item\label{completeness-part-2.2-squigly-02-output-subseteq-IH}
      $O(\stateC) \subseteq O(\stateA)$, and
    \item\label{completeness-part-2.2-squigly-02-input-subseteq-IH}
      $\disjoint{\co{ I(\stateC) }}{((M \setminus \mset{ \co{\aa} })
      \setminus N')}$
    \end{enumerate}
    We conclude by letting $ \stateB = \stateC$, and $N = \mset{\aa} \uplus N'$.
    The trace $p \st{\aa} p' \wt{ N' } \stateC$ proves that $p \wt{
      \mset{\aa} \uplus N' } \stateC$,
    moreover we already know that $\stateC$ is stable.
    The set inclusion $O(\stateC) \subseteq O( \stateA )$ follows from
    \ref{completeness-part-2.2-squigly-02-output-subseteq-IH},
    and lastly $\disjoint{\co{I(\stateB)}}{(M \setminus
      (\mset{\co{\aa}} \uplus N'))}$ is a consequence of
    $\disjoint{\co{ I(\stateC) }}{((M \setminus \mset{ \co{\aa} })
      \setminus N')}$ and of
    $ (M \setminus \mset{\co{\aa}}) \setminus N' = (\I \setminus
    (\mset{\co{\aa}} \uplus N'))$.
\end{proof}

We define the predicate $\goodSym$,
 $$
 \begin{array}{lll}
   \good{\Unit} \\
   \good{p \Par q}  &\mathit{if}\ \good{p}\, \mathit{or}\, \good{q}\\
   \good{p \extc q} &\mathit{if}\ \good{p}\, \mathit{or}\, \good{q}\\
 \end{array}
 $$
 \noindent
 This predicate is preserved by structural congruence.
 \begin{lemma}
  \label{lem:happy-sc}
  $\Forevery \serverA, \serverB \in \ACCS \wehavethat \serverA \equiv \serverB$ and $\good{\serverA} \imply \good{\serverB}$.
\end{lemma}

\begin{lemma}
  \label{lem:terminate-sc}
  For every $p, q \in \ACCS \wehavethat p \equiv q$ and $p \convi$ imply $q \convi$.
\end{lemma}

\begin{lemma}
  \label{lem:acnv-sc}
  For every $\serverA, \serverB \in \ACCS$ and $\trace \in \Actfin$, we have that $\serverA \equiv \serverB$ and $\liftFW{\serverA} \cnvalong \trace$
  imply $\liftFW{\serverB} \cnvalong \trace$.
\end{lemma}

\begin{lemma}
  \label{lem:must-sc-client}
  $\Forevery  \serverA, \client, \client' \in \ACCS \wehavethat \client \equiv \client'$ and $\musti{\serverA}{\client}$ then $\musti{\serverA}{\client'}$.
\end{lemma}

\begin{lemma}
  \label{lem:must-sc-server}
  $\Forevery \serverA, \serverB, \client  \in \ACCS \wehavethat \serverA \equiv \serverB$ and $\musti{\serverA}{\client}$ then $\musti{\serverB}{\client}$.
\end{lemma}

\newcommand{\serverC}{q'}




A typical technique to reason on the LTS of concurrent processes, and
so also of \svrclt systems, is trace zipping:
if $p \wt{s} p'$ and $q \wt{ \co{s }} q'$, an induction on $s$ ensures
that $p \Par q \wt{ } p' \Par q'$.
Zipping together different LTS is slightly more delicate: we can zip
weak transitions $ \wta{ s } $ together with the co-transitions $ \wt{
  \co{s} } $, but possibly moving inside equivalence classes of $\equiv$
instead of performing actual transitions in $\st{}$.

\begin{lemma}[Zipping]
  \label{lem:zipping}
  For every $p,q \in \ACCS$
  \begin{enumerate}
  \item 
    \label{pt:zipping-strong}
    $\forevery \mu \in \Act \wehavethat$
    if $p \sta{ \mu } p'$ and $q \st{\co{ \mu }} q'$ then
    $p \Par q \st{\tau}  p' \Par q'$ or $p \Par q \equiv p' \Par q'$;

  \item 
    \label{pt:zipping-weak}
    $\forevery s \in \Actfin \wehavethat$
    if $p \wta{s} p'$ and $q \wt{\co{s}} q'$ then
    $ p \Par q \wt{\varepsilon} \cdot \equiv p' \Par q'$.
  \end{enumerate}
\end{lemma}



Obviously, for every $\serverA,\serverB \in \States$ and output $\aa \in \Names$ we have
\begin{align}
\label{eq:sta-tau-iff-st-tau}
  \serverA \sta{\tau} \serverB \text{ if and only if } \serverA \st{\tau}   \serverB\\
\label{eq:wta-epsilon-iff-wt-epsilon}
  \serverA \wta{ \varepsilon }  \serverB \text{ if and only if } \serverA \wt{ \varepsilon }  \serverB\\
\label{eq:output-sta-st}
 \serverA \sta{\co{\aa}}  \serverB\text{ if and only if } \serverA \st{\co{\aa}}  \serverB
\end{align}
\noindent
together with the expected properties of finiteness,
the first one amounting to the finiteness of the global mailbox in any state:
\begin{align}
{\cardinality{\setof{ \co{a} \in \co{\Names} }{ p \sta{\co{a}} }}} \in \N
\\
\Forevery \mu \in \Act \wehavethat {\cardinality{\setof{ q }{ p \sta{\mu} q }}} \in \N
\\
  {\cardinality{ 
      \setof{ q \in \States}{p \sta{\tau} q } }} \in \N
\end{align}


\subsection{Client generators and their properties}
\label{sec:client-generators}

This subsection is devoted to the study of the semantic properties
of the clients produced by the function $g$.
In general these are the properties sufficient to obtain our
completeness result.

\begin{figure}[t]
  \hrulefill
  $$
  \begin{array}{lll}
    g(\varepsilon, \client) & = & \client  \\
    g(\aa.\trace, \client) & = &  \mailbox{\co{\aa}} \Par g( \trace , \client)  \\
    g(\co{\aa}.\trace, \client) & = & \aa . g(\trace , \client) \extc \tau.\Unit
  \end{array}
  $$
  \begin{align}
    c( \trace ) & = g( \trace ,\tau.\Unit) \label{def:cs} \\
    \tacc{ \trace }{L} & = g(\trace ,h(L)) \ \mathit{where} \ h(L) = \Pi\setof{ \mu.\Unit }{\mu \in L }  \label{def:tsL}
  \end{align}
  \caption{Functions to generate clients.}
  \label{fig:test-generators}
  \label{fig:client-generators}
  \hrulefill
\end{figure}

\begin{lemma}
  \label{lem:gen-reduces-if}
  \label{lem:gs-reduces}
  For every $\server \in \ACCS$ and $\trace \in \Actfin$, if $\server \st{\tau}$ then $g( \server , \trace ) \st{\tau}$.
\end{lemma}
\begin{proof}
By induction on the sequence $\trace$.
In the base case $s = \varepsilon$. The test generated by~$g$ is~$\server$,
which reduces by hypothesis, and so does $g(\varepsilon,  \server )$.
In the inductive case $\trace = \alpha.\traceB$, and we proceed by case-analysis on $\alpha$.
If~$\alpha$ is an output then $g(\alpha.\traceB,  \server) = \co{\alpha}.(g(\traceB,  \server)) \extc \tau.\Unit$
which reduces to $\Unit$ using the transition rule \extR.
If $\mu$ is an input then $g(\alpha.\traceB,  \server) = \mailbox{ \co{\alpha} } \Par g(\traceB,  \server)$
which reduces using the transition rule
\parR and the inductive hypothesis, which ensures that
$g(\traceB,  \server)$ reduces.
\end{proof}

\begin{lemma}
  \label{lem:gen-test-unhappy-if}
  For every $p$ and $s$, if $\lnot \good{p}$ then for every $s$,  $\lnot \good{g(s, p)}$.
\end{lemma}
\begin{proof}
The argument is essentially the same of \rlem{gen-reduces-if}
\leaveout{%
The proof is by induction on the sequence $s$.
In the base case $s = \varepsilon$. The generated test is $p$ and the hypothesis
tells us that $\lnot \good{p}$.
In the inductive case $s = \mu.s'$ and our inductive hypothesis tells us
that $\lnot \good{ g(s', p)}$.
We proceed by case-analysis on $\mu$.
If $\mu$ is an input then $g(\mu.s', p) = \mailbox{ \co{\mu} } \Par g(s', p)$ which reduces using the transition rule \parR and the inductive hypothesis tells us that $g(s', p)$ is unhappy,
as is $\co{\mu} \Par g(s', p)$ since $\lnot\good{\co{\mu}}$.
If $\mu$ is an output then $g(\mu.s', p) = \co{\mu}.(g(s', p)) \extc \tau.\Unit$ which
is unhappy.}
\end{proof}

\begin{lemma}
  \label{lem:gs-visible-action}
  For every $s \in \Actfin$,
  if $\gen{s}{q} \st{ \mu } o$
  then either
  \begin{enumerate}[(a)]
  \item $q \st{\mu} q'$, $s \in \Names^\star$, and $o = \mailbox{ \Pi \co{s}} \Par q'$, or
  \item
    $s = \traceA. \out{\mu}. \traceB$ for some $\traceA \in \Names^{\star}$ and $\traceB \in \Actfin$,
    and $o \equiv \mailbox{ \Pi \co{\traceA} } \Par \gen{\traceB}{q}$,
    and
    \begin{enumerate}[(i)]
    \item $\mu \in \Names$ implies $\gen{s}{q} \equiv \mailbox{ \Pi \co{\traceA} } \Par (\tau.\Unit \extc \mu.\gen{\traceB}{q})$,
    \item $\mu \in \out{\Names}$ implies $\gen{s}{q} \equiv \mailbox{ \Pi \co{\traceA} \Par \mu } \Par (\tau.\Unit \extc \mu.\gen{\traceB}{q})$.
    \end{enumerate}
  \end{enumerate}
\end{lemma}
\begin{proof}
  The proof is by induction on $s$.

  In the case, $s = \varepsilon$, and hence
  by definition $\gen{s}{q} = q$. The
  hypotheses $\gen{s}{q} \st{ \mu } o$ implies
  $q \st{\mu} o$, and $o \equiv \Nil \Par o \equiv \Pi{ \varepsilon } \Par o$.

  In the inductive case, $s = \nu . s'$.
  We have two cases, depending on whether $\nu$ is an output action
  or an input action.

  Suppose $\nu$ is an output. In this case
  $\gen{s}{q} = \tau.\Unit \extc \co{\nu}.\gen{s'}{q}$.
  The hypothesis $\gen{s}{q} \st{ \mu } o$ ensures that $\out{\nu} = \mu$,
  thus $\mu$ is an input action.
  By letting $\traceA = \varepsilon$ and $\traceB = s'$ we obtain the required
  $$
  \gen{s}{q} = \tau.\Unit \extc \mu.\gen{\traceB}{q} \equiv \Pi \co{\traceA} \Par (\tau.\Unit \extc \mu.\gen{\traceB}{q})
  $$
  and $ o \equiv \Pi \co{\traceA} \Par \gen{\traceB}{q}$.

  Now suppose that $\nu$ is an input action.
  By definition
  \begin{equation}
    \label{eq:gs-action-shape-gensq}
    \gen{s}{q} = \out{\nu} \Par \gen{s'}{q}
  \end{equation}
  and the inductive hypothesis ensures that either
  \begin{enumerate}[(1)]
  \item\label{pt:gs-input-1} $q \st{\mu} q'$, $s' \in \Names^\star$, and $o' = \Pi \co{s'} \Par q'$, or
  \item\label{pt:gs-input-2} $s' =  s'_1 .\out{\mu}. s'_2$, for some $s'_1 \in \Names^{\star} $ and $\traceB \in \Actfin$, and
    \begin{equation}
      \label{eq:gs-input-shape-o'}
      o' \equiv  \Pi \out{s'_1} \Par \gen{s'_2}{q}
    \end{equation}
    and
\begin{align}
      \label{eq:gs-input-shape-gs'}
      \mu \in \Names &\text{ implies }\gen{s'}{q}  \equiv  \Pi \out{s'_1} \Par (\tau.\Unit \extc \mu.\gen{s'_2}{q})\\
      \label{eq:gs-input-shape-gs'-out}
      \mu \in \co{\Names} &\text{ implies } \gen{s'}{q}  \equiv  \Pi \out{s'_1} \Par \mu \Par (\tau.\Unit \extc \mu.\gen{s'_2}{q})
\end{align}
  \end{enumerate}
  The action $\mu$ is either an input or an  output, and we organise the proof accorrdingly.

  Suppose $\mu$ is an input. Since $ \out{\nu}$ is an output, the transition
  $\gen{s}{q} \st{ \mu } o$ must be due to a transition $\gen{s'}{q} \st{\mu} o'$,
  thus \req{gs-action-shape-gensq} implies
  \begin{equation}
    \label{eq:gs-input-shape-o}
    o = \out{\nu} \Par o'
  \end{equation}

  In case \ref{pt:gs-input-1},
 then $s' \in \Names^\star$ and $\nu \in \Names$
  ensure $s \in \Names^\star$ and the equality
  $o \equiv \Pi \co{s} \Par q'$ follows from $ o' = \Pi \co{s'} \Par q' $
  and \req{gs-input-shape-o}.

  In case \ref{pt:gs-input-2}, let $\traceA = \nu.s'_1 $, $\traceB = s'_2$.
  Since $\nu$ is an input we have $ \traceA \in \Names^{\star}$.
  The equalities $s = \nu.s'$  and $ s' =  s'_1 . \out{\mu} . s'_2$
  imply that $s = \traceA \out{\mu} \traceB$.
  The required $ o \equiv \Pi \out{\traceA} \Par g(s'_2,q)$
  follows from $o' \equiv \Pi \out{s'_1} \Par g(s'_2,q)$ and
  \req{gs-input-shape-o}.

  Now we proceed as follows,
  $$
  \begin{array}{lllr}
    \gen{s}{q} & = & \out{\nu} \Par \gen{s'}{q} &\text{By }\req{gs-action-shape-gensq} \\
    & \equiv & \out{\nu} \Par (\Pi \out{s'_1} \Par (\tau.\Unit \extc \mu.g(s'_2,q)))&\text{By } \req{gs-input-shape-gs'} \\
    & \equiv &
    (\out{\nu} \Par \Pi \out{s'_1}) \Par (\tau.\Unit \extc \mu.g(s'_2,q))
    & \text{Associativity}\\
    & \equiv & \Pi \out{\traceA} \Par (\tau.\Unit \extc \mu.g(s'_2,q)) & \text{Because } \traceA = \nu.s'_1 \\
    & \equiv & \Pi \out{\traceA} \Par (\tau.\Unit \extc \mu.q(\traceB,q)) & \text{Because } \traceB = s'_2 \\
  \end{array}
  $$

  Now suppose that $\mu$ is an output.
  Then either $\co{\nu} = \mu$ or $\co{\nu} \neq \mu$.

  In the first case we let $\traceA = \varepsilon$ and $\traceB = s'$.
  \req{gs-action-shape-gensq} and $ \co{\nu} = \mu $ imply
  $\gen{s}{q} = \mu \Par \gen{\traceB}{q}$ from which we obtain the required
  $\gen{s}{q} \equiv \Pi \out{\traceA} \Par \mu \Par \gen{\traceB}{q} $,
  and $o \equiv  \Pi \out{\traceA} \Par \gen{\traceB}{q}$.

  If $\co{\nu} \neq \mu$ then \req{gs-action-shape-gensq}
  ensures that  the transition $\gen{s}{q} \st{ \mu } o $
  must be due to $\gen{s'}{q} \st{ \mu } o'$ and \req{gs-input-shape-o} holds.
  We use the inductive hypothesis.

  If \ref{pt:gs-input-1} is true, we proceed as already discussed.
  In case \ref{pt:gs-input-2} holds, let $\traceA =  \nu. s'_1$, we have that $\traceA \in \Names^{\star}$.
  Let $\traceB = s'_2$, now we have that
  $$
  \begin{array}{lllr}
    \gen{s}{q} & = & \out{\nu} \Par \gen{s'}{q} &\text{By }\req{gs-action-shape-gensq} \\
    & \equiv & \out{\nu} \Par (\Pi \out{s'_1} \Par \mu \Par g(s'_2,q)) & \text{By }\req{gs-input-shape-gs'-out} \\
    & \equiv & (\out{\nu} \Par \Pi \out{s'_1}) \Par \mu \Par g(s'_2,q) &\text{Associativity}\\
    & = &  \Pi \out{ \traceA } \Par \mu \Par g(s'_2,q) & \text{Because }\traceA =  \nu. s'_1\\
    & = &  \Pi \out{ \traceA } \Par \mu \Par \gen{\traceB}{q} & \text{Because }\traceB =  s'_2\\
  \end{array}
$$
\end{proof}

\begin{lemma}
  \label{lem:inversion-gen-accs-mu}
  For every $s \in \Actfin$,
  if $\gen{s}{q} \st{ \mu } p$
  then either:
  \begin{enumerate}[(a)]
  \item\label{pt:inversion-gen-accs-mu-left}
    there exists $q'$ such that $q \st{\mu} q'$, $s \in \Names^\star$ with $p \equiv \Pi \co{s} \Par q'$, or
  \item\label{pt:inversion-gen-accs-mu-right}
    $s = \traceA. \out{\mu}. \traceB$ for some $\traceA \in \Names^{\star}$ and $\traceB \in \Actfin$
    with $p \equiv \gen{\traceA.\traceB}{q}$.
  \end{enumerate}
\end{lemma}
\begin{proof}
The proof is by induction over the sequence $s$.

In the base case $s = \varepsilon$ and we have $\gen{\varepsilon}{q} = q$.
We show \rpt{inversion-gen-accs-mu-left} and choose $q' = p$.
We have $\gen{\varepsilon}{q} = q \st{\mu} p$,
$p \equiv \Pi \co{\varepsilon} \Par q'$ and $\varepsilon \in \Names^\star$
as required.

In the inductive case $s = \nu.s'$.
We proceed by case-analysis on $\nu$.
If $\nu$ is an input, then $g (\nu.s', q) = \co{\nu} \Par g (s', q)$.
The hypothesis $g (\nu.s', q) \st{\mu} p$ implies that either:
\begin{enumerate}[(i)]
\item $\co{\nu} \st{\mu} \Nil$ with $p = \Nil \Par g (s', q)$ and $\co{\nu} = \mu$, or
\item $g (s', q) \st{\mu} \hat{p}$ with $p = \co{\nu} \Par \hat{p}$.
\end{enumerate}

In the first case we show \rpt{inversion-gen-accs-mu-right}.
We choose $\traceA = \varepsilon$, $\traceB = s'$.
We have $s = \mu.s' = \co{\nu}.s'$, $p = \Nil \Par g (s', q) \equiv g (\varepsilon.s', q)$
and $\varepsilon \in \Names^\star$ as required.

In the second case the inductive the hypothesis tells us that either:
\begin{enumerate}[(H-a)]
\item\label{pt:inversion-gen-accs-mu-left-IH}
  there exists $q'$ such that $q \st{\mu} q'$, $s' \in \Names^\star$
  with $\hat{p} \equiv \Pi \co{s'} \Par q'$, or
\item\label{pt:inversion-gen-accs-mu-right-IH}
  $s = \traceA. \out{\mu}. \traceB$ for some $\traceA \in \Names^{\star}$ and $\traceB \in \Actfin$
  with $\hat{p} \equiv \gen{\traceA.\traceB}{q}$.
\end{enumerate}
\rpt{inversion-gen-accs-mu-left-IH} or \rpt{inversion-gen-accs-mu-right-IH} is true.

If \rpt{inversion-gen-accs-mu-left-IH} is true then
$s' \in \Names^\star$ and there exists $q''$ such that $q \st{\mu}
q''$ with $\hat{p} \equiv \Pi \co{s'} \Par q''$.
We prove \rpt{inversion-gen-accs-mu-left}. We choose $q' = q''$ and $s = \nu.s'$.
We have $p \equiv \co{\nu} \Par \hat{p} \equiv \co{\nu} \Par \Pi
\co{s'} \Par q'' \equiv \Pi \co{\nu.s'} \Par q''$ and $\nu.s' \in
\Names^\star$ as required.

If \rpt{inversion-gen-accs-mu-right-IH} is true then
$s' = s'_1. \out{\mu}. s'_2$ for some $s'_1 \in \Names^{\star}$ and $s'_2 \in \Actfin$
with $\hat{p} \equiv \gen{s'_1.s'_2}{q}$.
We prove \rpt{inversion-gen-accs-mu-left}.
We choose $\traceA = \nu.s'_1$ and $\traceB = s'_2$.
We have $p \equiv \co{\nu} \Par \hat{p} \equiv \co{\nu} \Par
\gen{s'_1.s'_2}{q} \equiv \gen{\nu.s'_1.s'_2}{q}$ as required.

If $\nu$ is an output, then $g (\nu.s', q) = \co{\nu}.(g (s', q)) \extc \tau.\Unit$.
We prove \rpt{inversion-gen-accs-mu-right} and choose $\traceA = \varepsilon$, $\traceB = s'$.
The hypothesis $g (\nu.s', q) \st{\mu} p$ implies that $\mu = \co{\nu}$ and
$p \equiv g (s', q) \equiv g (\varepsilon.s', q)$ as required.
\end{proof}

\begin{lemma}
  \label{lem:gs-tau}
  For every $s \in \Actfin$,
  if $\gen{s}{q} \st{ \tau } o$
  then either
  \begin{enumerate}[(a)]
  \item\label{pt:gs-tau-1} $ \good{o}$, or
  \item\label{pt:gs-tau-2} $s \in \Names^\star$, $q \st{\tau} q'$, and $o = \Pi \co{s} \Par q'$, or
  \item\label{pt:gs-tau-3} $s \in \Names^\star$, $q \st{\nu} q'$, and $s =  \traceA . \nu . \traceB$, and $o \equiv \Pi \co{\traceA.\traceB} \Par q'$ , or
  \item\label{pt:gs-tau-4} $o \equiv \gen{\traceA.\traceB.\traceC}{q}$ where $ s = \traceA . \mu . \traceB . \co{\mu} . \traceC$ with $ \traceA . \mu  . \traceB \in \Names^\star$.
  \end{enumerate}
\end{lemma}
\begin{proof}
  By induction on the structure of $s$.

  In the base case $s = \varepsilon$. We prove \ref{pt:gs-tau-2}.
  Trivially $s \in \Names^\star$, and by definition $\gen{\varepsilon}{q} = q$,
  the hypothesis implies therefore that $q \st{\tau} o$. The $q'$ we are after is $o$ itself,
  for $o \equiv \Nil \Par o =\Pi \co{s} \Par o$.

  In the  inductive case $s = \nu.s'$. We proceed by case analysis on
  whether $\nu \in \co{\Names}$ or $\nu \in \Names$.

  If $\nu$ is an output, by definition $ \gen{s}{q} = \tau.\Unit \extc \co{\nu}.gen{s'}{q}$.
  Since $\co{\nu}.\gen{s'}{q} \stable$, the silent move $ \gen{s}{q} \st{\tau} o $
  is due to rule \extL, thus $o = \Unit \extc\co{\nu}.\gen{s'}{q}$, and thus $\good{o}$.
  We have proven \ref{pt:gs-tau-1}.

  Suppose now that $\nu$ is an input, by definition
  \begin{equation}
    \label{eq:gs-tqu-shape-test}
    \gen{s}{q} = \co{\nu} \Par \gen{s'}{q}
  \end{equation}
  The silent move  $ \gen{s}{q} \st{\tau} o $ must have been derived via the rule \com,
  or the rule \parR.

  If \com was employed we know that
  $$
  \begin{prooftree}
    \co{\nu} \st{ \out{\nu} } \Nil
    \quad
    \begin{prooftree}
      \vdots
      \justifies
      \gen{s'}{q} \st{ \nu } o'
    \end{prooftree}
    \justifies
    \co{\nu} \Par \gen{s'}{q} \st{\tau} \Nil \Par o'
  \end{prooftree}
  $$
  and thus $o \equiv o'$.
  Since $\nu$ is an input and $\gen{s'}{q} \st{ \nu } o'$,
  \rlem{gs-visible-action} ensures that either
  \begin{enumerate}[(1)]
  \item\label{pt:gs-input-aux-1} $q \st{\nu} q'$, $s' \in \Names^\star$, and $o' = \Pi \co{s'} \Par q'$, or
  \item\label{pt:gs-input-aux-2}
    $s' = s'_1. \out{\nu}. s'_2$ for some $s'_1 \in \Names^{\star}$ and $s'_2 \in \Actfin$,
    and
    \begin{align}
      o' & \equiv \mailbox{ \Pi \co{s'_1} } \Par \gen{s'_2}{q}\\
      \gen{s}{q} & \equiv \mailbox{ \Pi \co{s'_1}}  \Par (\tau.\Unit \extc \nu.\gen{s'_2}{q})
    \end{align}
  \end{enumerate}

  In case \ref{pt:gs-input-aux-1} we prove \rpt{gs-tau-3}.
  Since $\nu$ is an input, $s' \in \Names^\star$ ensures that $s \in \Names^\star$.
  By letting $\traceA = \varepsilon$ and $\traceB = s'$ we obtain $s = \traceA.\nu.\traceB$.
  We have to explain why $o \equiv \mailbox{ \Pi \co{\traceA.\traceB} } \Par q'$.
  This follows from the definitions of $\traceA$ and $\traceB$, from
  $o \equiv \Nil \Par o' $ and from $o' \equiv  \mailbox{ \Pi \co{s'}} \Par q'$.

  In case \ref{pt:gs-input-aux-2} we prove \rpt{gs-tau-4}.
  Let $\traceA = \varepsilon$, $\traceB =  s'_1$ and $\traceC =  s'_2$.
  $$
  \begin{array}{lllr}
    s' & = & \nu s'_1 . \co{\nu} . s'_2  & \text{By inductive hypothesis}\\
    \nu.s' & = & \nu .s'_1 . \co{\nu} . s'_2  \\
    s & = & \nu .s'_1 . \co{\nu} . s'_2  & \text{Because } s = \nu.s'\\
    s & = & \traceA . \nu .\traceB . \co{\nu} . \traceC  & \text{By definition}
  \end{array}
  $$
  and $ \traceA . \nu .\traceB \in \Names^\star$ as required.
  Moreover $o' = \mailbox{ \Pi \co{s'_1} } \Par \gen{s'_2}{q} =  \mailbox{ \Pi \co{\traceB} } \Par \gen{\traceC}{q} = \gen{\traceB.\traceC}{q} = \gen{\traceA.\traceB.\traceC}{q}$
  as required. This concludes the argument due to an applicaton of \com.

  If \parR was employed we know that
  $$
  \begin{prooftree}
    \begin{prooftree}
      \vdots
      \justifies
       \gen{s'}{q} \st{ \tau } o'
    \end{prooftree}
    \justifies
    \co{\nu} \Par \gen{s'}{q} \st{\tau} \co{\nu} \Par o'
  \end{prooftree}
  $$
  thus $     \gen{s'}{q} \st{ \tau } o'$ and
  \begin{equation}
    \label{eq:gs-tau-shape-o}
    o = \mailbox{ \co{\nu} } \Par o'
  \end{equation}

  Since $s'$ is smaller than $s$, thanks to $\gen{s'}{q} \st{ \tau } o'$ we apply the inductive hypothesis
  to obtain either
  \begin{enumerate}[(i)]
  \item\label{pt:ih-gs-tau-1} $ \good{o'}$, or
  \item\label{pt:ih-gs-tau-2} $s' \in \Names^\star$, $q \st{\tau} q'$, and $o' = \mailbox{ \Pi \co{s'} } \Par q'$, or
  \item\label{pt:ih-gs-tau-3} $s' \in \Names^\star$, $q \st{\mu} q'$, and $s' =  s'_1 . \mu . s'_2$, or
  \item\label{pt:ih-gs-tau-4} $o' \equiv \gen{s'_1.s'_2.s'_3}{q}$ where $ s' = s'_1 . \mu . s'_2 . \co{\mu} . s'_3$ with $ s'_1 . \mu  . \traceB' \in \Names^\star$,
  \end{enumerate}

  If \ref{pt:ih-gs-tau-1} then \req{gs-tau-shape-o} implies \ref{pt:gs-tau-1}.
  If \ref{pt:ih-gs-tau-2} then $s = \nu.s' $ and the assumption that $\nu$ is input
  imply that $s \in \Names^\star$. \req{gs-tau-shape-o} and $o' = \mailbox{ \Pi \co{s'}} \Par q'$
  imply that $o = \Pi \co{s} \Par q'$. We have proven \ref{pt:gs-tau-2}.
  If \ref{pt:ih-gs-tau-2} we prove \ref{pt:gs-tau-2}, because
  $ s' \in \Names^\star$ ensures $s \in \Names^\star$ and $s' = s'_1 . \mu . s'_2$
  let use prove $s = \traceA . \nu . \traceB$ by letting $\traceA = \nu.s'_1$ and $\traceB = s'_2$.

  If \ref{pt:ih-gs-tau-4} we prove \ref{pt:gs-tau-4}.
  We have $o' \equiv \gen{s'_1.s'_2.s'_3}{q}$
  and $s' = s'_1 . \lambda . s'_2 . \out{\lambda} .s'_3$
  with $s'_1 . \lambda . s'_2 \in \Names^{\star}$.

  Let $\traceA = \nu . s'_1 $, $\traceB =  s'_2$ and $\traceC = s'_3$.
  Since  $ s'_1 . \mu  . \traceB' \in \Names^\star$, we have $\traceA . \mu  . \traceB \in \Names^\star$.
  We also have
  $$
  \begin{array}{lllr}
    s' & = & s'_1 . \mu . s'_2 . \co{\mu} . s'_3 & \text{By inductive hypothesis}\\
    \nu.s' & = & \nu.s'_1 . \mu . s'_2 . \co{\mu} . s'_3 \\
    s & = & \nu.s'_1 . \mu . s'_2 . \co{\mu} . s'_3 & \text{Because } s = \nu.s'\\
    s & = & \traceA . \mu . \traceB . \co{\mu} . \traceC & \text{By definition }
  \end{array}
  $$
  It remains to prove that $o \equiv \gen{\traceA.\traceB.\traceC}{q}$.
  This is a consequence of \req{gs-tau-shape-o}, of
  $o' \equiv \gen{s'_1.s'_2.s'_3}{q}$, and of the definitions of $\traceA, \traceB,$ and $\traceC$.
\end{proof}


\begin{lemma}
\label{lem:stable-derivative-gs-happy}
For every $s \in \Actfin$, and process $q$
such that or $q \st{\tau} q'$ implies $\good{q'}$,
and for every $\mu \in \Names. q \st{\mu} q'$ implies $\good{q'}$,
if $\gen{s}{q} = o_0 \st{\tau} o_1 \st{\tau} o_2 \st{\tau} \ldots o_n \stable$ and $n > 0$ then $\good{o_i}$ for some $i \in [1,n]$.
\end{lemma}
\begin{proof}
  \rlem{gs-tau} implies that one of the following is true,
  \begin{enumerate}[(a)]
  \item\label{pt:gstau-1} $\good{o_1}$, or
  \item\label{pt:gstau-2} $s \in \Names^\star$, $q \st{\tau} q'$, and $o_1 = \mailbox{ \Pi \co{s} } \Par q'$, or
  \item\label{pt:gstau-3} $s \in \Names^\star$, $q \st{\mu} q'$, and $s =  \traceA . \mu . \traceB$, and $o_1 \equiv \Pi \mailbox{ \co{\traceA.\traceB} } \Par q'$ , or
  \item\label{pt:gstau-4} $o_1 \equiv \gen{\traceA.\traceB.\traceC}{q}$ where $ s = \traceA . \mu . \traceB . \co{\mu} . \traceC$ with $ \traceA . \mu  . \traceB \in \Names$.
  \end{enumerate}
  If \ref{pt:gstau-1} we are done.
  If \ref{pt:gstau-2} or \ref{pt:gstau-3} then $\good{q'}$, and thus $\good{o_1}$.

  In the base case $\len{s} = 0$, thus \ref{pt:gstau-4} is false. It follows that $\good{o_1}$.

  In the inductive case $s = \nu . s'$.
  We have to discuss only the case in which \ref{pt:gstau-4} is true.
  The inductive hypothesis ensures that
  \begin{center}
    For every $s' \in \bigcup_{i = 0}^{\len{s}-1}$, if $\gen{s'}{q} \st{\tau} o'_1 \st{\tau} o'_2 \st{\tau} \ldots o'_m \stable$ and $m > 0$
    then $\good{o'_j}$ for some $j$.
  \end{center}
  Note that $ o_1 \st{\tau} $ so the reduction sequence $o_1 \wt{} o_n$ cannot be empty, thus $ m > 0$. This and  $\len{\traceA\traceB\traceC} < \len{s}$
  let us apply the inductive hypothesis to state that
  $$
  \gen{\traceA.\traceB.\traceC}{q} \st{\tau} o'_1 \st{\tau} o'_2 \st{\tau} \ldots o'_m \stable
  \text{ implies } o'_j \text{ for some } j.
  $$
  We conclude the argument via \rlem{st-compatible-with-equiv} and because $\equiv$ preserves success.
\end{proof}

\begin{lemma}
  \label{lem:gs-stable}
  For every $s \in \Actfin$ and process $q$, if $\gen{s}{q} \stable$ then
  \begin{enumerate}
  \item $ s \in \Names^\star $,
  \item $ q \stable $,
  \item $I(q) \cap \co{s} = \emptyset $,
  \item $R(\gen{s}{q}) = \co{s} \cup R(q)$.
  \end{enumerate}
\end{lemma}
\begin{proof}
  By induction on $s$.
  In the base case $\varepsilon \in \Names^\star$,
  and $\gen{ \varepsilon }{q} = q$,
  thus $q \stable$. The last two points follow from
  this equality and from $\varepsilon$ containing no actions.

  In the inductive case $s = \mu . s'$.
  The hypothesis $\gen{\mu . s'}{q} \stable$ and the definition
  of $g$ imply that $\gen{\mu . s'}{q} = \co{\mu} \Par \gen{s'}{q}$,
  thus $\mu \in \Names$. The inductive hypothesis ensures that
  \begin{enumerate}
  \item $ s' \in \Names^\star $,
  \item $ q \stable $,
  \item for every $I(q) \cap \co{s'} = \emptyset $,
  \item for every $R(\gen{s'}{q}) = \co{s'} \cup R(q)$
  \end{enumerate}
  Since $ \co{\mu} \Par \gen{s'}{q} \stable $ rule \com
  cannot be applied, thus $ q \Nst{\mu} $, and so
  $I(q) \cap \co{s} = \emptyset $.
  From $R(\gen{s'}{q}) = \co{s'} \cup R(q)$ we obtain
  $R(\gen{s}{q}) = \co{s} \cup R(q)$.
\end{proof}

\begin{lemma}
  \label{lem:gen-test-lts-co-mu}
  For every $\mu \in \Act$, $s$ and $p$, $g(\mu.s, p) \st{\co{\mu}} g(s,p)$.
\end{lemma}
\begin{proof}
We proceed by case-analysis on $\mu$.
If $\mu$ is an input then $g(\mu.s, p) = \co{\mu} \Par g(s, p)$.
We have $\mailbox{ \co{\mu} } \Par g(s, p) \st{\co{\mu}} \Nil \Par g(s, p) \equiv g(s, p)$
as required.
If $\mu$ is an output then $g(\mu.s, p) = \co{\mu}.g(s, p) \extc \tau.\Unit$.
We have $\co{\mu}.g(s, p) \extc \tau.\Unit \st{\co{\mu}} g(s, p)$ as required.
\end{proof}

\begin{lemma}
  \label{lem:inversion-feeder-tau-accs}
  $\Forevery \trace \in \Actfin, \and q \in \ACCS \suchthat
  c(s) \sta{\tau} q$ either
  \begin{enumerate}[(a)]
  \item\label{inversion-feeder-tau-accs-ok} $\good{q}$, or
  \item\label{inversion-feeder-tau-accs-split}
    there exist $\ab$, $\traceA$, $\traceB$ and $\traceC$ with
    $\traceA.\ab.\traceB \in \Names^\star$ such that
    $s = \traceA.\ab.\traceB.\co{\ab}.\traceC$ and
    $q \equiv c(\traceA.\traceB.\traceC)$.
    \end{enumerate}
\end{lemma}
\begin{proof}
The proof is by induction on $s$.

In the base case $s = \varepsilon$, $c(\varepsilon) = \tau.\Unit$ and then
$q = \Unit$. We prove \ref{inversion-feeder-tau-accs-ok} with $\good{\Unit}$.

In the inductive case $s = \mu.s'$.
We proceed by case-analysis over $\mu$.

If $\mu$ is an input then $c(\mu.s') = \co{\mu} \Par c(s')$.
We continue by case-analysis over the reduction $\co{\mu} \Par c(s') \st{\tau} q$.
It is either due to:
\begin{enumerate}[(i)]
\item\label{inversion-feeder-tau-accs-input-com}
  a communication between $\co{\mu}$ and $c(s')$ such that
  $\co{\mu} \st{\co{\mu}} \Nil$ and $c(s') \st{\mu} q'$ with $q = \Nil \Par q'$, or
\item\label{inversion-feeder-tau-accs-input-tau}
  a reduction of $c(s')$ such that $c(s') \st{\tau} q'$ with $q = \co{\mu} \Par q'$.
\end{enumerate}
If \ref{inversion-feeder-tau-accs-input-com} is true then
\rlem{inversion-gen-accs-mu} tells us that
there exist $s'_1$ and $s'_2$ such that
$s' = s'_1.\co{\mu}.s'_2$ and $q' \equiv c(s'_1.s'_2)$ with $s'_1 \in \Names^*$.
We prove (\ref{inversion-feeder-tau-accs-split}).
We choose $\ab = \mu$, $\traceA = \varepsilon$, $\traceB = s'_1$, $\traceC = s'_2$.
We show the first requirement by
$s = \mu.s' = \mu.s'_1.\co{\mu}.s'_2 = \varepsilon.\mu.s'_1.\co{\mu}.s'_2 = \traceA.\ab.\traceB.\co{\ab}.\traceC$.
The second requirement is $q = \Nil \Par q' \equiv c(s'_1.s'_2) = c(\varepsilon.s'_1.s'_2) = c(\traceA.\traceB.\traceC)$.

We now consider the case (\ref{inversion-feeder-tau-accs-input-tau}).
The inductive hypothesis tells us that either:
\begin{enumerate}
\item\label{inversion-feeder-tau-accs-ok-IH} $\good{q'}$, or
\item\label{inversion-feeder-tau-accs-split-IH} there exist $\iota$, $s'_1$, $s'_2$ and $s'_3$ with
  $s'_1.\iota.s'_2 \in \Names^*$ such that
  $s' = s'_1.\iota.s'_2.\co{\iota}.s'_3$ and
  $q' \equiv c(s'_1.s'_2.s'_3)$.
\end{enumerate}
If (\ref{inversion-feeder-tau-accs-ok-IH}) is true then
we prove \ref{inversion-feeder-tau-accs-ok} with
$q = \mailbox{ \co{\mu} } \Par q'$ and $\mailbox{ \co{\mu} }\Par q'$ and $ \good{q'}$.
If (\ref{inversion-feeder-tau-accs-split-IH}) is true then
we prove (\ref{inversion-feeder-tau-accs-split}).
We choose  $\ab = \iota$, $\traceA = \mu.s'_1$, $\traceB = s'_2$, $\traceC = s'_3$.
We show the first requirement with
$s = \mu.s' = \mu.s'_1.\iota.s'_2.\co{\iota}.s'_3 = \traceA.\ab.\traceB.\co{\ab}.\traceC$.
The second requirement is $q = \mailbox{ \co{\mu} } \Par q' \equiv \mailbox{ \co{\mu} } \Par c(s'_1.s'_2.s'_3) = c(\mu.s'_1.s'_2.s'_3) = c(\traceA.\traceB.\traceC)$.

If $\mu$ is an output then $c(\mu.s') = \co{\mu}.(c s') \extc \tau.\Unit$.
The hypothesis $c(\mu.s') \st{\tau} q$ implies $q = \Unit$.
We prove (\ref{inversion-feeder-tau-accs-ok}) with $\good{\Unit}$.
\end{proof}

\section{Counter-example to existing completeness result}
\label{sec:counterexample}


\newcommand{\llch}{\ensuremath{\ll_{{\sf ch}}}}
\newcommand{\readyset}{\ensuremath{R}}
\renewcommand\acnvalong{\Downarrow_{a}}


In this section we recall the definition of the alternative
preorder $\llch$ by \cite{DBLP:conf/fsttcs/CastellaniH98},
and show that it is not complete with respect to~$\testleqS$,
\ie $\,\testleqS\, \not\subseteq\; \llch$.
  We start with some               
  auxiliary definitions.
\smallskip


The predicate $\rr{\I}$ is defined by the following two rules: 
\begin{itemize}
\item $ p \rr{\I} p$ if $p \stable$ and 
 $I(p) \cap I = \emptyset$,
\item $ p \rr{\I \uplus \mset{a}} p'' $ if $\,p \wt{a} p'$ and $ p' \rr{I} p'$
\end{itemize}
\smallskip

The {\em generalised acceptance set} of a process $p$ after a trace
$\trace$ with respect to a multiset of input actions $\I$ is defined by
$$
\gas{p}{s}{\I} = \setof{ O(p'') }{ p \wta{s} p' \rr{\I} p''}
$$

The set of {\em input
  multisets} of a process $p$ after a trace $s$ is defined by
$$\im{p}{s} = \setof{ \mset{\aa_1, \ldots,
    \aa_n } }{ \aa_i \in \Names, p \wta{s} \wt{\aa_1} \ldots
  \wt{\aa_n}}$$

\smallskip


The convergence predicate over traces performed by forwarders is
  denoted~$\acnvalong$, and defined as~$\cnvalong$, but over the LTS
  given in \rexa{forwarders-in-TACCS}.


The preorder $\llch$ is now defined as follows:

\begin{definition}[Alternative preorder $\llch$~\cite{DBLP:conf/fsttcs/CastellaniH98}]
\label{def:ff-original}
Let $p \llch q$ if for every $s \in \Actfin \wehavethat p \acnvalong s$  implies
\begin{enumerate}
\item $q \acnvalong s$,
\item\label{pt:ff-original-acc-sets}
  for every $\readyset \in \acc{q}{s}$ and every $I \in \im{p}{s}$ such that $I \cap \readyset = \emptyset$ there exists some $O \in \gas{p}{s}{I}$
  such that $O \setminus \co{I} \subseteq \readyset$.\hfill$\blacksquare$
\end{enumerate}
\end{definition}



\smallskip


\smallskip
We illustrate the three auxiliary definitions using the
  process $\pierre= b.(\tau.\Omega \extc c. \co{d} )$
introduced in ~\rexa{p-testleqS-Nil}. We may infer that
\begin{equation}
  \label{eq:pierre-rr}
  \pierre \rr{\mset{b,c}} \mailbox{\co{d}}
\end{equation}
thanks to the following derivation tree
$$
\begin{prooftree}
  \begin{prooftree}
    \begin{prooftree}
      \justifies
      \out{d} \rr{\emptyset}  \out{d}
      \using
      \out{d} \stable \text{ and } I(\out{d}) \cap \emptyset = \emptyset
    \end{prooftree}
    \justifies
    \tau.\Omega \extc c.\co{d} \rr{\mset{c}} \out{d}
    \using
    \tau.\Omega \extc c.\co{d} \wt{c} \out{d}
  \end{prooftree}
  \justifies
  \pierre \rr{\mset{b,c}} \out{d}
  \using
  \pierre \wt{b} \tau.\Omega \extc c.\co{d}
\end{prooftree}
$$

Let us now consider the generalised acceptance set of \pierre after the
trace $\varepsilon$ with respect to the multiset $\mset{b, c}$.
We prove that
\begin{equation}
  \label{eq:pierre-gas}
  \gas{ \pierre }{ \varepsilon }{ \mset{b, c} } = \set{ \set{ \co{d} } }
\end{equation}

By definition 
  $\gas{ \pierre }{ \varepsilon }{ \mset{b, c} } =  \setof{ O(p'') }{
  \pierre \wta{ \varepsilon} p' \rr{\mset{b, c}} p'' }$. Since $
\pierre \stable$, we have
  \begin{equation}
    \label{eq:gas-p}
    \gas{ \pierre }{ \varepsilon }{ \mset{b, c} } =  \setof{ O(p'') }{ \pierre \rr{\mset{b, c}} p'' }
  \end{equation}
  Then, thanks to \req{pierre-rr} 
  we get
  $
  O(\co{d}) = \set{ \co{d} } \in \gas{ \pierre }{ \varepsilon }{ \mset{b, c} }
  $.
  We show now that  $\set{ \co{d} }$ is the only element of this
  acceptance set.
By (\ref{eq:gas-p}) above, it is enough to show that
$\pierre \rr{\mset{b, c}} p''$ implies $p'' = \out{d} $.
Observe
that 
  \begin{enumerate}
  \item $ I( \pierre ) \cap \mset{ b,c } \neq \emptyset $,
  \item $ \pierre \wt{ \aa } p'$ implies $\aa = b$, and
  \item
    There are two different states $p'$ such that
   $ \pierre\wt{b}p'$, but 
    the only one that can do the input $c$ is $p' = \tau.\Omega \extc c.\co{d}$.
  \end{enumerate}
  This implies that the only way to infer
  $\pierre \rr{\mset{b, c}} p''$ is via the derivation tree
  that proves \req{pierre-rr} above. Thus~$p'' = \co{d} $.



\begin{counterexample}
  \label{ex:ffCH98-incomplete}
  The alternative preorder $\llch$
  is not complete for $\testleqS$, namely $p \testleqS q$ does not imply $p \llch q$.
\begin{proof}
  The cornestone 
  of the proof is the process
 $\pierre = b.(\tau.\Omega \extc c. \co{d} )$ discussed above.
In \rexa{p-testleqS-Nil} we have shown that $\pierre \testleqS \Nil$.
Here we show that $\pierre \not\llch \Nil$, because the pair $(\pierre,\Nil)$ does not
satisfy Condition \ref{pt:ff-original-acc-sets} of \rdef{ff-original}.

Since $\pierre \stable$, we know that $\pierre \conv$, and thus by
definition $\pierre \acnvalong \varepsilon$. We also
have by definition $\acc{\Nil}{\varepsilon} = \set{ \emptyset }$, and
$\im{\pierre}{\varepsilon} = \set{\emptyset, \mset{ b}, \mset{ b, c
  }}$.


Let us check Condition \ref{pt:ff-original-acc-sets} of
  \rdef{ff-original} for $p = \pierre$ and $q = \Nil$. Since there is
  a unique $\readyset\in \acc{\Nil}{\varepsilon}$, which is
  $\emptyset$, and $I \cap \emptyset = \emptyset$ for any $I$, we only
  have to check that for every $I \in \im{\pierre}{\varepsilon}$ there
  exists some $O \in \gas{\pierre}{\varepsilon}{I}$ such that
  $O \setminus \co{I} \subseteq \emptyset$.

Let $I = \mset{ b, c }$.
By~\req{pierre-gas} it must be $O =\set{ \co{d} }$. Since
  $ \set{ \co{d} } \setminus \co{I} = \set{ \co{d} } \setminus \co{
    \mset{ b, c } } = \set{ \co{d} } \not\subseteq \emptyset$, the
  condition is not satisfied. Thus $\pierre \not\llch \Nil$.
\end{proof}
\end{counterexample}


  


\section{Highlights of the Coq mechanisation}
\label{sec:coq}

\setminted{fontsize=\footnotesize,linenos=true}

\renewenvironment{mdframed}{}{}

\subsection{Preliminaries}
We begin this section recalling the definition of~$\opMust$, which is given in
\rdef{must-extensional}. It is noteworthy that the mechanised definition, i.e.
\lstinline!must_extensional!, depends on the typeclass \lstinline!Sts!
(\rfig{typeclass-sts}), and {\em not} the type class \lstinline!Lts!. This lays
bare what stated in \rsec{intro}: to define~$\opMust$  a reduction
semantics (i.e. a state transition system), and a predicate $\goodSym$ over
clients suffice.

\subsubsection{State Transition Systems}

The typeclass for state transition systems (Sts) is defined as follows, where
\mintinline{coq}{A} is the set of states of the Sts. It included a notion of stability
which is axiomatized and decidable.

\begin{mdframed}
\begin{minted}{coq}
Class Sts (A: Type) := {
    sts_step: A → A → Prop;
    sts_state_eqdec: EqDecision A;
    sts_step_decidable: RelDecision sts_step;

    sts_stable: A → Prop;
    sts_stable_decidable p : Decision (sts_stable p);
    sts_stable_spec1 p : ¬ sts_stable p -> { q | sts_step p q };
    sts_stable_spec2 p : { q | sts_step p q } → ¬ sts_stable p;
  }.
\end{minted}
\label{fig:typeclass-sts}
\end{mdframed}

\subsubsection{Maximal computations}

A computation is maximal if it is infinite or if its last state is stable.
Given a state \mintinline{coq}{s}, the type \lstinline|max_exec_from s| contains all
the maximal traces that start from \mintinline{coq}{s}. Note the use of a coinductive
type to allow for infinite executions.

\begin{mdframed}
\begin{minted}{coq}
Context `{Sts A}.

CoInductive max_exec_from: A -> Type :=
 | MExStop s (Hstable: sts_stable s) : max_exec_from s
 | MExStep s s' (Hstep: sts_step s s') (η: max_exec_from s') :
    max_exec_from s.
\end{minted}
\end{mdframed}

\leaveout{
\pl{%
\lstinline!max_exec_from! assumes an STS defined on states of type A. It is done with \lstinline!Context `{Sts A}!.
In the definition of \lstinline!must_extensional!, we use \lstinline!max_exec_from! with pairs $(\server,\client)$.
Coq will try to find a way to compose two \lstinline!LTS A L! and \lstinline!LTS B L! to get a \lstinline!LTS (A * B) L!.
It is done using the instance \lstinline!parallel_lts! together with the cast \lstinline!sts_of_lts! from TransitionSystems.v.
Finally, \lstinline!parallel_lts! relies on \lstinline!parallel_steps!.
}
}

\subsection{The must-preorder}
\subsubsection{Client satisfaction}

The predicate $\goodSym$ is defined as any predicate over the states
of an LTS that satisfies certain properties: it is preserved by
structural congruence, by outputs in both directions
(if $p \st{\co{a}} p'$ then $\good{p} \Leftrightarrow \good{p'}$).

It is defined as a typeclass indexed over the type of states and
labels, because we expect a practitioner 
to reason on a single canonical notion of "good" at a time.

\begin{mdframed}
\begin{minted}{coq}

Class Good (A L : Type) `{Lts A L, ! LtsEq A L} := {
    good : A -> Prop;
    good_preserved_by_eq p q : good p -> p ≡ q -> good q;
    good_preserved_by_lts_output p q a :
      p ⟶[ActOut a] q -> good p -> good q;
    good_preserved_by_lts_output_converse p q a :
      p ⟶[ActOut a] q -> good q -> good p
}.

\end{minted}
\end{mdframed}

\subsubsection{Must testing}

\rdef{must-extensional}: We write $\Must{\server}{\client} $ if every maximal
  computation of $\csys{\server}{\client}$ is successful.

Given an integer \mintinline{coq}{n} and a maximal execution $\eta$, the function
\lstinline|mex_take_from n| applied to $\eta$ returns \mintinline{coq}{None} if $\eta$
is shorter than \mintinline{coq}{n} and \lstinline|Some p|, where \mintinline{coq}{p} is a
finite execution corresponding to the first \mintinline{coq}{n} steps of $\eta$.

Then, we define the extensional version of $\Must{p}{e}$ by stating that, for
all maximal executions~$\eta$ starting from $(p, e)$, there exists an integer
$n$ such that the $n$-th element of $\eta$ is good.
The $n$th element is obtained by taking the last element of the finite prefix of
length~$n$ computed using the function above.

\begin{mdframed}
\begin{minted}{coq}
Context `{good : B -> Prop}.

Fixpoint mex_take_from (n: nat) {x} (η: max_exec_from x) :
  option (finexec_from x) :=
 match n with
  | 0 => Some $ FExSingl x
  | S n => match η with
            | MExStop x Hstable => None
            | MExStep x x' Hstep η' =>
                let p' := mex_take_from n η' in
                (λ p', FExStep x x' (bool_decide_pack _ Hstep) p') <$> p'
           end
 end.

Definition must_extensional (p : A) (e : B) : Prop :=
  forall η : max_exec_from (p, e), exists n fex,
    mex_take_from n η = Some fex /\ good (fex_from_last fex).2.
\end{minted}
\end{mdframed}

\subsubsection{The preorder}

Definition~\ref{def:testleq} is mechanised in a straightforward way:

\begin{mdframed}
\begin{minted}{coq}
Definition pre_extensional (p : A) (q : R) : Prop :=
   forall (r : B), must_extensional p r -> must_extensional q r.

Notation "p ⊑ₑ q" := (pre_extensional p q).
\end{minted}
\end{mdframed}

\subsection{Behavioural characterizations}

\subsubsection{Labeled Transition Systems}

An LTS is a typeclass indexed by the type of states and the type of labels. The
type of labels must be equipped with decidable equality and be countable, as
enforced by the \mintinline{coq}{Label} typeclass.
An action \lstinline|a : Act L| is either an internal action $\tau$ or an external
action: an input or an output of a label in \mintinline{coq}{L}.

\begin{mdframed}
\begin{minted}{coq}
Class Label (L: Type) := {
  label_eqdec: EqDecision L;
  label_countable: Countable L;
}.

Inductive Act (A: Type) := ActExt (μ: ExtAct A) | τ.

Class Lts (A L : Type) `{Label L} := {
    lts_step: A → Act L → A → Prop;
    lts_state_eqdec: EqDecision A;

    lts_step_decidable a α b : Decision (lts_step a α b);

    lts_outputs : A -> gset L;
    lts_outputs_spec1 p1 x p2 :
      lts_step p1 (ActExt (ActOut x)) p2 -> x ∈ lts_outputs p1;
    lts_outputs_spec2 p1 x :
      x ∈ lts_outputs p1 -> {p2 | lts_step p1 (ActExt (ActOut x)) p2};

    lts_stable: A → Act L → Prop;
    lts_stable_decidable p α : Decision (lts_stable p α);
    lts_stable_spec1 p α : ¬ lts_stable p α → { q | lts_step p α q };
    lts_stable_spec2 p α : { q | lts_step p α q } → ¬ lts_stable p α;
  }.

Notation "p ⟶ q"      := (lts_step p τ q).
Notation "p ⟶{ α } q" := (lts_step p α q).
Notation "p ⟶[ α ] q" := (lts_step p (ActExt μ) q).
\end{minted}
\end{mdframed}

An LTS $L$ is cast into an STS by taking only the~$\tau$-transitions, as
formalised by the following instance, which says that \mintinline{coq}{A} can be
equipped with an STS structure when, together with some labels \mintinline{coq}{L},
\mintinline{coq}{A} is equipped with a LTS structure.

\begin{mdframed}
\begin{minted}{coq}
Program Instance sts_of_lts `{Label L} (M: Lts A L): Sts A :=
  {|
    sts_step p q := sts_step p τ q;
    sts_stable s := lts_stable s τ;
  |}.
\end{minted}
\end{mdframed}

\subsubsection{Weak transitions}

Let ${\wt{}} \subseteq \States \times \Actfin \times \States$ denote the least relation such that:
  \begin{description}
  \item[\rname{wt-refl}] $\state \wt{\varepsilon} \stateA$,
  \item[\rname{wt-tau}] $\state \wt{ \trace } \stateB$ if $\state \st{\tau} \stateA$,
    and $\stateA \wt{ \trace } \stateB$
  \item[\rname{wt-mu}]  $\state \wt{\mu.s} \stateB$ if $\state \st{\mu} \stateA$
    and $\stateA \wt{\trace} \stateB$.
  \end{description}

\begin{mdframed}
\begin{minted}{coq}
Definition trace L := list (ExtAct L).

Inductive wt : A -> trace L -> A -> Prop :=
| wt_nil p : wt p [] p
| wt_tau s p q t (l : p ⟶ q) (w : wt q s t) : wt p s t
| wt_act μ s p q t (l : p ⟶[μ] q) (w : wt q s t) : wt p (μ :: s) t.

Notation "p ⟹[s] q" := (wt p s q).
\end{minted}
\end{mdframed}

\subsubsection{Product of LTS}

The characteristic function of the transition relation of the
LTS resulting from the parallel composition of two LTS.
States of the parallel product of $L_1$ and $L_2$ are pairs $(a, b) \in L_1
\times L_2$.
The first two cases correspond to unsynchronized steps from either LTS, and the
third case corresponds to the LTS taking steps with dual actions. The predicate
\lstinline|act_match l1 l2| states that the two actions are visible and are dual
of each other.

\begin{mdframed}
\begin{minted}{coq}
Inductive parallel_step `{M1: Lts A L, M2: Lts B L} :
  A * B → Act L → A * B → Prop :=
| ParLeft l a1 a2 b: a1 -[l]→ a2 → parallel_step (a1, b) l (a2, b)
| ParRight l a b1 b2: b1 -[l]→ b2 → parallel_step (a, b1) l (a, b2)
| ParSync l1 l2 a1 a2 b1 b2:
  act_match l1 l2 → a1 -[l1]→ a2 → b1 -[l2]→ b2 →
  parallel_step (a1, b1) τ (a2, b2)
.
\end{minted}
\end{mdframed}

\begin{figure}
  \hrulefill
\begin{center}
  \scalebox{.9}{%
    \begin{tikzpicture}
      \node (Lts) {\lstinline!Lts!};
      \node (LtsEq) [below =+7pt of Lts]  {\lstinline!LtsEq!};
      \node (LtsOba) [below =+7pt of LtsEq]  {\lstinline!LtsOba!};
      \node (LtsFb) [below left=+7pt and +1pt of LtsOba]  {\lstinline!LtsObaFb!};
      \node (LtsFw) [below right=+7pt and +1pt of LtsOba]  {\lstinline!LtsObaFw!};

      \path[->]
      (Lts) edge (LtsEq)
      (LtsEq) edge (LtsOba)
      (LtsOba) edge (LtsFb)
      (LtsOba) edge (LtsFw);
    \end{tikzpicture}
  }
\end{center}
  \caption{Typeclasses to formalise LTSs.}
  \hrulefill
  \label{fig:structure-typeclasses-lts}
\end{figure}
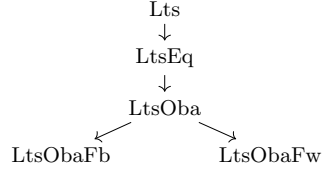

\subsection{Typeclasses for LTS}
Our basic typeclasses are
$\text{\mintinline{coq}{Lts}} \geq \text{\mintinline{coq}{LtsEq}}\geq \text{\mintinline{coq}{LtsOba}}$, where
$\text{\mintinline{coq}{LtsOba}}$ is a super-class of both \mintinline{coq}{LtsObaFB} and \mintinline{coq}{LtsObaFW}.
The class \mintinline{coq}{LtsObaFB} represents output-buffered agents with feedback,
while \mintinline{coq}{LtsObaFW} represents output-buffered agents with forwarding.

\begin{mdframed}
  \begin{minted}{coq}
Class LtsOba (A L : Type) `{Lts A L, !LtsEq A L} :=
    MkOBA {
      lts_oba_output_commutativity {p q r a α} :
        p ⟶[ActOut a] q → q ⟶{α} r →
        ∃ t, p ⟶{α} t ∧ t ⟶≡[ActOut a] r ;
      lts_oba_output_confluence {p q1 q2 a μ} :
        μ ≠ ActOut a → p ⟶[ActOut a] q1 → p ⟶[μ] q2 →
        ∃ r, q1 ⟶[μ] r ∧ q2 ⟶≡[ActOut a] r ;
      lts_oba_output_tau {p q1 q2 a} :
        p ⟶[ActOut a] q1 → p ⟶ q2 →
        (∃ t, q1 ⟶ t ∧ q2 ⟶≡[ActOut a] t) ∨ q1 ⟶≡[ActIn a] q2 ;
      lts_oba_output_deter {p1 p2 p3 a} :
        p1 ⟶[ActOut a] p2 → p1 ⟶[ActOut a] p3 → p2 ≡ p3 ;
      lts_oba_output_deter_inv {p1 p2 q1 q2} a :
        p1 ⟶[ActOut a] q1 → p2 ⟶[ActOut a] q2 → q1 ≡ q2 → p1 ≡ p2;
      (* Multiset of outputs *)
      lts_oba_mo p : gmultiset L;
      lts_oba_mo_spec1 p a : a ∈ lts_oba_mo p <-> a ∈ lts_outputs p;
      lts_oba_mo_spec2 p a q :
        p ⟶[ActOut a] q -> lts_oba_mo p = {[+ a +]} ⊎ lts_oba_mo q;
    }.

Class LtsObaFB (A L: Type) `{LtsOba A L} :=
  MkLtsObaFB {
      lts_oba_fb_feedback {p1 p2 p3 a} :
        p1 ⟶[ActOut a] p2 → p2 ⟶[ActIn a] p3 → p1 ⟶≡ p3
    }.

Class LtsObaFW (A L : Type) `{LtsOba A L} :=
  MkLtsObaFW {
      lts_oba_fw_forward p1 a :
        ∃ p2, p1 ⟶[ActIn a] p2 ∧ p2 ⟶≡[ActOut a] p1;
      lts_oba_fw_feedback {p1 p2 p3 a} :
        p1 ⟶[ActOut a] p2 → p2 ⟶[ActIn a] p3 → p1 ⟶≡ p3 ∨ p1 ≡ p3;
    }.
  \end{minted}
\end{mdframed}

\subsubsection{Termination}

We write $\state \conv$ and say that {\em $\state$ converges} if every
sequence of $\tau$-transitions performed by $\state$ is finite. This is
expressed extensionally by the property that all maximal computations starting
from $p$ contain a \emph{stable} process, meaning that it is finite.

\begin{mdframed}
\begin{minted}{coq}
Definition terminate (p : A) : Prop :=
    forall η : max_exec_from p, exists n fex,
      mex_take_from n η = Some fex /\ lts_stable (fex_from_last fex) τ.
\end{minted}
\end{mdframed}

\subsubsection{Convergence along a trace}

To define the behavioural characterisation of the preorder, we first define
${\cnvalong} \subseteq \States \times \Actfin$ as the least relation such
that,
\begin{description}
\item[\cnvepsilon] $\state \cnvalong \varepsilon$ if $\state \conv$,
\item[\cnvmu] $ \state \cnvalong \mu.\trace $ if $p \conv$ and for each $\stateA$,
$\state \wt{\mu} \stateA \implies \stateA \cnvalong \trace$.
\end{description}

\noindent
This corresponds to the following inductive predicate in Coq:

\begin{mdframed}
\begin{minted}{coq}
Inductive cnv : A -> trace L -> Prop :=
  | cnv_ext_nil p : terminate p -> cnv p []
  | cnv_ext_act p μ s :
    terminate p -> (forall q, p ⟹{μ} q -> cnv q s) -> cnv p (μ :: s).

Notation "p ⇓ s" := (cnv p s).
\end{minted}
\end{mdframed}

\subsection{Forwarders}

We define a mailbox $MO$ as a multiset of names.

\begin{mdframed}
\begin{minted}{coq}
Definition mb (L : Type) `{Label L} := gmultiset L.
\end{minted}
\end{mdframed}

\noindent
\rdef{liftFW} and \rfig{rules-liftFW}\emph{
Lifting of a transition relation to transitions of forwarders.
}

\begin{mdframed}
\begin{minted}{coq}
Inductive lts_fw_step {A L : Type} `{Lts A L} :
  A * mb L -> Act L -> A * mb L -> Prop :=
| lts_fw_p p q m α:
  lts_step p α q -> lts_fw_step (p ▷ m) α (q ▷ m)
| lts_fw_out_mb m p a :
  lts_fw_step (p ▷ {[+ a +]} ⊎ m) (ActExt $ ActOut a) (p ▷ m)
| lts_fw_inp_mb m p a :
  lts_fw_step (p ▷ m) (ActExt $ ActIn a) (p ▷ {[+ a +]} ⊎ m)
| lts_fw_com m p a q :
  lts_step p (ActExt $ ActIn a) q ->
  lts_fw_step (p ▷ {[+ a +]} ⊎ m) τ (q ▷ m).
\end{minted}
\end{mdframed}

\noindent
\rdef{strip-def} and \rdef{fw-eq}\emph{
For any LTS $\genlts$,
two states of $\liftFW{\genlts}$ are equivalent, denoted
$\serverA \triangleright M \doteq \serverB \triangleright N$, if
$ \strip{ \serverA } \simeq \strip{ \serverB }$
and $M \uplus \outputmultiset{\serverA} = N \uplus \outputmultiset{\serverB}$.
}

\begin{mdframed}
\begin{minted}{coq}
Inductive strip `{Lts A L} : A -> gmultiset L -> A -> Prop :=
| strip_nil p : p ⟿{∅} p
| strip_step p1 p2 p3 a m :
  p1 ⟶[ActOut a] p2 -> p2 ⟿{m} p3 -> p1 ⟿{{[+ a +]} ⊎ m} p3

where "p ⟿{ m } q" := (strip p m q).

Definition fw_eq `{LtsOba A L} (p : A * mb L) (q : A * mb L) :=
  forall (p' q' : A),
    p.1 ⟿{lts_oba_mo p.1} p' ->
    q.1 ⟿{lts_oba_mo q.1} q' ->
    p' ⋍ q' /\ lts_oba_mo p.1 ⊎ p.2 = lts_oba_mo q.1 ⊎ q.2.

Infix "≐" := fw_eq (at level 70).
\end{minted}
\end{mdframed}

\noindent
\rlem{harmony-sta}\emph{
For every $\genlts_\StatesA$ and every
$\serverA \triangleright M, \serverB \triangleright N \in \StatesA \times MO$,
and every $\alpha \in L$, if
$
\serverA \triangleright M \mathrel{({\doteq} \cdot {\sta{\alpha}})}
\serverB \triangleright N
$ then
$
\serverA  \triangleright M \mathrel{({\sta{\alpha}} \cdot {\doteq})} \serverB' \triangleright N'.
$
}

\begin{mdframed}
\begin{minted}{coq}
Lemma lts_fw_eq_spec `{LtsObaFB A L} p q t mp mq mt α :
  p ▷ mp ≐ t ▷ mt -> (t ▷ mt) ⟶{α} (q ▷ mq) -> p ▷ mp ⟶≐{α} q ▷ mq.
\end{minted}
\end{mdframed}

\noindent
\rlem{liftFW-works}. \emph{For every LTS~$\genlts \in \obaFB$, $\liftFW{\genlts}
\in \obaFW$.}

\begin{mdframed}
\begin{minted}{coq}
Program Instance LtsMBObaFW `{LtsObaFB A L} : LtsObaFW (A * mb L) L.
\end{minted}
\end{mdframed}

\noindent
\rlem{musti-obafb-iff-musti-obafw}\emph{
For every $\genlts_A, \genlts_B \in \obaFB, \server \in A, \client \in B$,
$\musti{\server}{\client}$ if and only if $\musti{\liftFW{\server}}{\client}$.
}

\begin{mdframed}
\begin{minted}{coq}
Lemma must_iff_must_fw
  {@LtsObaFB A L IL LA LOA V, @LtsObaFB B L IL LB LOB W,
   !FiniteLts A L, !Good B L good }
  (p : A) (e : B) : must p e ↔ must (p, ∅) e.
\end{minted}
\end{mdframed}

\subsection{The Acceptance Set Characterisation}

The behavioural characterisation with acceptance sets
(Definition~\ref{def:accset-leq}) is formalised as follows.
Note that \mintinline{coq}{lts_outputs}, used in the second part of the definition, is
part of the definition of an \mintinline{coq}{Lts}, and produces the finite set of
outputs that a process can immediately produce.

\begin{mdframed}
\begin{minted}{coq}
Definition bhv_pre_cond1 `{Lts A L, Lts B L} (p : A) (q : B) :=
  forall s, p ⇓ s -> q ⇓ s.

Notation "p ≼₁ q" := (bhv_pre_cond1 p q) (at level 70).

Definition bhv_pre_cond2 `{Lts A L, Lts B L} (p : A) (q : B) :=
  forall s q', p ⇓ s -> q ⟹[s] q' -> q' ↛ ->
    ∃ p', p ⟹[s] p' /\ p' ↛ /\ lts_outputs p' ⊆ lts_outputs q'.

Notation "p ≼₂ q" := (bhv_pre_cond2 p q) (at level 70).

Definition bhv_pre `{@Lts A L HL, @Lts B L HL} (p : A) (q : B) :=
  p ≼₁ q /\ p ≼₂ q.

Notation "p ≼ q" := (bhv_pre p q) (at level 70).
\end{minted}
\end{mdframed}

Given an LTS that satisfies the right conditions, \mustequivalence coincides
with the behavioural characterisation above on the LTS of forwarders
(Theorem~\ref{thm:testleqS-equals-bhvleq}).

\begin{mdframed}
\begin{minted}{coq}
Section correctness.
  Context `{LtsObaFB A L, LtsObaFB R L, LtsObaFB B L}.
  Context `{!FiniteLts A L, !FiniteLts B L, !FiniteLts R L, !Good B L}.
  (* The LTS can express the tests required for completeness *)
  Context `{!gen_spec_conv gen_conv, !gen_spec_acc gen_acc}.

  Theorem equivalence_bhv_acc_ctx (p : A) (q : R) :
    p ⊑ₑ q <-> (p, ∅) ≼ (q, ∅).
End correctness.
\end{minted}
\end{mdframed}

\subsection{The Must Set characterisation}

The behavioural characterisation with must sets
(Definition~\ref{def:denicola-char}) is formalised as follows.

\begin{mdframed}
  \begin{minted}{coq}
Definition MUST `{Lts A L} (p : A) (G : gset (ExtAct L)) :=
    forall p', p ⟹ p' -> exists μ p0, μ ∈ G /\ p' ⟹{μ} p0.

Definition MUST__s `{FiniteLts A L} (ps : gset A) (G : gset (ExtAct L)) :=
  forall p, p ∈ ps -> MUST p G.

Definition AFTER `{FiniteLts A L} (p : A) (s : trace L) (hcnv : p ⇓ s) :=
  wt_set p s hcnv.

Definition bhv_pre_ms_cond2
  `{@FiniteLts A L HL LtsA, @FiniteLts B L HL LtsB} (p : A) (q : B) :=
  forall s h1 h2 G, MUST__s (AFTER p s h1) G -> MUST__s (AFTER q s h2) G.

Notation "p ≾₂ q" := (bhv_pre_ms_cond2 p q) (at level 70).

Definition bhv_pre_ms `{@FiniteLts A L HL LtsA, @FiniteLts B L HL LtsB}
  (p : A) (q : B) := p ≼₁ q /\ p ≾₂ q.

Notation "p ≾ q" := (bhv_pre_ms p q).
  \end{minted}
\end{mdframed}

\noindent
\rlem{acceptance-sets-and-must-sets-have-same-expressivity}\emph{
Let $\genlts_A, \genlts_B \in \obaFB$.
For every $\serverA \in \StatesA$  and
$\serverB \in \StatesB $ such that $\liftFW{ \serverA } \bhvleqone \liftFW{ \serverB }$,
we have that
$\liftFW{ \serverA } \msleqtwo \liftFW{ \serverB }$ if and only if
$\liftFW{ \serverA } \asleqAfw \liftFW{ \serverB }$.
}

\begin{mdframed}
\begin{minted}{coq}
Context `{@LtsObaFB A L LL LtsA LtsEqA LtsObaA}.
Context `{@LtsObaFB B L LL LtsR LtsEqR LtsObaR}.

Lemma equivalence_bhv_acc_mst2 (p : A) (q : B) :
  (p, ∅) ≼₁ (q, ∅) -> (p, ∅) ≾₂ (q, ∅) <-> (p, ∅) ≼₂ (q, ∅).
\end{minted}
\end{mdframed}

Given an LTS that satisfies the right conditions, \mustequivalence coincides
with the behavioural characterisation above on the LTS of forwarders
(Theorem~\ref{thm:testleqS-equals-mustsetleq}).

\begin{mdframed}
\begin{minted}{coq}
Section correctness.
  Context `{LtsObaFB A L, LtsObaFB R L, LtsObaFB B L}.
  Context `{!FiniteLts A L, !FiniteLts B L, !FiniteLts R L, !Good B L}.
  (* The LTS can express the tests required for completeness. *)
  Context `{!gen_spec_conv gen_conv, !gen_spec_acc gen_acc}.

  Theorem equivalence_bhv_mst_ctx (p : A) (q : R) :
    p ⊑ₑ q <-> (p, ∅) ≾ (q, ∅).
End correctness.
\end{minted}
\end{mdframed}

\subsection{From extensional to intensional definitions}

\rprop{ext-impl-int}\emph{
Given a countably branching STS~$\sts{\SysStates}{\to}$, and a decidable predicate~$Q$
on~$\SysStates$, for all~$s \in \SysStates$, $\mathsf{ext}_Q(s)$ implies
$\mathsf{int}_Q(s).$
}

\begin{mdframed}
\begin{minted}{coq}
Context `{Hsts: Sts A, @CountableSts A Hsts}.
Context `{@Bar A Hsts}.

Theorem extensional_implies_intensional x:
  extensional_pred x -> intensional_pred x.
\end{minted}
\end{mdframed}

\rcor{ext-int-eq-conv}\emph{
For every $\server \in \States$,
\begin{enumerate}
\item $\state \conv $ if and only if $\state \convi$,
\item for every $\client$ we have that $\Must{\server}{\client}$ if
  and only if $\musti{\server}{\client}$.
\end{enumerate}
}

\begin{mdframed}
\begin{minted}{coq}
Context `{Label L}.
Context `{!Lts A L, !FiniteLts A L}.

Lemma terminate_extensional_iff_terminate (p : A) :
  terminate_extensional p <-> terminate p.

Inductive must_sts `{Sts (A * B), good : B -> Prop} (p : A) (e : B) :
  Prop :=
| m_sts_now : good e -> must_sts p e
| m_sts_step
    (nh : ¬ good e)
    (nst : ¬ sts_stable (p, e))
    (l : forall p' e', sts_step (p, e) (p', e') -> must_sts p' e')
  : must_sts p e
.

Lemma must_extensional_iff_must_sts
  `{good : B -> Prop, good_decidable : forall (e : B), Decision (good e)}
  `{Lts A L, !Lts B L, !LtsEq B L, !Good B L good,
    !FiniteLts A L, !FiniteLts B L} (p : A) (e : B) :
  must_extensional p e <-> must_sts p e.
\end{minted}
\end{mdframed}

Equivalence between the inductive definitions of $\opMust$ defined using Sts and $\opMust$ defined using Lts.

\begin{mdframed}
\begin{minted}{coq}
Inductive must `{Lts A L, !Lts B L, !LtsEq B L, !Good B L good}
  (p : A) (e : B) : Prop :=
| m_now : good e -> must p e
| m_step
    (nh : ¬ good e)
    (ex : ∃ t, parallel_step (p, e) τ t)
    (pt : forall p', p ⟶ p' -> must p' e)
    (et : forall e', e ⟶ e' -> must p e')
    (com : forall p' e' μ, e ⟶[μ] e' -> p ⟶[co μ] p' -> must p' e')
  : must p e
.

Lemma must_sts_iff_must `{Lts A L, !Lts B L, !LtsEq B L, !Good B L good}
  (p : A) (e : B) : must_sts p e <-> must p e.
\end{minted}
\end{mdframed}

\subsection{Completeness}

Properties of the functions that generate clients (\rtab{properties-functions-to-generate-clients}).

\begin{mdframed}
\begin{minted}{coq}
Class gen_spec {A L : Type} `{Lts A L, !LtsEq A L, !Good A L good}
  (gen : list (ExtAct L) -> A) := {
    gen_spec_ungood : forall s, ¬ good (gen s) ;
    gen_spec_mu_lts_co μ s : gen (μ :: s) ⟶⋍[co μ] gen s;
    gen_spec_out_lts_tau_ex a s : ∃ e', gen (ActOut a :: s) ⟶ e';
    gen_spec_out_lts_tau_good a s e : gen (ActOut a :: s) ⟶ e -> good e;
    gen_spec_out_lts_mu_uniq {e a μ s} :
    gen (ActOut a :: s) ⟶[μ] e -> e = gen s /\ μ = ActIn a;
  }.

Class gen_spec_conv {A L : Type} `{Lts A L, ! LtsEq A L, !Good A L good}
  (gen_conv : list (ExtAct L) -> A) := {
    gen_conv_spec_gen_spec : gen_spec gen_conv ;
    gen_spec_conv_nil_stable_mu μ : gen_conv [] ↛[μ] ;
    gen_spec_conv_nil_lts_tau_ex : ∃ e', gen_conv [] ⟶ e';
    gen_spec_conv_nil_lts_tau_good e : gen_conv [] ⟶ e -> good e;
  }.

Class gen_spec_acc {A : Type} `{Lts A L, ! LtsEq A L, !Good A L good}
  (gen_acc : gset L -> list (ExtAct L) -> A) := {
    gen_acc_spec_gen_spec O : gen_spec (gen_acc O);
    gen_spec_acc_nil_stable_tau O : gen_acc O [] ↛;
    gen_spec_acc_nil_stable_out O a : gen_acc O [] ↛[ActOut a];
    gen_spec_acc_nil_mu_inv O a e : gen_acc O [] ⟶[ActIn a] e -> a ∈ O;
    gen_spec_acc_nil_mem_lts_inp O a :
      a ∈ O -> ∃ r, gen_acc O [] ⟶[ActIn a] r;
    gen_spec_acc_nil_lts_inp_good μ e' O :
      gen_acc O [] ⟶[μ] e' -> good e';
  }.
\end{minted}
\end{mdframed}

\noindent
\rprop{must-iff-acnv}\emph{
For every $\genlts_{\States} \in \obaFW$,
$\server \in \States$, and
$\trace \in \Actfin$ we have that $\musti{\server}{ \testconv{ \trace} }$
if and only if~$\server \cnvalong \trace$.
}

\begin{mdframed}
\begin{minted}{coq}
Lemma must_iff_cnv
  `{@LtsObaFW A L IL LA LOA V, @LtsObaFB B L IL LB LOB W,
    !Good B L good, !gen_spec_conv gen_conv} (p : A) s :
   must p (gen_conv s) <-> p ⇓ s.
Proof. split; [eapply cnv_if_must | eapply must_if_cnv]; eauto. Qed.
\end{minted}
\end{mdframed}

\noindent
\rlem{must-output-swap-l-fw}\emph{
  Let $\genlts_A \in \obaFW$ and
  $\genlts_B \in \obaFB$.
  $\Forevery \serverA_1, \serverA_2 \in \StatesA$,
  every $\client_1, \client_2 \in \StatesB$ and name $\aa \in \Names$ such that
  $\serverA_1 \st{\co{\aa}} \serverA_2$ and
  $\client_1 \st{\co{\aa}} \client_2$,
  if $\musti{\serverA_1}{\client_2}$ then $\musti{\serverA_2}{\client_1}$.
}

\begin{mdframed}
\begin{minted}{coq}
Lemma must_output_swap_l_fw
  `{@LtsObaFW A L IL LA LOA V, @LtsObaFB B L IL LB LOB W, !Good B L good}
  (p1 p2 : A) (e1 e2 : B) (a : L) :
  p1 ⟶[ActOut a] p2 -> e1 ⟶[ActOut a] e2 -> must p1 e2 -> must p2 e1.
\end{minted}
\end{mdframed}

\noindent
\rlem{completeness-part-2.2-diff-outputs}.\emph{
Let $\genlts_A \in \obaFW$.
For every $\server \in \States$, $\trace \in \Actfin$,
and every $L, E \subseteq \Names$, if
$\co{L} \in \accht{ \server }{ \trace }$
then $\Nmusti{ \server }{ \testacc{\trace}{E \setminus L}}$.
}

\begin{mdframed}
\begin{minted}{coq}
Lemma not_must_gen_a_without_required_output
  `{@LtsObaFW A L IL LA LOA V, @LtsObaFB B L IL LB LOB W,
    !Good B L good, !gen_spec_acc gen_acc} (q q' : A) s O :
  q ⟹[s] q' -> q' ↛ -> ¬ must q (gen_acc (O ∖ lts_outputs q') s).
\end{minted}
\end{mdframed}

\noindent
\rlem{completeness-part-2.2-auxiliary}\emph{
Let $\genlts_A \in \obaFW$.
$\Forevery \server \in \States, \trace \in \Actfin$,
and every finite set $\ohmy \subseteq \co{\Names}$,
if $\server \cnvalong s$ then either
\begin{enumerate}[(i)]
\item
  $\musti{\server}{\testacc{ \trace }{ \bigcup \co{ \accht{p}{s}
      \setminus \ohmy }}}$, or
\item
  there exists $\widehat{\ohmy} \in \accht{ \server }{ \trace }$ such that $\widehat{\ohmy} \subseteq \ohmy$.
\end{enumerate}
}

\begin{mdframed}
\begin{minted}{coq}
Lemma must_gen_a_with_s
  `{@LtsObaFW A L IL LA LOA V, @LtsObaFB B L IL LB LOB W,
    !FiniteLts A L, !Good B L good, !gen_spec_acc gen_acc}
  s (p : A) (hcnv : p ⇓ s) O :
  (exists p', p ⟹[s] p' /\ lts_stable p' τ /\ lts_outputs p' ⊆ O)
    \/ must p (gen_acc (oas p s hcnv ∖ O) s).
\end{minted}
\end{mdframed}

\noindent
\rlem{completeness}.\emph{
For every $\genlts_A, \genlts_B \in \obaFW$ and
servers $\serverA \in \StatesA, \serverB \in \StatesB $,
if ${ \serverA } \testleqS { \serverB }$
then ${ \serverA } \asleq { \serverB }$.
}

\begin{mdframed}
\begin{minted}{coq}
Lemma completeness_fw
  `{@LtsObaFW A L IL LA LOA V, @LtsObaFB B L IL LB LOB W,
    @LtsObaFW C L IL LC LOC VC, !FiniteLts A L, !FiniteLts C L,
    !FiniteLts B L, !Good B L good,
    !gen_spec_conv gen_conv, !gen_spec_acc gen_acc}
  (p : A) (q : C) : p ⊑ q -> p ≼ q.
\end{minted}
\end{mdframed}

\noindent
\rprop{bhv-completeness}.
For every $\genlts_A, \genlts_B \in \obaFB$ and
servers $\serverA \in \StatesA, \serverB \in \StatesB $,
if $\serverA \testleqS \serverB$ then $\liftFW{ \serverA } \asleq \liftFW{ \serverB }$.

\begin{mdframed}
\begin{minted}{coq}
Lemma completeness
  `{@LtsObaFB A L IL LA LOA V, @LtsObaFB B L IL LB LOB W,
    @LtsObaFB C L IL LC LOC VC,
    !FiniteLts A L, !FiniteLts B L, !FiniteLts C L, !Good C L good,
    !gen_spec_conv gen_conv, !gen_spec_acc gen_acc}
  (p : A) (q : B) : p ⊑ q -> p ▷ ∅ ≼ q ▷ ∅.
\end{minted}
\end{mdframed}

\subsection{Soundness}

\rfig{rules-mustset-main}. Rules to define inductively the predicate $\opMustset$.

\begin{mdframed}
\begin{minted}{coq}
Inductive mustx
  `{Lts A L, !FiniteLts A L, !Lts B L, !LtsEq B L, !Good B L good}
  (ps : gset A) (e : B) : Prop :=
| mx_now (hh : good e) : mustx ps e
| mx_step
    (nh : ¬ good e)
    (ex : forall (p : A), p ∈ ps -> ∃ t, parallel_step (p, e) τ t)
    (pt : forall ps',
        lts_tau_set_from_pset_spec1 ps ps' -> ps' ≠ ∅ ->
        mustx ps' e)
    (et : forall (e' : B), e ⟶ e' -> mustx ps e')
    (com : forall (e' : B) μ (ps' : gset A),
        lts_step e (ActExt μ) e' ->
        wt_set_from_pset_spec1 ps [co μ] ps' -> ps' ≠ ∅ ->
        mustx ps' e')
  : mustx ps e.
\end{minted}
\end{mdframed}

\noindent
\rlem{musti-if-mustset-helper}\emph{
For every LTS $\genlts_A, \genlts_B$ and every
$X \in \pparts{\StatesA}$, we have that
$\mustset{X}{\client}$ if and only if for every $\serverA \in X
\wehavethat \musti{\serverA}{\client}$.
}

\begin{mdframed}
\begin{minted}{coq}
Lemma must_set_iff_must_for_all
  `{Lts A L, !FiniteLts A L, !Lts B L, !LtsEq B L, !Good B L good}
  (X : gset A) (e : B) : X ≠ ∅ ->
    (forall p, p ∈ X -> must p e) <-> mustx X e.
\end{minted}
\end{mdframed}

Lifting of the predicates~$\bhvleqone$ and~$\bhvleqtwo$ to sets of servers.

\begin{mdframed}
\begin{minted}{coq}
Definition bhv_pre_cond1__x `{FiniteLts P L, FiniteLts Q L}
 (ps : gset P) (q : Q) := forall s, (forall p, p ∈ ps -> p ⇓ s) -> q ⇓ s.

Notation "ps ≼ₓ1 q" := (bhv_pre_cond1__x ps q) (at level 70).

Definition bhv_pre_cond2__x
  `{@FiniteLts P L HL LtsP, @FiniteLts Q L HL LtsQ}
  (ps : gset P) (q : Q) :=
  forall s q', q ⟹[s] q' -> q' ↛ ->
    (forall p, p ∈ ps -> p ⇓ s) ->
    exists p, p ∈ ps /\ exists p',
      p ⟹[s] p' /\ p' ↛ /\ lts_outputs p' ⊆ lts_outputs q'.

Notation "ps ≼ₓ2 q" := (bhv_pre_cond2__x ps q) (at level 70).

Notation "ps ≼ₓ q" := (bhv_pre_cond1__x ps q /\ bhv_pre_cond2__x ps q)
   (at level 70).
\end{minted}
\end{mdframed}

\noindent
\rlem{alt-set-singleton-iff}.\emph{
For every LTS $\genlts_A, \genlts_B$ and servers $\server \in \StatesA$,
$\serverB \in \StatesB$,
$\serverA \asleq \serverB$ if and only if $\set{\serverA} \asleqset \serverB$.
}

\begin{mdframed}
\begin{minted}{coq}
Lemma alt_set_singleton_iff
  `{@FiniteLts P L HL LtsP, @FiniteLts Q L HL LtsQ}
  (p : P) (q : Q) : p ≼ q <-> {[ p ]} ≼ₓ q.
\end{minted}
\end{mdframed}

\noindent
\rlem{bhvleqone-preserved}.\emph{
Let $\genlts_\StatesA, \genlts_\StatesB \in \obaFW$.
For every set $X \in \pparts{ \StatesA }$, and
$\serverB \in \StatesB$, such that
$X \cnvleqset \serverB$ then
\begin{enumerate}
\item
  $\serverB \st{ \tau } \serverB'$ implies $X \cnvleqset \serverB'$,
\item
  $X \convi$, $X \wt{\mu} X'$ and $\serverB \st{\mu} \serverB'$ imply $X' \cnvleqset \serverB'$.
\end{enumerate}
}

\noindent
\rlem{bhvleqtwo-preserved}\emph{
Let $\genlts_\StatesA, \genlts_\StatesB \in \obaFW$.
For every
$X, X' \in \pparts{ \StatesA }$ and $ \serverB \in \StatesB$,
such that $X \accleqset \serverB$, then
\begin{enumerate}
\item
  $\serverB \st{ \tau } \serverB'$ implies $X \accleqset \serverB'$,
\item
  for every $\mu \in \Act$,
  if $X \convi$, then for every  $\serverB \st{\mu} \serverB'$ and set $X \wt{\mu} X'$
  we have $X' \accleqset \serverB'$.
\end{enumerate}
}

\begin{mdframed}
\begin{minted}{coq}
Lemma bhvx_preserved_by_tau
  `{@FiniteLts P L HL LtsP, @FiniteLts Q L HL LtsQ}
  (ps : gset P) (q q' : Q) : q ⟶ q' -> ps ≼ₓ q -> ps ≼ₓ q'.

Lemma bhvx_preserved_by_mu
  `{@FiniteLts P L HL LtsP, @FiniteLts Q L HL LtsQ}
  (ps0 : gset P) (q : Q) μ ps1 q'
  (htp : forall p, p ∈ ps0 -> terminate p) :
  q ⟶[μ] q' -> wt_set_from_pset_spec ps0 [μ] ps1 ->
  ps0 ≼ₓ q -> ps1 ≼ₓ q'.
\end{minted}
\end{mdframed}

\noindent
Lemma \emph{
Let $\genlts_\StatesA, \genlts_\StatesB \in \obaFW$ and $\genlts_\StatesC \in \obaFB$.
For every $X \in \pparts{ \StatesA }$ and
$\serverB \in \StatesB $ such that
$X \asleqset \serverB$, for every $\client \in \StatesC$
if $\lnot \good{\client}$ and $\mustset{X}{\client}$
then $\csys{\serverB}{\client} \st{\tau}$.
}

\begin{mdframed}
\begin{minted}{coq}
Lemma stability_nbhvleqtwo
  `{@LtsObaFW P L Lbl LtsP LtsEqP LtsObaP,
    @LtsObaFW Q L Lbl LtsQ LtsEqQ LtsObaQ,
    !FiniteLts P L, !FiniteLts Q L, !Lts B L, !LtsEq B L, !Good B L good}
  (X : gset P) (q : Q) e :
  ¬ good e -> mustx X e -> X ≼ₓ2 q -> exists t, (q, e) ⟶{τ} t.
\end{minted}
\end{mdframed}

\noindent
Lemma\emph{
Let $\genlts_\StatesA, \genlts_\StatesB \in \obaFW$.
For every $X \in \pparts{ \StatesA }$ and
$\serverB, \serverB' \in \StatesB$,
such that $X \asleqset \serverB$, then
for every $\mu \in \Act$, if $X \cnvalong \mu$ and $\serverB \st{\mu} \serverB'$ then $X \wt{\mu}$.
}
\begin{mdframed}
\begin{minted}{coq}
Lemma bhvx_mu_ex `{@FiniteLts P L HL LtsP, @FiniteLts Q L HL LtsQ}
  (ps : gset P) (q q' : Q) μ
  : ps ≼ₓ q -> (forall p, p ∈ ps -> p ⇓ [μ]) ->
    q ⟶[μ] q' -> exists p', wt_set_from_pset_spec1 ps [μ] {[ p' ]}.
\end{minted}
\end{mdframed}

\noindent
Lemma
\emph{
For every $\genlts_\StatesA \in \obaFW$, $\genlts_\StatesB \in \obaFB$,
every set of processes $X \in \pparts{ \StatesA }$, every $\client \in \StatesB$, and every $\mu \in \Act$,
if $\mustset{X}{\client}$, $\lnot \good{\client}$ and $\client \st{\mu}$ then $X \cnvalong \co{\mu}$.
}

\begin{mdframed}
\begin{minted}{coq}
Lemma ungood_acnv_mu `{LtsOba A L, !FiniteLts A L, !Lts B L, !LtsEq B L,
  !Good B L good} ps e e' : mustx ps e ->
  forall μ p, p ∈ ps -> e ⟶[co μ] e' -> ¬ good e -> p ⇓ [μ].
\end{minted}
\end{mdframed}

\noindent
\rlem{soundness-set}.
\emph{
Let $\genlts_\StatesA, \genlts_\StatesB \in \obaFW$ and
$\genlts_\StatesC \in \obaFB$.
For every set of processes $X \in \pparts{ \StatesA }$,
server $\serverB \in \StatesB$ and client $\client \in \StatesC$,
if $\mustset{X}{\client}$ and $X \asleqset \serverB$ then $\musti{\serverB}{\client}$.
}

\begin{mdframed}
\begin{minted}{coq}
Lemma soundnessx `{
    @LtsObaFW A L Lbl LtsA LtsEqA LtsObaA,
    @LtsObaFW C L Lbl LtsC LtsEqC LtsObaC,
    @LtsObaFB B L Lbl LtsB LtsEqB LtsObaB,
    !FiniteLts A L, !FiniteLts C L, !FiniteLts B L, !Good B L good}
    (ps : gset A) (e : B) : mustx ps e ->
      forall (q : C), ps ≼ₓ q -> must q e.
\end{minted}
\end{mdframed}

\noindent
\rprop{bhv-soundness}.
\emph{
For every $\genlts_A, \genlts_B \in \obaFB$ and
servers $\serverA \in \StatesA, \serverB \in \StatesB $,
if $\liftFW{ \serverA } \asleq \liftFW{ \serverB }$ then $\serverA \testleqS \serverB$.
}

\begin{mdframed}
\begin{minted}{coq}
Lemma soundness
  `{@LtsObaFB A L IL LA LOA V, @LtsObaFB C L IL LC LOC T,
    @LtsObaFB B L IL LB LOB W,
    !FiniteLts A L, !FiniteLts C L, !FiniteLts B L, !Good B L good }
  (p : A) (q : C) : p ▷ ∅ ≼ q ▷ ∅ -> p ⊑ q.
\end{minted}
\end{mdframed}

\begin{corollary}
  Let $\genlts_A, \genlts_B \in \obaFB$.
  For every $\serverA \in \States$  and
  $\serverB \in \StatesB $, we have that
  $\serverA \testleqS \serverB$ if and only if
  $\serverA \failleq \serverB$.
\end{corollary}

\begin{mdframed}
\begin{minted}{coq}

Section failure.

  Definition Failure `{FiniteLts A L} (p : A)
    (s : trace L) (G : gset (ExtAct L)) :=
    p ⇓ s -> exists p', p ⟹[s] p' /\
      forall μ, μ ∈ G -> ¬ exists p0, p' ⟹{μ} p0.

  Definition fail_pre_ms_cond2
    `{@FiniteLts A L HL LtsA, @FiniteLts B L HL LtsB}
    (p : A) (q : B) := forall s G, Failure q s G -> Failure p s G.

  Definition fail_pre_ms
    `{@FiniteLts A L HL LtsA, @FiniteLts B L HL LtsB} (p : A) (q : B) :=
    p ≼₁ q /\ fail_pre_ms_cond2 p q.

  Context `{LL : Label L}.
  Context `{LtsA : !Lts A L, !FiniteLts A L}.
  Context `{LtsR : !Lts R L, !FiniteLts R L}.

  Context `{@LtsObaFB A L LL LtsA LtsEqA LtsObaA}.
  Context `{@LtsObaFB R L LL LtsR LtsEqR LtsObaR}.

  Theorem equivalence_pre_failure_must_set (p : A) (q : R) :
   (p ▷ ∅) ≾ (q ▷ ∅) <-> (p ▷ ∅) ⋖ (q ▷ ∅).

End failure.
\end{minted}
\end{mdframed}

\subsection{Coinductive definition}

The coinductive preorder is defined in Coq using a coinductive predicate as
follows:

\begin{mdframed}
  \begin{minted}{coq}
CoInductive copre `{@FiniteLts A L HL LtsP, @FiniteLts B L HL LtsQ}
    (ps : gset A) (q : B) : Prop := {
    c_tau q' : q ⟶ q' -> copre ps q'
  ; c_now : (forall p, p ∈ ps -> p ⤓) -> q ↛ ->
            exists p p', p ∈ ps /\ p ⟹ p' /\ p' ↛ /\ lts_outputs p' ⊆ lts_outputs q
  ; c_step : forall μ q' ps', (forall p, p ∈ ps -> p ⇓ [μ]) ->
                         q ⟶[μ] q' -> wt_set_from_pset_spec ps [μ] ps' -> copre ps' q'
  ; c_cnv : (forall p, p ∈ ps -> p ⤓) -> q ⤓
  }.
  \end{minted}
\end{mdframed}

The soundness and completeness of the preorder is expressed in
\rthm{coinductive-char-equiv-main}, which is formalised as follows in our
development.

\begin{mdframed}
  \begin{minted}{coq}
Theorem eqx `{@FiniteLts A L HL LtsP, @FiniteLts B L HL LtsQ} (X : gset A) (q : B) :
  X ≼ₓ q <-> X ⩽ q.
  \end{minted}
\end{mdframed}

\subsection{The preorder on traces in normal forms}

A trace in normal form is simply a list of pairs of multisets, representing
inputs and outputs:

\begin{mdframed}
  \begin{minted}{coq}
Definition ntrace L `{Label L} : Type := list (gmultiset L * gmultiset L).
  \end{minted}
\end{mdframed}

Given such a trace in normal form, we can linearize it to a trace by choosing an
abritrary order for the elements of the multisets. This is achieved using the
\mintinline{coq}{elements} that maps a multiset to a list that contains its
elements in an arbitrary order.

\begin{mdframed}
  \begin{minted}{coq}
Fixpoint linearize `{Label L} (nt : ntrace L) : trace L :=
  match nt with
  | [] => []
  | (mi, mo) :: nt' =>
      let inputs := map ActIn (elements mi) in
      let outputs := map ActOut (elements mo) in
      inputs ++ outputs ++ linearize nt'
  end.
  \end{minted}
\end{mdframed}

Using this operation, we can define the preorder as follows:

\begin{mdframed}
  \begin{minted}{coq}
Definition bhv_lin_pre_cond1 `{Lts P L, Lts Q L} (p : P) (q : Q) :=
    forall s, p ⇓ linearize s -> q ⇓ linearize s.

Notation "p ⪷₁ q" := (bhv_lin_pre_cond1 p q) (at level 70).

Definition bhv_lin_pre_cond2 `{@Lts P L HL, @Lts Q L HL} (p : P) (q : Q) :=
  forall nt q',
    p ⇓ linearize nt -> q ⟹[linearize nt] q' -> q' ↛ ->
    ∃ p', p ⟹[linearize nt] p' /\ p' ↛ /\ lts_outputs p' ⊆ lts_outputs q'.

Notation "p ⪷₂ q" := (bhv_lin_pre_cond2 p q) (at level 70).

Definition bhv_lin_pre `{@Lts P L HL, @Lts Q L HL} (p : P) (q : Q) := p ⪷₁ q /\ p ⪷₂ q.

Notation "p ⪷ q" := (bhv_lin_pre p q) (at level 70).
  \end{minted}
\end{mdframed}

Finally, the theorem stating that this preorder characterises the \mustpreorder
is formulated as follows in our development:

\begin{mdframed}
  \begin{minted}{coq}
Lemma asyn_iff_bhv
 `{@LtsObaFW P L IL LA LOA V, @LtsObaFW Q L IL LB LOB W, 
    !FiniteLts Q L, !FiniteLts Q L} : 
  forall (p : P) (q : Q), p ⪷ q <-> p ≼ q.
  \end{minted}
\end{mdframed}

\clearpage

\section{Mapping of results from the paper to the Coq code}

\scalebox{.8}{%
  \begin{tabular}{| l | l | l |}
    \hline
    Paper & Coq File & Coq name \\
    \hline\hline
    \rfig{axioms} & TransitionSystems.v & Class LtsOba\\
    \hline
    \rfig{mechanisation-lts} & TransitionSystems.v & Class Sts, ExtAct, Act, Label, Lts\\
    \hline
    \rdef{must-extensional} & Equivalence.v & must\_extensional\\
    \hline
    \rdef{testleqS} & Equivalence.v & pre\_extensional\\
    \hline
    \req{syntax-processes} & ACCSInstance.v & proc\\
    \hline
    \rfig{Axiom-LtsEq} & TransitionSystems.v & LtsEq\\
    \hline
    \rdef{must-extensional} & MustEx.v & must\_extensional\\
    \hline
    \rdef{inf-transition-sequence} & TransitionSystems.v & max\_exec\_from\\
    \hline
    $\state \wt{ \trace } \stateA $ & TransitionSystems.v & wt\\
    \hline
    $\state \conv$ & Equivalence.v & terminate\_extensional\\
    \hline
    $\state \cnvalong \trace $ & TransitionSystems.v & cnv\\
    \hline
    \rlem{cnvalong-iff-prefix} & TransitionSystems.v & cnv\_iff\_prefix\_terminate\\
    \hline
    \rlem{st-wtout-st} & TransitionSystems.v & stable\_tau\_preserved\_by\_wt\_output, stable\_tau\_input\_preserved\_by\_wt\_output\\
    \hline
    \rlem{st-wtout-Nok} & Must.v & ungood\_preserved\_by\_wt\_output\\
    \hline
    \req{axioms-forwarders} & TransitionSystems.v & Class LtsObaFW \\
    \hline
    \rdef{accset-leq} & Must.v & bhv\_pre\\
    \hline
    \rfig{rules-liftFW} & TransitionSystems.v & lts\_fw\_step\\
    \hline
    \rdef{liftFW} & TransitionSystems.v & MbLts\\
    \hline
    \rdef{strip-def} & TransitionSystems.v & strip\\
    \hline
    \rdef{fw-eq} & TransitionSystems.v & fw\_eq\\
    \hline
    \rlem{harmony-sta} & TransitionSystems.v & lts\_fw\_eq\_spec\\
    \hline
    \rlem{liftFW-works} & TransitionSystems.v & Instance LtsMBObaFW\\
    \hline
    \rlem{musti-obafb-iff-musti-obafw} & Lift.v & must\_iff\_must\_fw\\
    \hline
    \rlem{musti-obafb-iff-musti-obafw}
    & Lift.v & lift\_fw\_ctx\_pre\\
    \hline
    \rthm{testleqS-equals-bhvleq} & Equivalence.v &  equivalence\_bhv\_acc\_ctx\\
    \hline
    \rdef{denicola-char} & Must.v & bhv\_pre\_ms\\
    \hline
    \rlem{acceptance-sets-and-must-sets-have-same-expressivity} & Must.v & equivalence\_bhv\_acc\_mst\\
    \hline
    \rthm{testleqS-equals-mustsetleq} & Must.v & equivalence\_bhv\_mst\_ctx\\
    \hline
    \rlem{ACCS-obaFB} & ACCSInstance.v & ACCS\_ltsObaFB\\
    \hline
    \rcor{characterisation-for-aCCS} & ACCSInstance.v & bhv\_iff\_ctx\_ACCS\\
    \hline
    \rprop{ext-impl-int} & Bar.v & extensional\_implies\_intensional\\
    \hline
    $\state \convi$ & TransitionSystems.v & terminate\\
    \hline
    $\musti{\serverA }{\serverB}$ & Must.v & must\_sts\\
    \hline
    \rcor{ext-int-eq-conv} & Equivalence.v & terminate\_extensional\_iff\_terminate\\
    \hline
    \rtab{properties-functions-to-generate-clients} & Completeness.v & Class gen\_spec, gen\_spec\_conv, gen\_spec\_acc \\
    \hline
    \rprop{must-iff-acnv} & Completeness.v & must\_iff\_cnv \\
    \hline
    \rlem{must-output-swap-l-fw} & Lift.v & must\_output\_swap\_l\_fw \\
    \hline
    \rlem{completeness-part-2.2-diff-outputs} & Completeness.v & not\_must\_gen\_a\_without\_required\_output \\
    \hline
    \rlem{completeness-part-2.2-auxiliary} & Completeness.v & must\_gen\_a\_with\_s \\
    \hline
    \rlem{completeness} & Completeness.v & completeness\_fw\\
    \hline
    \rprop{bhv-completeness} & Completeness.v & completeness\\
    \hline
    \rlem{musti-if-mustset-helper} & Soundness.v & must\_set\_iff\_must\_for\_all\\
    \hline
    \rfig{rules-mustset-main} & Soundness.v & mustx\\
    \hline
    $X \cnvleqset \serverB$ and $X \accleqset \serverB$ & Soundness.v & bhv\_pre\_cond1$_x$ and bhv\_pre\_cond2$_x$ \\
    \hline
    \rlem{alt-set-singleton-iff} & Soundness.v & must\_set\_iff\_must\\
    \hline
    \rlem{bhvleqone-preserved}, \rlem{bhvleqtwo-preserved} & Soundness.v & bhvx\_preserved\_by\_tau, bhvx\_preserved\_by\_mu\\
    \hline
    \rlem{soundness-set} & Soundness.v & soundnessx \\
    \hline
    \rprop{bhv-soundness} & Soundndess.v & soundness \\
    \hline
    \rcor{asynleq-equals-bhvleq} & Normalisation.v & asyn\_iff\_bhv \\
    \hline
    \rthm{coinductive-char-equiv-main} & Coin.v & eq\_ctx \\
    \hline
  \end{tabular}
  }

\end{document}